\documentclass[11pt,a4paper]{article}
\usepackage[utf8]{inputenc}
\usepackage[T1]{fontenc}
\usepackage{lmodern}
\usepackage[margin=1in]{geometry}
\usepackage{amsmath,amssymb,amsthm,bm}
\usepackage{graphicx,float,subcaption}
\usepackage{microtype}
\usepackage[numbers,sort&compress]{natbib}
\usepackage{xurl}
\usepackage[hidelinks,unicode]{hyperref}
\usepackage{bookmark}
\newcommand{\bx}{\boldsymbol{x}}

\newtheorem{remark}{Remark}

\newtheorem{proposition}{Proposition}[section]

\date{}
\title{Multiscale modeling of host–pathogen interactions and mucociliary clearance during non-tuberculous mycobacterial pulmonary infection}
\author{\parbox{0.97\textwidth}{\centering
\normalsize Jindong Wang$^{1,\dagger}$,
Kali Konstantinopoulos$^{2,3,\dagger}$,
Po-Chun Kuo$^{4}$,
Mingchao Cai$^{5}$,
Ning Wei$^{4}$,
Elsje Pienaar$^{2,6,7}$ and
Wenrui Hao$^{8}$\\[8pt]
{\footnotesize
$^{1}$Mathematical Institute, University of Oxford, Oxford, UK\\
$^{2}$Weldon School of Biomedical Engineering, Purdue University, West Lafayette, IN, USA \\
$^{3}$Indiana University School of Medicine, Indianapolis, IN, USA \\
$^{4}$Department of Mathematics, Purdue University,
West Lafayette, IN, USA\\
$^{5}$Department of Mathematics, Morgan State University,
Baltimore, MD, USA\\
$^{6}$Regenstrief Center for Healthcare Engineering, Purdue University, West Lafayette, IN, USA\\
$^{7}$Stellenbosch Institute for Advanced Study (STIAS), Stellenbosch, South Africa\\
$^{8}$Department of Mathematics, The Pennsylvania State University,
University Park, PA, USA\\[2pt]
$^{\dagger}$These authors contributed equally to this work.\\[5pt]
Corresponding author: Wenrui Hao (\href{mailto:wxh64@psu.edu}{\nolinkurl{wxh64@psu.edu}}).
}}}
\hypersetup{pdftitle={Multiscale modeling of host–pathogen interactions and mucociliary clearance during non-tuberculous mycobacterial pulmonary infection},pdfauthor={Jindong Wang; Kali Konstantinopoulos; Po-Chun Kuo; Mingchao Cai; Ning Wei; Elsje Pienaar; Wenrui Hao}}
\begin{document}
\maketitle
\begin{abstract}
Non-tuberculous mycobacterial (NTM) infections are a clinical challenge in cystic fibrosis (CF), where impaired mucociliary clearance and altered mucus rheology promote bacterial colonization despite host immune responses. Understanding how bacterial growth, immune cell dynamics, and mucus transport regulate infection progression is difficult because these processes interact across spatial and temporal scales. We develop a computational framework bridging a mechanistic agent-based model (ABM) of NTM infection with a spatially resolved partial differential equation (PDE) model. The PDE model couples bacterial proliferation, macrophage chemotaxis, immune-mediated clearance, mucus degradation, and viscoelastic transport in a two-compartment geometry representing mucus and lung tissue. Parameters are calibrated using data from the established ABM, yielding an efficient continuum representation while preserving cellular mechanisms.

The PDE model reproduces bacterial and macrophage dynamics and enables analyses of mucus-related mechanisms and therapies. Sensitivity analysis identifies mucus viscosity, bacterial diffusivity, and macrophage mobility as key regulators of bacterial persistence through mucociliary clearance and tissue colonization. Simulations reveal nonlinear effects of mucolytic therapies: enhanced clearance reduces bacterial burden in mucus, whereas excessive viscosity reduction may promote migration into lung tissue, supporting combination with antibacterial treatment. This framework provides a quantitative platform for studying pulmonary infections, evaluating therapies, and developing patient-specific digital twins.

\end{abstract}


\section{Introduction}

Non-tuberculous mycobacterial (NTM) infections are an increasingly important cause of morbidity in individuals with cystic fibrosis (CF), where chronic airway disease and impaired host defenses promote persistent bacterial colonization and infection \cite{Olivier2003-xw}. Although advances in antimicrobial therapies have improved patient outcomes, pulmonary NTM infections remain difficult to eradicate because bacterial persistence is influenced not only by microbial replication but also by the complex interactions between immune responses and the airway microenvironment. In particular, effective host defense depends on both biological mechanisms, including innate and adaptive immune responses, and physical mechanisms such as mucociliary clearance, which continuously removes pathogens from the airway surface. Understanding how immune responses and mucus clearance interact to determine infection outcomes is therefore essential for developing more effective therapeutic strategies.

A defining characteristic of pulmonary NTM infection is its multiscale nature. Bacterial replication and immune signaling occur at the molecular and cellular levels, host--pathogen interactions evolve within lung tissue, and mucus transport and mucociliary clearance regulate pathogen removal across the entire airway. Cystic fibrosis disrupts host defense at each of these scales by promoting bacterial adaptation \cite{Planet2022adaptation}, impairing macrophage antibacterial function \cite{Jaganathan2022emerging, LEVEQUE2017impact, Keren2020macrophage}, and reducing mucociliary clearance due to abnormal mucus properties \cite{Laube2020-hv}. Consequently, disease progression emerges from the interplay of bacterial, immune, and transport processes that span multiple spatial and temporal scales. Disentangling the individual contributions of these mechanisms through experimental or clinical studies alone remains challenging, motivating the development of computational models capable of integrating these interacting processes.

Mathematical and computational modeling is an important tool for investigating pulmonary host-pathogen interactions. Previous studies have developed models incorporating airway structure \cite{Joslyn2022-tx, Pitcher2020-to}, whole-host immune responses \cite{Flores-Garza2022-pb, Zhang2022-ng}, and spatial infection dynamics. Among these approaches, agent-based models (ABMs) provide detailed mechanistic descriptions of stochastic interactions between bacteria and immune cells \cite{weathered2022role, Weathered2023-gm, Hult2021-in, Millar2020-oi}, while partial differential equation (PDE) models efficiently describe spatial transport, bacterial spread, and mucus dynamics over tissue scales \cite{hao2016modeling, hao2014mathematical, Bartlett2023-jt, Xu2025-va}. These complementary modeling paradigms, however, are typically developed independently. ABMs capture rich cellular-scale biology but are computationally demanding for large-scale parameter exploration and treatment optimization, whereas PDE models are computationally efficient but often rely on phenomenological parameterizations that are only loosely connected to the underlying cellular mechanisms. Bridging these two approaches would enable predictive models that retain mechanistic fidelity while remaining computationally tractable.

In this work, we address this challenge by developing a multiscale continuum framework for early pulmonary NTM infection in cystic fibrosis that is systematically calibrated using data generated from previously established agent-based models \cite{weathered2022role, Weathered2023-ie, Weathered2023-gm}. The proposed PDE model couples bacterial proliferation, macrophage recruitment and immune-mediated clearance, mucus transport, and bacterial exchange between the mucus layer and underlying lung tissue within a unified reaction-diffusion-advection framework. By calibrating the continuum model to mechanistic ABM simulations, we establish a direct link between stochastic cellular interactions and tissue-scale transport dynamics, yielding a computationally efficient surrogate that preserves the essential biological mechanisms governing infection progression but adds the impact of mucus clearance.

This study makes three principal contributions. First, we develop a spatially resolved PDE model that integrates bacterial dynamics, macrophage-mediated immunity, mucus transport, and compartmental exchange in cystic fibrosis airways. Second, we establish a systematic calibration strategy that bridges mechanistic agent-based simulations with continuum modeling, providing an efficient multi-fidelity framework for studying pulmonary infections. Third, using the calibrated model, we identify the dominant biological and transport mechanisms controlling bacterial persistence through global sensitivity analysis and investigate the effects of mucus-targeted and antibacterial therapies on infection outcomes. Collectively, these results demonstrate how the interaction between mucociliary clearance and host immunity shapes pulmonary NTM infection and provide a quantitative framework for evaluating therapeutic strategies and developing future patient-specific computational models.

\section{Methods}
To describe the coupled interactions among bacterial infection, immune response, and mucus transport, we formulate a multiphase PDE system on a two-region airway geometry. The model incorporates bacterial proliferation, macrophage chemotaxis, mucus degradation, and mucus-driven transport through a coupled advection--diffusion--reaction framework.
\subsection{PDE Model}

Let $\Omega=\Omega_1\cup\Omega_2\subset\mathbb{R}^2$ denote the computational domain shown in Figure~\ref{fig:area}, where $\Omega_1$ represents the mucus region and $\Omega_2$ represents the lung parenchyma. The boundary $\Gamma_0$ denotes the upper air-mucus interface, $\Gamma_2$ denotes the lower tissue boundary, and $\Gamma_1$ represents the interface separating the mucus and tissue regions.

To distinguish the two subdomains, we introduce the characteristic function
\begin{equation}
\chi_{\Omega_2}(\bx)=
\begin{cases}
0, & \bx\in\Omega_1,\\
1, & \bx\in\Omega_2.
\end{cases}
\end{equation}

\begin{figure}[H]
    \centering
    \includegraphics[width=0.8\linewidth,height=0.76\textheight,keepaspectratio]{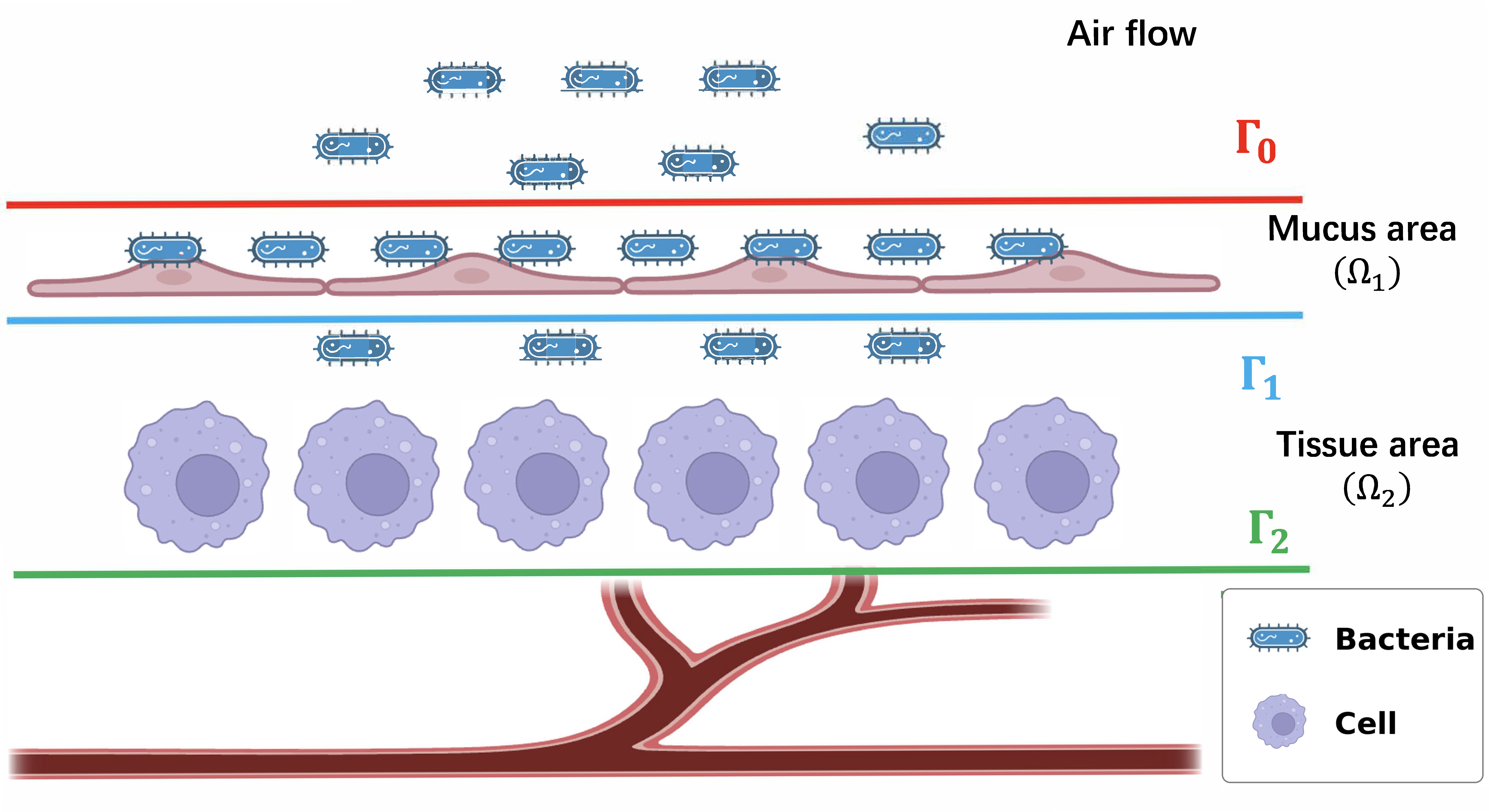}
    \caption{Schematic representation of a localized pulmonary bacterial infection in cystic fibrosis airways. The computational domain $\Omega=\Omega_1\cup\Omega_2$ consists of the mucus area $\Omega_1$, and the lung tissue region $\Omega_2$. The upper boundary $\Gamma_0$ represents the mucus–air interface exposed to airflow, while $\Gamma_1$ denotes the interface between the mucus and tissue regions. The lower boundary $\Gamma_2$ corresponds to the basal tissue boundary. Bacteria and macrophages are present in both regions and can transition across the interface $\Gamma_1$. This spatial configuration illustrates the two-region geometry used in the PDE model.}
    \label{fig:area}
\end{figure}

The primary state variables in the model are the bacterial density $B(\bx,t)$, macrophage density $M(\bx,t)$, and mucin density $m(\bx,t)$. Mucins are glycoproteins that contribute to mucus viscosity, and we therefore represent mucus properties by tracking mucin concentration. The units of these variables and the model parameters are provided in the Supplementary Material.

Based on this geometric configuration, we next introduce the governing equations for each biological component in the model.
\subsubsection{Bacteria dynamics}
The bacterial population serves as the primary infectious component and is influenced by transport, proliferation, and immune-mediated clearance. Bacteria are transported by the mucus velocity field and diffuse heterogeneously in $\Omega$:
\begin{equation}\label{eq:pde-B}
\begin{aligned}
\displaystyle
\partial_t B+\nabla\cdot(\bm uB)-\nabla\cdot\!\big(D_B\nabla B\big)
=\\\lambda_B B\Big(1-\frac{B}{K_B}\Big)
+d_{Mb}\frac{B}{B+K_{Bi}}\,M
-d_B\,MB\,\chi_{\Omega_2}
\end{aligned}
\end{equation}
where the diffusion coefficient is piecewise constant:
\begin{equation}\label{eq:DB}
D_B(\bx)=
\begin{cases}
D_1 (\text{or } D_{B_1}), & \bx\in\Omega_1,\\
D_2 (\text{or } D_{B_2}), & \bx\in\Omega_2,
\end{cases}
\qquad \text{with } D_2<D_1.
\end{equation}
Here $\bm u$ denotes the velocity field, $\lambda_B$ is the bacterial growth rate, $K_B$ is the bacterial carrying capacity, $d_{Mb}$ is the maximum macrophage burst rate, $K_{Bi}$ is the half-saturation constant for intracellular bacteria driving macrophage bursting, and $d_B$ is the clearance rate of bacteria due to macrophage phagocytosis.

The bacterial dynamics therefore couple microbial growth with spatial transport and macrophage-mediated interactions across the heterogeneous airway environment.
\subsubsection{Macrophage dynamics}
To capture the innate immune response, macrophages are modeled as migrating immune cells that respond to bacterial gradients and interact directly with bacteria. Macrophages are advected by $\bm u$, diffuse, and migrate up bacterial gradients via chemotaxis:
\begin{equation}\label{eq:pde-M}
\begin{aligned}
\displaystyle
\partial_t M+\nabla\cdot(\bm uM)-\nabla\cdot(D_M\nabla M)
+\delta\,\nabla\cdot\!\big(M\nabla B\big)
=\\-d_M M
-d_{Mb}\frac{B}{B+K_{Bi}}\,M.
\end{aligned}
\end{equation}
Here $D_M$ denotes the macrophage diffusion coefficient, $d_M$ is the macrophage death rate, and $\delta$ is the chemotaxis sensitivity.

This formulation allows macrophages to both migrate toward infection sites and regulate bacterial populations through localized immune activity.
\subsubsection{Mucus viscosity dynamics}
In addition to cellular dynamics, mucus transport and degradation play a central role in regulating bacterial clearance within the airway. Mucus viscosity evolves only in the mucus region $\Omega_1$ and is represented by the concentration of mucin $(m)$:
\begin{equation}\label{eq:pde-m}
\begin{aligned}
\displaystyle
\partial_t m+\nabla\cdot(\bm um)-\nabla\cdot(D_m\nabla m)=-d_m m.
\end{aligned}
\end{equation}
Here $d_m$ denotes the mucin degradation rate constant and $D_m$ denotes the mucin diffusion coefficient.

The mucin equation therefore informs the spatial medium through which both cellular transport and immune interactions occur.
\subsubsection{Mucus velocity}
To characterize mucus motion within the airway, we further introduce a viscoelastic fluid description for the mucus layer. We model the mucus as a viscoelastic fluid by adopting a Maxwell constitutive law coupled with the Stokes equations. The Maxwell model relates the extra-stress tensor $\boldsymbol{\tau}$ to the velocity gradient $\nabla \bm u$. Under the assumptions of incompressibility and negligible inertia, the governing equations read
\begin{equation}
\left\{
\begin{aligned}
\boldsymbol{\tau} + \lambda \frac{{\rm D}\boldsymbol{\tau}}{{\rm D}t} &= 2\eta_1 \boldsymbol{\varepsilon}(\bm u),\\
\nabla p &= \nabla \cdot \boldsymbol{\tau},\\
\nabla \cdot \bm u &= 0,
\end{aligned}
\right.
\label{Maxwell-Stokes}
\end{equation}
where $\lambda>0$ is the relaxation time, $\eta_1>0$ is the elastic viscosity, and $\boldsymbol{\varepsilon}(\bm u)={1\over2}(\nabla \bm u+\nabla \bm u^T)$ represents the symmetric part of the velocity gradient.

In two dimensions, we write the extra-stress tensor and the velocity field in component form as
\[
    \bm{\tau}=
    \begin{bmatrix}
        \tau_{11} & \tau_{12}\\
        \tau_{12} & \tau_{22}
    \end{bmatrix},
    \qquad
   \bm u=
    \begin{bmatrix}
        u_{1}\\
        u_{2}
    \end{bmatrix}.
\]
Under the system \eqref{Maxwell-Stokes}, the following structural property holds.

Under simplifying geometric assumptions, the Maxwell--Stokes system admits a reduced velocity structure that can be incorporated into the transport equations.

\begin{proposition}\label{clm:quadratic-u}
Assume that $\bm{\tau}=\bm{\tau}(y,t)$ and $\bm u=\bm u(y,t)$ are independent of $x$ and satisfy the Maxwell--Stokes system \eqref{Maxwell-Stokes}. If, in addition, $u_2(y,t)\equiv 0$, then $u_1(y,t)$ is a quadratic polynomial in $y$, with coefficients depending on time.
\end{proposition}

\begin{proof}
Since $\nabla p=\nabla\cdot\bm{\tau}$ and $p$ is a scalar function, we have
\[
\nabla\times(\nabla\cdot\bm{\tau})=\nabla\times(\nabla p)=0.
\]
Because $\bm{\tau}$ is independent of $x$,
\[
\nabla\cdot\bm{\tau}
=
\begin{bmatrix}
\partial_x\tau_{11}+\partial_y\tau_{12}\\
\partial_x\tau_{12}+\partial_y\tau_{22}
\end{bmatrix}
=
\begin{bmatrix}
\partial_y\tau_{12}\\
\partial_y\tau_{22}
\end{bmatrix}.
\]
Therefore,
\[
\nabla\times(\nabla\cdot\bm{\tau})
=
\partial_x(\partial_y\tau_{22})
-
\partial_y(\partial_y\tau_{12})
=
-\partial_{yy}\tau_{12}.
\]
It follows that $\partial_{yy}\tau_{12}=0$, and hence
\[
\tau_{12}(y,t)=a(t)y+b(t)
\]
for some functions $a(t)$ and $b(t)$.

Next, since $u_2=0$ and $\partial_x\bm{\tau}=0$, we obtain
\[
\bm u\cdot\nabla\bm{\tau}
=
u_1\partial_x\bm{\tau}
+
u_2\partial_y\bm{\tau}
=0.
\]
Thus the material derivative reduces to
\[
\frac{{\rm D}\bm{\tau}}{{\rm D}t}
=
\partial_t\bm{\tau}
+
\bm u\cdot\nabla\bm{\tau}
=
\partial_t\bm{\tau}.
\]
The constitutive equation in \eqref{Maxwell-Stokes} therefore becomes
\[
(1+\lambda\partial_t)\bm{\tau}
=
2\eta_1\boldsymbol{\varepsilon}(\bm u)
=
\eta_1\big(\nabla\bm u+(\nabla\bm u)^T\big).
\]
Since $\bm u=(u_1(y,t),0)^T$, we have
\[
\nabla\bm u+(\nabla\bm u)^T
=
\begin{bmatrix}
0 & \partial_y u_1\\
\partial_y u_1 & 0
\end{bmatrix}.
\]
Comparing the $(1,2)$ components gives
\[
\tau_{12}+\lambda\partial_t\tau_{12}
=
\eta_1\partial_y u_1.
\]
Using $\tau_{12}(y,t)=a(t)y+b(t)$, we obtain
\[
\eta_1\partial_y u_1
=
\big(a(t)+\lambda a'(t)\big)y
+
\big(b(t)+\lambda b'(t)\big).
\]
Integrating with respect to $y$ yields
\[
u_1(y,t)
=
\frac{1}{\eta_1}
\left[
\frac{1}{2}\big(a(t)+\lambda a'(t)\big)y^2
+
\big(b(t)+\lambda b'(t)\big)y
\right]
+
C(t),
\]
where $C(t)$ is an integration function depending only on time. Hence $u_1(y,t)$ is a quadratic polynomial in $y$.
\end{proof}

This observation motivates the reduced velocity approximation adopted in the remainder of the model. Therefore, we approximate the velocity field by a quadratic function and impose a zero-velocity condition on the interface,
\begin{equation}
  \bm  u \equiv\bm 0 \qquad \text{on } \Gamma_1.
\end{equation}
\subsubsection{Initial, boundary and interface conditions}
To complete the PDE formulation, we prescribe biologically motivated initial conditions together with boundary and interface transmission conditions.
\paragraph{Initial conditions}
We prescribe the initial distributions of bacteria, macrophages, and mucin. These initial quantities are assumed to be nonzero in the mucus region $\Omega_1$ and zero in the tissue region $\Omega_2$. Specifically,
\begin{equation}
B(\bx,0)=
\begin{cases}
B_0(\bx), & \bx\in\Omega_1,\\
0, & \bx\in\Omega_2,
\end{cases}
\end{equation}
\begin{equation}
M(\bx,0)=
\begin{cases}
M_0(\bx), & \bx\in\Omega_1,\\
0, & \bx\in\Omega_2,
\end{cases}
\end{equation}
\begin{equation}
m(\bx,0)=m_0(\bx), \qquad \bx\in\Omega_1,
\end{equation}

\paragraph{Boundary conditions}
We impose a Neumann-type condition for the bacteria that accounts for both diffusive flux and advective inflow/outflow induced by the velocity field:
\begin{equation}
    (\bm u\cdot \bm n )B-D_B{\partial B\over \partial \bm{n}}=\gamma_BB\quad {\rm on}~ \partial\Omega
\end{equation}
where $\gamma_B$ is nonzero on the outflow boundary, i.e. where $\bm u\cdot\bm n>0$ on $\partial\Omega\cap\partial\Omega_1$, and vanishes on the remaining parts of the boundary.

For the macrophages, due to the distinct roles of advection and chemotaxis, we impose the following boundary conditions:
\begin{equation}
  (\bm u\cdot \bm n )M+\delta {\partial B\over \partial \bm{n}}M -D_M		{\partial M\over \partial \bm{n}}=  \left\{
    \begin{aligned}
       & \alpha(B)(M-M_{\rm ref}) ~\quad   {\rm on}~\Gamma_2\\
    &\gamma_M M ~\qquad  {\rm on}~\Gamma^+\\
     &0   ~\quad  {\rm on}~\partial \Omega\backslash(\Gamma_2\cup \Gamma^+)
    \end{aligned}
    \right.
\end{equation}
Here
\[
\alpha(B)=\alpha_0+\alpha_1 B,
\]
where $\alpha_0$ represents the homeostatic macrophage recruitment rate and $\alpha_1$ characterizes the sensitivity of macrophage recruitment to bacteria-induced chemotactic signals. $\Gamma^+\subset\partial\Omega\cap\partial\Omega_1$ is the outflow boundary (i.e. $\bm u\cdot\bm n>0$) and $M_{ref}$  is a reference macrophage concentration constant representing healthy levels.

For the mucin, since it is present only in the mucus region, we impose the following boundary conditions:
\begin{equation}
    \left\{
\begin{aligned}
    D_m{\partial m\over \partial \bm{n}}+\beta(MB)(m-m_0)&=0 \quad {\rm on}~  \Gamma_1,\\
(\bm u\cdot \bm n)m-D_m				{\partial m\over \partial \bm{n}}&=\gamma_m m \quad {\rm on}~ \partial \Omega_1\backslash\Gamma_1.\\
\end{aligned}\right.
\end{equation}
Here
\[
\beta(MB) = \beta_0+\beta_1MB
\]
where $\beta_0$ represents the baseline mucin production rate and $\beta_1$ characterizes the sensitivity of mucin production to macrophage-bacteria interactions. The coefficient $\gamma_m$ is nonzero on the outflow boundary (i.e.\ $\bm u\cdot\bm n>0$ on $\partial\Omega\cap\partial\Omega_1$) and vanishes elsewhere.

Since bacteria and macrophages may migrate between the mucus and tissue regions, transmission conditions are imposed across the interface $\Gamma_1$.
\paragraph{Interface conditions}
On the interface $\Gamma_1$, we impose the following transmission conditions for bacteria and macrophages:
\begin{equation}
    \left\{
    \begin{aligned}
        {\partial B^m\over \partial \bm{n}^m} = -{\partial B^t\over \partial \bm{n}^t}& = \tau_B(B^t-B^m) \quad  {\rm on}~ \Gamma_1,\\
{\partial M^m\over \partial \bm{n}^m} = -{\partial M^t\over \partial \bm{n}^t}& = \tau_M(M^t-M^m) \quad  {\rm on}~ \Gamma_1,
    \end{aligned}
    \right.
\end{equation}
Here the superscripts $(\cdot)^m$ and $(\cdot)^t$ denote the values of the corresponding quantities on $\Gamma_1$ approached from the mucus region $\Omega_1$ and the tissue region $\Omega_2$, respectively, and $\bm n^{m}$ and $\bm n^{t}$ are the outward unit normal vectors associated with $\Omega_1$ and $\Omega_2$. The parameters $\tau_B>0$ and $\tau_M>0$ represent the interfacial transfer coefficients for bacteria and macrophages, respectively.

\subsection{Agent-Based Model}
We calibrated the PDE model to data from a previously-published ABM of NTM airway infection and innate immune responses \cite{weathered2022role, Weathered2023-gm, Weathered2023-ie}. Briefly, the ABM simulates a 2 mm $\times$ 2 mm $\times$ 0.06 mm section ($2.4 \times 10^{-4} cm^{3}$) of the pulmonary airway. This space is discretized into a 100 $\times$ 100 $\times$ 3 grid, with grid cubes 20 $\mu m$ in length to accommodate a single macrophage. It includes 2 diffusible compounds, bacteria chemoattractant (representing pathogen-associated molecular patterns, or PAMPs) and macrophage chemoattractant (representing macrophage-released chemokines). NTM bacteria and host macrophages are the model agents. Bacteria may divide and secrete bacterial chemoattractant. Macrophages may move, secrete macrophage chemoattractant, phagocytose and eliminate bacteria, and undergo necrosis, releasing some bacteria back into the environment. Further details may be found in Weathered et al. 2022; Weathered, Wei, et al. 2023; and Weathered, Pennington, et al. 2023 \cite{weathered2022role, Weathered2023-gm, Weathered2023-ie}.

While the ABM provides spatiotemporal and mechanistic information about NTM airway infection, it does not include explicit mucus clearance mechanisms. To overcome this limitation while keeping the spatiotemporal and mechanistic information, we developed the PDE model of the airway mucus and underlying lung parenchyma. The PDE model explicitly represents mucus movement and viscosity to simulate the impact of mucus clearance on NTM infection outcomes, while incorporating insights gained from the ABM.

\subsubsection{ABM to PDE Model Parameter and Variable Translation}
Macrophage and bacteria agent counts in the ABM were converted to concentrations using  Equation \eqref{ConcConvert} and the values in Table \ref{tab:S4}. This conversion assumes a well-mixed ABM airway compartment and does not preserve spatiotemporal information on the distribution of cells in the ABM.

\begin{equation}
\text{$C_X$} = \frac{\text{$N_X$} \times \text{$\rho_X$} \times \text{$V_{CX}$}}{\text{$V_{sim}$}}
\label{ConcConvert}
\end{equation}

where $X$ represents either bacteria ($B$) or macrophages ($M$), $C_X$ represents the concentration (in $g/ml$), $N_X$ represents the total number of cells in the ABM, $\rho_X$ represents the density (in $g/cm^3$), $V_{CX}$ represents the average per cell volume of cell $X$ (in $cm^3$), and $V_{sim}$ represents the total volume of the ABM simulation (in $cm^3$).

\begin{table}[H]
	\caption{Cell metrics}\label{tab:S4}
	\begin{center}
    \scalebox{0.9}{
        \begin{tabular}{ | c | c | c | }
        \hline
        Cell type      & Density, $\rho_X$ ($g/cm^3$)& Cell Volume, $V_{CX}$ ($cm^3$)\\
        \hline
        Macrophage     & 1                     & 5 $\times 10^ {-9}$ \cite{Krombach1997-vw}  \\
        Bacterium      & 1.1 \cite{Loferer-Krossbacher1998-lg}    & 8.4 $\times 10^ {-12}$ \cite{Signore2008-ef} \\
        \hline
        \end{tabular}}
    \end{center}
\end{table}

Whenever possible, discrete ABM parameters were converted to continuous parameters to inform parameter ranges for the PDE model calibration. These parameter conversions fell into three categories: unit conversions, discrete-to-continuous parameter conversions, and differential-equation based conversions in which rate constants from ordinary differential equations were fitted to ABM data to derive a parameter value.

Unit conversions for the bacteria carrying capacity and half-saturation of intracellular bacteria involved a conversion from a scalar number of bacteria to a concentration of bacteria. Bacteria carrying capacity was defined in the ABM as the number of NTM bacteria agents that could fit in a single grid cube (Volume = $8.0 \times 10^{-9} \text{cm}^3$). Bacteria carrying capacity in the PDE model was calculated using an adapted version of Equation \eqref{ConcConvert}, with the ABM carrying capacity as the cell count and division by the grid cube volume instead of the whole simulation volume. Similarly, the half-saturation of intracellular bacteria inside of a macrophage in the PDE model was derived from the ABM’s threshold of intracellular bacteria required for an infected macrophage to burst and release NTM bacteria back into the simulation environment. By adapting Equation \eqref{ConcConvert}, we calculated the half saturation using Equation \eqref{burstThr}:

\begin{equation}
K_{Bi}=\frac{\frac{\text{$\tau_{burst}$}}{2}\times \text{$\rho_B$} \times \text{$V_{CB}$}}{\text{$V_{CM}$}}
\label{burstThr}
\end{equation}

where $\tau_{burst}$ represents the ABM threshold for macrophage bursting.

Macrophage volume is used since intracellular bacteria exist within macrophages, so it is the bacterial concentration within a macrophage that determines bursting.

Bacterial growth rates are calculated in the ABM from the bacterium doubling time, as shown in Equation \eqref{doubleTime}.

\begin{equation}
\begin{split}
\text{$\lambda_B$}= \frac{ln(2)}{\text{$t_d$ }\times{\text{$n_{step}$}}}
\end{split}
\label{doubleTime}
\end{equation}

where $t_d$ represents the bacterial doubling time (in hours), and $n_{step}$ represents the number of ABM time steps per hour.

In the stochastic ABM, the growth rate of each bacterium is randomly assigned a value within 10\% of this baseline bacterial growth rate so that each bacterium grows at a slightly different rate. To get a range of growth rates for PDE calibration, we found the minimum and maximum average bacterial growth rates (in units per 6-minute ABM time step) for bacterial agents present at the end of a simulation and converted these rates to units of per day by multiplying by a factor of 240 (24 hours $\times$ 10 time steps per hour).

The clearance rate constant of bacteria by macrophages in the PDE model was found by converting the ABM’s discrete probability of a macrophage clearing a bacterium to a continuous rate constant. Equation \eqref{bactKill} was used, adapted from Jones et al. 2017 to convert transition rates to probabilities, to yield a rate constant of bacteria killed (See Supplementary Material for a full derivation) \cite{jones2017procedure},
\begin{equation}
d_B=-240ln(1-\text{$p_{kill}$}),
\label{bactKill}
\end{equation}
where $p_{kill}$ represents the per-timestep per macrophage probability of killing an intracellular bacterium.

We confirmed the accuracy of this conversion by implementing this rate constant in an ODE model of bacterial death and comparing the results against the ABM’s bacteria concentrations over time (see the Supplementary Material for more information).

The final category is ABM data-based conversions. In the PDE model, macrophages may die due to bursting from intracellular bacteria load or from other causes at a rate $d_M$. This general death rate was calculated from the ABM’s macrophage death. The average death rate per timestep was converted to a death rate per day for use in the PDE model.

Finally, the baseline macrophage recruitment was calculated to counteract general macrophage death (not death due to bursting) under the assumption that in non-infection conditions, the body would maintain homeostasis by replenishing macrophages in the pulmonary tissue. We calculated the macrophage concentration needed on the $\Gamma_2$ boundary to maintain the initial macrophage concentration and calculated $a_0$ as shown in Equation \eqref{recruitMac} (See Supplementary Material for derivation).

\begin{equation}
a_0=3 \cdot d_M M_0
\label{recruitMac}
\end{equation}

\subsection{Numerical Method}

To numerically solve the coupled PDE system with interfacial transmission conditions, we employ a finite element framework that permits discontinuities across the interface $\Gamma_1$ while preserving the coupled transport structure of the model.

Since the interface conditions allow jumps in the bacteria and macrophage densities, standard globally continuous spaces are not sufficient for the discretization. To account for possible discontinuities across the interface $\Gamma_1$, we introduce the following broken Sobolev space for the bacteria and macrophage densities:
\begin{equation}
V(\Omega)
:=
\left\{
v\in L^2(\Omega)\; \middle| \;
v|_{\Omega_1}\in H^1(\Omega_1)
\ \text{and}\
v|_{\Omega_2}\in H^1(\Omega_2)
\right\}.
\end{equation}

To characterize the interfacial coupling terms, we introduce the following jump and average operators across $\Gamma_1$. We define the normal jump across the interface $\Gamma_1$ for a vector-valued function $\bm v$ and a scalar function $v$ by
\[
[\bm v]
:= \bm v^{m}\cdot\bm n^{m}+\bm v^{t}\cdot\bm n^{t},
\qquad
[v]
:= v^{m}\bm n^{m}+v^{t}\bm n^{t},
\]
where the superscripts $(\cdot)^m$ and $(\cdot)^t$ denote traces taken from the mucus region $\Omega_1$ and the tissue region $\Omega_2$, respectively, and $\bm n^{m}$ and $\bm n^{t}$ are the corresponding outward unit normal vectors.

The average of a scalar function $v$ across the interface is defined by
\[
\{v\}=\frac{1}{2}\left(v^{m}+v^{t}\right).
\]
These interface operators are used to weakly enforce the transmission conditions between the mucus and tissue regions.

We also denote by $(\cdot,\cdot)_D$ the $L^2$ inner product over a domain $D$, and by $\langle\cdot,\cdot\rangle_D$ the $L^2$ inner product over the boundary $\partial D$. We also denote the outflow boundary on $\partial\Omega_1\setminus\Gamma_1$ by $\Gamma^{+}$.

Using the above interface notation and applying standard integration-by-parts arguments, we obtain the following weak formulation of the coupled PDE system. The variational problem for the whole system is: Find $(B,M,m)\in \left(V(\Omega)\times V(\Omega)\times H^1(\Omega_1)\right)\times(0,T)$ such that
\begin{equation}
	\label{eq:variation}
	\left\{
	\begin{aligned}
		({\partial B\over \partial t},E)_{\Omega} -(\bm{u}B,\nabla E)_\Omega+(D_B\nabla B,\nabla E)_\Omega=\\(\lambda_B B(1-{B\over K_B})+d_{Mb}{B\over B+K_{Bi}}M,E)_\Omega-\langle \gamma_B B,E\rangle_{\Gamma^+}\\
        -(d_BMB,E)_{\Omega_2}-\langle [\tau_B  B],[D_BE]\rangle_{\Gamma_1} \\
		({\partial M\over \partial t},F)_{\Omega} -(\bm{u}M,\nabla F)_\Omega+(D_M\nabla M,\nabla F)_\Omega-\delta( M\nabla B,\nabla F)_\Omega=\\
        (-d_MM-d_{Mb}{B\over B+K_{Bi}}M ,F)_\Omega -\langle \gamma_M M,F\rangle_{\Gamma^+}\\
		- \langle [\tau_M M],[D_MF]\rangle_{\Gamma_1}-\langle \alpha(B)(M-M_{\rm ref}),F\rangle_{\Gamma_2}-\delta\langle \{\nabla B\},[MF]\rangle_{\Gamma_1}\\
		({\partial m\over \partial t},s)_{\Omega_1}-(\bm{u}m,\nabla s)_{\Omega_1}+(D_m\nabla m,\nabla s)_{\Omega_1}=\\(-d_mm,s)_{\Omega_1}-\langle\beta(MB)(m-m_0),s\rangle_{\Gamma_1}-\langle \gamma_m m,s\rangle_{\Gamma^+}
	\end{aligned}
	\right.
\end{equation}
for any $(E,F,s)\in V(\Omega)\times V(\Omega)\times H^1(\Omega_1)$.

The interface contributions naturally arise from the transmission conditions imposed on $\Gamma_1$, while the boundary integrals account for advective outflow and immune recruitment mechanisms. We consider a conforming triangular mesh $\mathcal{T}_h$ of $\Omega$ such that the interface $\Gamma_1$ coincides with the mesh skeleton. We denote by $\mathcal{T}_h^i=\mathcal{T}_h\cap\Omega_i$ the restriction of the mesh to $\Omega_i$, for $i=1,2$.

Based on this triangulation, we define a piecewise linear finite element approximation space.
\begin{equation}
W_h(\mathcal{T}_h)
:=
\left\{
v_h\in H^1(\mathcal{T}_h)
\;\middle|\;
v_h|_K\in \mathcal{P}_1(K),
\ \forall K\in\mathcal{T}_h
\right\},
\end{equation}
where $\mathcal{P}_1(K)$ denotes the space of linear polynomials on a triangular element $K$.

To allow for possible discontinuities across the interface $\Gamma_1$, we introduce the following broken finite element space for the discretization of the bacteria and macrophage densities:
\begin{equation}
V_h(\mathcal{T}_h)
:=
\left\{
v_h\in L^2(\mathcal{T}_h)
\;\middle|\;
v_h|_{\Omega_1}\in W_h(\mathcal{T}_h^1)
\ \text{and}\
v_h|_{\Omega_2}\in W_h(\mathcal{T}_h^2)
\right\}.
\end{equation}

We then derive the semi-discrete finite element approximation corresponding to the variational formulation: find $(B_h,M_h,m_h)\in \left(V_h(\mathcal{T}_h)\times V_h(\mathcal{T}_h)\times W_h(\mathcal{T}_h^1)\right)\times(0, T)$ such that
\begin{equation}
	\label{eq:semidiscrete}
	\left\{
	\begin{aligned}
		({\partial B_h\over \partial t},E_h)_{\Omega} -(\bm{u}B_h,\nabla E_h)_\Omega+(D_B\nabla B_h,\nabla E_h)_\Omega=\\(\lambda_B B_h(1-{B_h\over K_B})+d_{Mb}{B_h\over B_h+K_{Bi}}M_h,E_h)_\Omega-\langle \gamma_B B_h,E_h\rangle_{\Gamma^+}\\
        -(d_BM_hB_h,E_h)_{\Omega_2}-\langle [\tau_B  B_h],[D_BE_h]\rangle_{\Gamma_1} \\
		({\partial M_h\over \partial t},F_h)_{\Omega} -(\bm{u}M_h,\nabla F_h)_\Omega+(D_M\nabla M_h,\nabla F_h)_\Omega\\
        -\delta( M_h\nabla B_h,\nabla F_h)_\Omega=
        (-d_MM_h-d_{Mb}{B_h\over B_h+K_{Bi}}M_h ,F_h)_\Omega\\
        -\langle \gamma_M M_h,F_h\rangle_{\Gamma^+}- \langle [\tau_M M_h],[D_MF_h]\rangle_{\Gamma_1}\\
		-\langle \alpha(B_h)(M_h-M_{\rm ref}),F_h\rangle_{\Gamma_2}-\delta\langle \{\nabla B_h\},[M_hF_h]\rangle_{\Gamma_1}\\
		({\partial m_h\over \partial t},s_h)_{\Omega_1}-(\bm{u}m_h,\nabla s_h)_{\Omega_1}+(D_m\nabla m_h,\nabla s_h)_{\Omega_1}=\\(-d_mm_h,s_h)_{\Omega_1}-\langle\beta(M_hB_h)(m_h-m_0),s_h\rangle_{\Gamma_1}-\langle \gamma_m m_h,s_h\rangle_{\Gamma^+}
	\end{aligned}
	\right.
\end{equation}
    for any $(E_h,F_h,s_h)\in \left(V_h(\mathcal{T}_h)\times V_h(\mathcal{T}_h)\times W_h(\mathcal{T}_h^1)\right)\times(0,T)$.

To obtain a fully implementable numerical scheme, we discretize the time derivative using a fully implicit Euler method.
\begin{equation}
\frac{\partial A}{\partial t}
\approx
\frac{A-A^{-}}{\Delta t},
\end{equation}
where the time interval $(0,T)$ is uniformly partitioned as
\[
0=t_0<t_1<\cdots<t_N=T,
\qquad
\Delta t=\frac{T}{N}.
\]
Here $A^{-}$ denotes the value of $A$ at the previous time step.

The initial values $B_h^0$, $M_h^0$, and $m_h^0$ are obtained by interpolating the initial data into the finite element space. We then obtain discrete solutions $B_h^i$, $M_h^i$, and $m_h^i$ for $i=1,\ldots,N$.

The resulting numerical method provides a computational framework for simulating the coupled bacteria-macrophage-mucus dynamics in heterogeneous airway geometries with interfacial transport effects.

\subsection{Calibration procedure}

Let $\boldsymbol{\theta}$ denote the vector of unknown parameters to be calibrated. For a given parameter vector $\boldsymbol{\theta}$, we solve the PDE system over the computational domain and compute the model-predicted average concentrations of bacteria and macrophages in the mucus region:
\begin{equation}
\begin{aligned}
    (B_i^{m})_h(\boldsymbol{\theta})
&=\frac{1}{|\Omega_1|}\int_{\Omega_1} B_h(\bx,t_i;\boldsymbol{\theta})\,{\rm d}x,\\
(M_i^{m})_h(\boldsymbol{\theta})
&=\frac{1}{|\Omega_1|}\int_{\Omega_1} M_h(\bx,t_i;\boldsymbol{\theta})\,{\rm d}x.
\end{aligned}
\end{equation}

The calibration objective is defined by the normalized least-squares functional
\begin{equation}
J(\boldsymbol{\theta})
:=J_B(\bm\theta)+J_M(\bm\theta)
\end{equation}
\begin{equation}\label{eq:loss}
    \begin{aligned}
   J_B(\bm\theta):=     \sum_{i=1}^{N}
\frac{\left(B_i^{m}-(B_i^{m})_h(\boldsymbol{\theta})\right)^2}{(B_i^{m})^2}
,\\
J_M(\bm\theta):=  \sum_{i=1}^{N}  \frac{\left(M_i^{m}-(M_i^{m})_h(\boldsymbol{\theta})\right)^2}{(M_i^{m})^2},
    \end{aligned}
\end{equation}
where $B_i^{m}$ and $M_i^{m}$ denote the observed mucus-region average concentrations at time $t_i$. The calibrated parameter vector is obtained by
\begin{equation}
\boldsymbol{\theta}^{*}
=
\arg\min_{\boldsymbol{\theta}\in\Theta} J(\boldsymbol{\theta}),
\end{equation}
where $\Theta$ denotes the admissible parameter space.

The optimization problem is solved using the L-BFGS-B algorithm with
tolerances of $10^{-12}$ for the objective function and the projected
gradient. To improve the calibration accuracy, parameters to
which the objective function is particularly sensitive are subsequently
re-optimized in warm-started refinement stages, with each stage
initialized from the calibrated parameter values obtained in the
preceding stage. The nonlinear PDE
system is solved using Newton's method with a relative tolerance of
$10^{-6}$. The initial values of the PDE variables are obtained by
interpolating the prescribed initial data into the finite element
space.

The model is identified from the temporal evolution of the
spatially aggregated observables generated by the ABM. These
observables are precisely the macroscopic quantities described by the
proposed model. Accordingly, the calibration loss is formulated in terms of the
corresponding spatial averages of the PDE solution, since these
quantities are directly comparable with the aggregate ABM outputs and
allow the calibration to focus on the robust macroscopic dynamics. Hence, the ABM and the PDE model are compared at a consistent level of
description: the effects of the underlying spatial interactions are
incorporated through their net influence on the aggregate dynamics.
The resulting model is parsimonious, involving 15 fitted parameters
for 84 temporal ABM observations. The calibrated parameters and their
prescribed bounds are reported in Table~\ref{tab:model-fit}. The initial
parameter values are selected based on the reference values reported in
\cite{hao2016modeling,Weathered2023-gm,denneny2020mucins}. This consistency across the available ABM
data indicates that the identified terms represent robust macroscopic
dynamics.

\begin{table}[H]
\centering
\caption{
Calibration parameters and their prescribed bounds.
}
\label{tab:model-fit}
\small
\begin{tabular}{l|c}
\hline
Parameter(s) & Bounds \\
\hline
$D_1,\;D_2,\;d_B,\;D_M,\;\delta,\;d_M,\; K_B,\; K_{Bi}$
& $[10^{-9},\,1]$ \\

$d_{Mb}$
& $[10^{-8},\,10]$ \\

$\lambda_B$
& $[4.61\times10^{-4},10]$\\

$\tau_B,\;\tau_M,\;\alpha_0$
& $[0,\,100]$ \\

$M_{\rm ref}$
& $[10^{-4},\,0.75]$ \\

$\alpha_1$
& $[-100,\,100]$ \\
\hline
\end{tabular}
\end{table}

\begin{remark}
The numerical time step used to solve the PDE system may be smaller than the data sampling interval. Model outputs are then evaluated at the observation time points for comparison with the data.
\end{remark}

\subsection{Treatment formulation}

To model treatment effects, we incorporate mucin-dependent velocity, diffusion, and viscosity into the PDE system. Both the velocity field and diffusion coefficients are assumed to depend on mucin concentration through the viscosity.

\paragraph{Viscosity law.}
We adopt the exponential relation to represent mucus viscosity ($\eta$)
\begin{equation}
\eta(m)=\eta_0\exp\!\left(\frac{m}{m_0}-1\right),
\end{equation}
where $\eta_0>0$ is the reference viscosity and $m_0$ is a characteristic mucin concentration.

\paragraph{Velocity field.}
We consider a unidirectional velocity field of the form
\begin{equation}
\bm u=(u_1,0)^T,
\end{equation}
where
\begin{equation}
u_1(y,t)=\frac{1}{\eta(m)}u_0(y,t),
\end{equation}
for a prescribed baseline velocity profile $u_0$. Here $\eta=\eta(m)$ is an increasing function of the mucus density $m$.

\paragraph{Diffusion coefficients.}
Relaxing the piecewise-constant assumption \eqref{eq:DB}, we allow the
motility of bacteria and macrophages to depend on the local rheology of
the medium. By analogy with the Stokes--Einstein relation, both
diffusivities are taken to be inversely proportional to the viscosity,
\begin{equation}
  D_B(\bm{x},t)=\frac{D_0^B(\bm{x})}{\eta\bigl(m(\bm{x},t)\bigr)},
  \qquad
  D_M(\bm{x},t)=\frac{D_0^M(\bm{x})}{\eta\bigl(m(\bm{x},t)\bigr)},
  \label{eq:diff-visc}
\end{equation}
where the reference diffusivities are piecewise constant,
\begin{equation}
  D_0^{\bullet}(\bm{x})=
  \begin{cases}
    D_{0,1}^{\bullet}, & \bm{x}\in\Omega_1,\\[2pt]
    D_{0,2}^{\bullet}, & \bm{x}\in\Omega_2,
  \end{cases}
  \qquad \bullet\in\{B,M\},
  \label{eq:D0}
\end{equation}
and the viscosity is
extended to the tissue region by $\eta\equiv\eta_0$ on $\Omega_2$. At the
reference mucin concentration $m=m_0$ we have $\eta=\eta_0$, so that
$D_{0,i}^{\bullet}/\eta_0$ recovers the constant diffusivities of
\eqref{eq:DB}.

\paragraph{Antibiotic treatment.}
Antibiotic treatment is represented by an additional bacterial removal term in the bacteria equation:
\begin{equation}
\begin{aligned}
    \partial_t B
+\nabla\cdot(\bm u B)
-\nabla\cdot(D_B\nabla B)
=
\lambda_B B\Big(1-\frac{B}{K_B}\Big)\\
+d_{Mb}\frac{B}{B+K_{Bi}}M
-d_BMB\,\chi_{\Omega_2}
-d_{\rm anti} B.
\end{aligned}
\end{equation}
The treatment coefficient is scaled as
\begin{equation}
d_{\rm anti}=\zeta d_BM_0,
\end{equation}
where $\zeta>0$ is a dimensionless treatment-strength parameter.

\paragraph{Outcome measures.}
Treatment efficacy is evaluated using the normalized total bacterial and macrophage burdens:
\begin{equation}
V_B(T)=\frac{1}{B_0}\int_{\Omega}B(\bx,T)\,{\rm d}x,
\qquad
V_M(T)=\frac{1}{M_0}\int_{\Omega}M(\bx,T)\,{\rm d}x.
\end{equation}

\section{Results}

\subsection{Calibration data and results}

We calibrated the PDE model using ten representative groups of data generated from the agent-based model \cite{Weathered2023-gm}. Each dataset provides time-series measurements of the spatial average concentrations of bacteria and macrophages in the mucus region over a time horizon of $T=14$ days. The data sampling interval is $\Delta t=1/6$ day, giving observation times
\[
t_i=i\Delta t,
\qquad i=1,\ldots,84.
\]
The observed quantities of bacteria and macrophages are denoted by $B_i^{m}$ and $M_i^{m}$, respectively.

For each ABM-derived dataset, the PDE model was calibrated by minimizing the objective function defined in the Methods section. The resulting calibrated parameter sets reproduce the mucus-region bacterial and macrophage concentration trends obtained from the ABM simulations and are used as baseline parameter sets for the subsequent sensitivity analysis and treatment simulations.

Figure~\ref{fig:no756-calibration} shows representative calibration results comparing the ABM-derived data with the PDE simulations in the mucus region, together with the corresponding PDE predictions in the tissue region. It also illustrates the PDE system evolution at different time points, both with and without the mucus velocity field. The comparison indicates that the mucus velocity field mainly affects the spatial redistribution of bacteria and macrophages, producing more pronounced transport in the mucus region. See Supplementary Figures ~\ref{fig:appendix-set048-velocity-comparison}--\ref{fig:appendix-set843-velocity-comparison} for full calibration results.

To provide a direct assessment of the calibration
accuracy, we report the mean squared relative errors between the calibrated PDE
predictions and the corresponding ABM observations. The dataset-wise errors are
summarized in Table~\ref{tab:relative-error}. The  errors remain below $0.04$
for bacteria and below $7\times 10^{-3}$ for macrophages in every dataset considered. These results demonstrate that the PDE model consistently
reproduces the aggregate ABM observations across the different
calibration conditions.

The ABM outputs are generated by a stochastic model and may therefore
contain realization-dependent variability \cite{fadikar2018calibrating}. The deterministic PDE model is intended
to capture the reproducible macroscopic evolution represented by these
aggregate observations. Consequently, the remaining
PDE--ABM discrepancy may contain contributions from both the
approximation of the macroscopic dynamics and stochastic variability in
the ABM observations. Although these contributions are not separated
quantitatively here, the reported MSRE values and the accurate
reproduction of the principal temporal behaviour support the calibrated
PDE model as a robust quantitative description of the macroscopic
dynamics.

\begin{figure}[t]
    \centering
    \begin{subfigure}{0.78\textwidth}
        \centering
        \includegraphics[width=.8\linewidth,height=0.76\textheight,keepaspectratio]{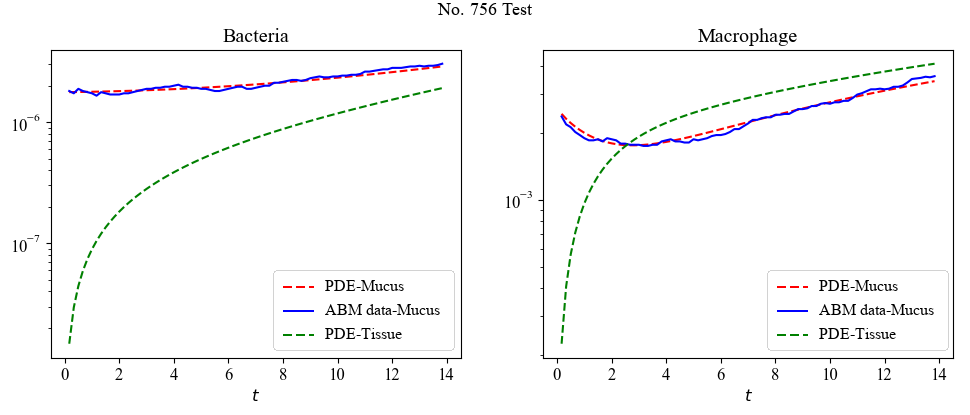}
        \caption{Concentration fitting for bacteria and macrophages in the representative parameter set No.~756, including the calibration reference, mucus-region simulation, and tissue-region simulation.}
        \label{fig:no756-calibration}
    \end{subfigure}

    \vspace{0.5em}
    \begin{subfigure}{0.30\textwidth}
        \centering
        \includegraphics[width=\linewidth,height=0.76\textheight,keepaspectratio]{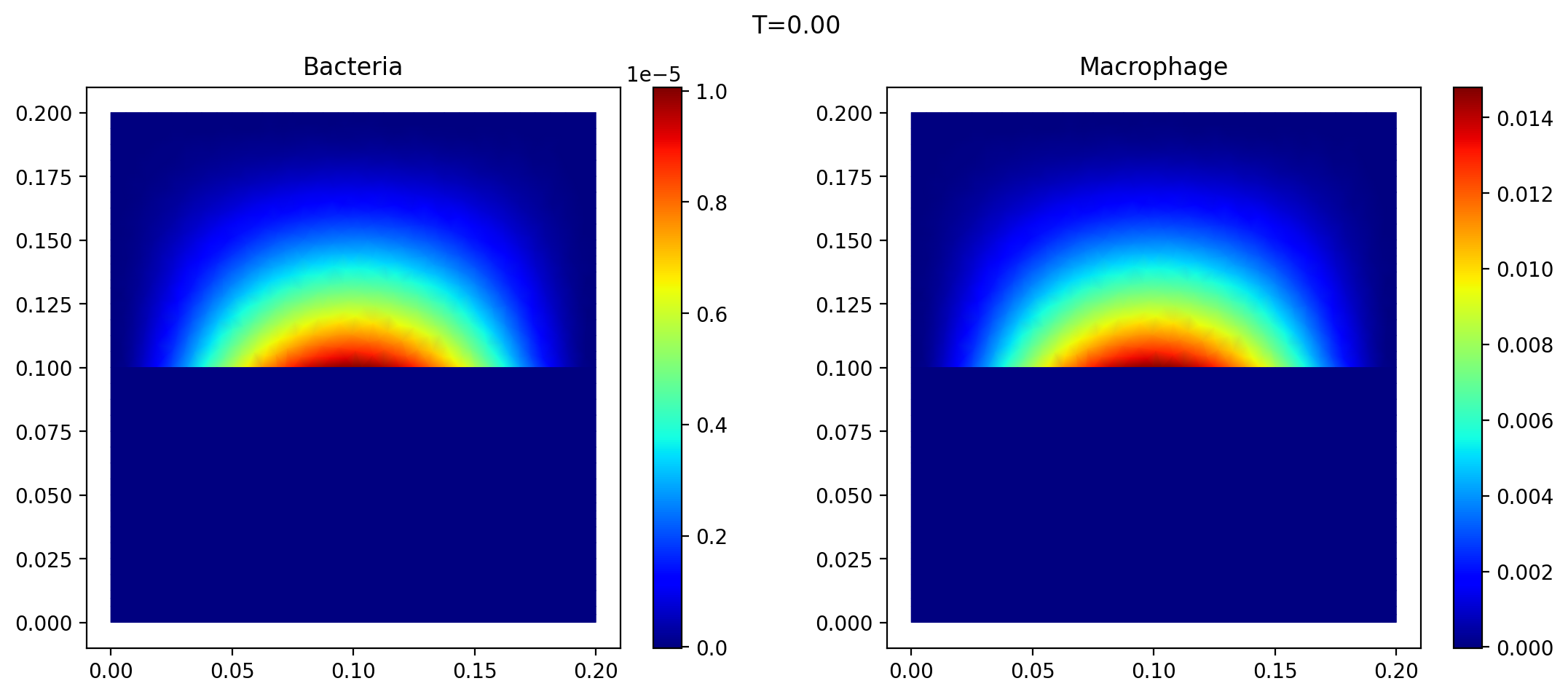}
        \caption{No velocity, $t=0$ d.}
        \label{fig:no756-baseline-t0}
    \end{subfigure}
    \hfill
    \begin{subfigure}{0.30\textwidth}
        \centering
        \includegraphics[width=\linewidth,height=0.76\textheight,keepaspectratio]{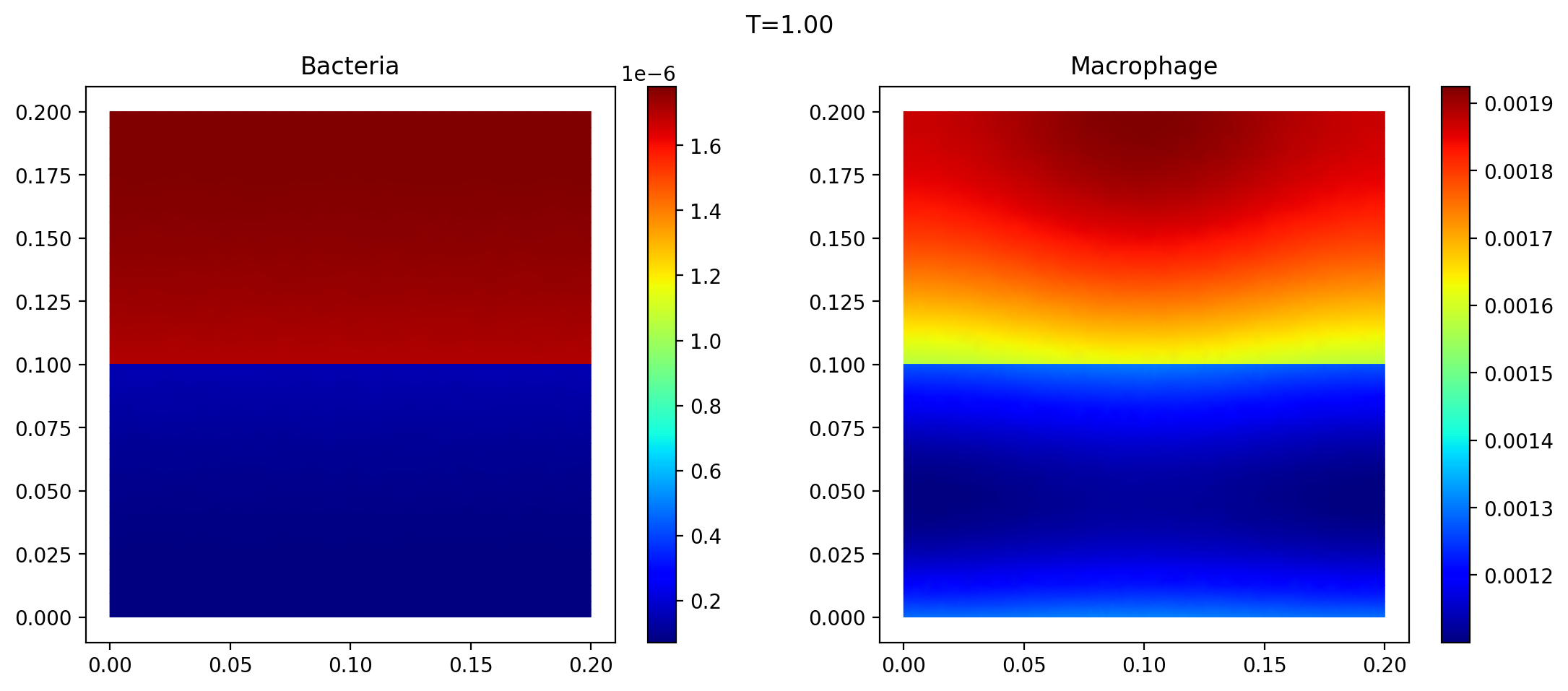}
        \caption{No velocity, $t=1$ d.}
        \label{fig:no756-baseline-t1}
    \end{subfigure}
    \hfill
    \begin{subfigure}{0.30\textwidth}
        \centering
        \includegraphics[width=\linewidth,height=0.76\textheight,keepaspectratio]{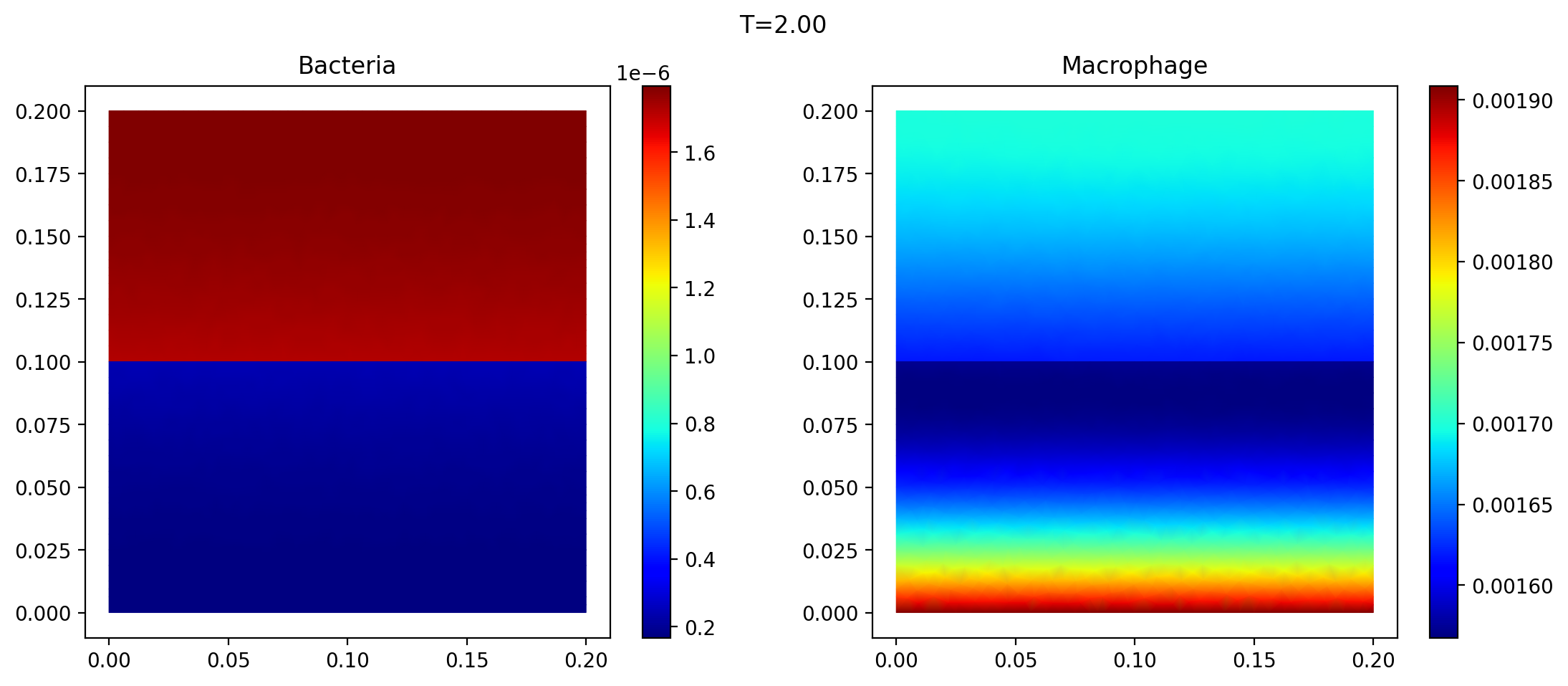}
        \caption{No velocity, $t=2$ d.}
        \label{fig:no756-baseline-t2}
    \end{subfigure}

    \vspace{0.5em}
    \begin{subfigure}{0.30\textwidth}
        \centering
        \includegraphics[width=\linewidth,height=0.76\textheight,keepaspectratio]{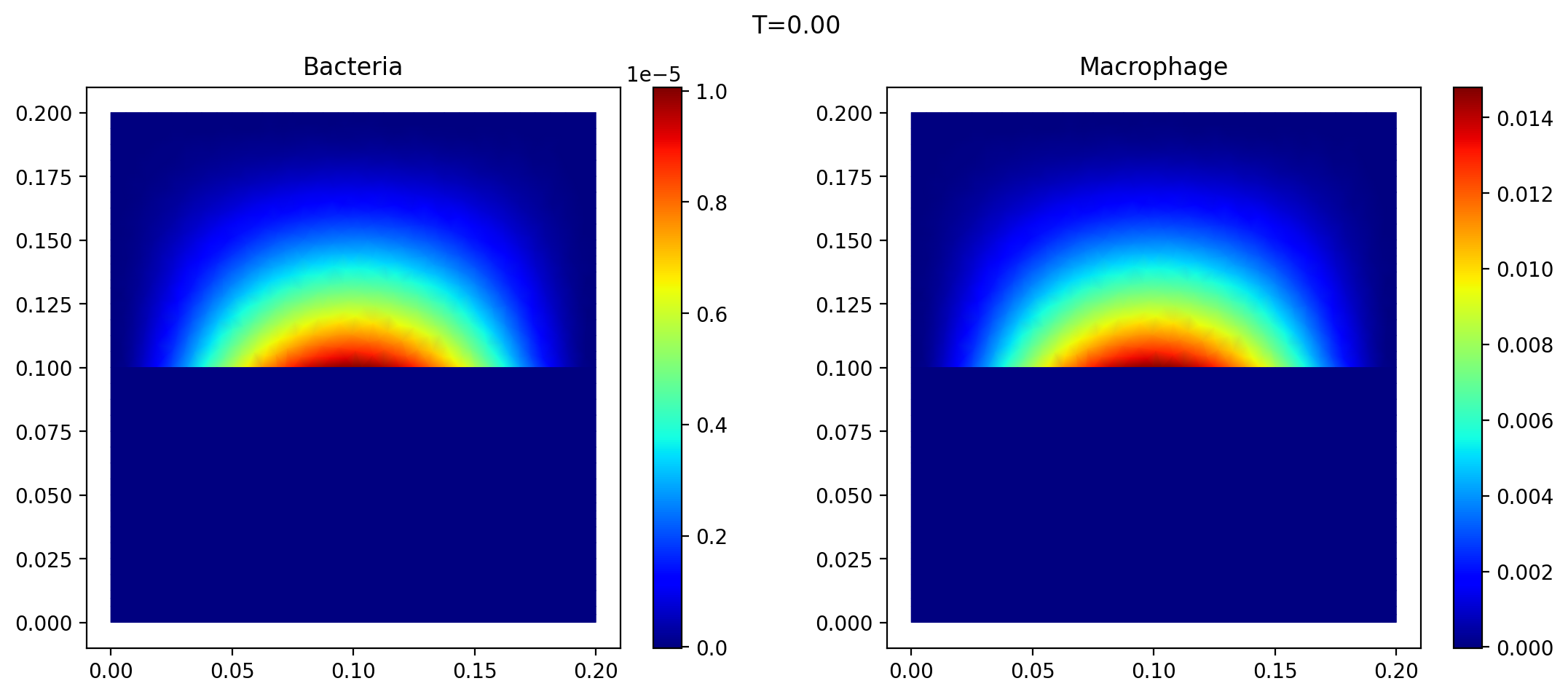}
        \caption{With velocity, $t=0$ d.}
        \label{fig:no756-velocity-t0}
    \end{subfigure}
    \hfill
    \begin{subfigure}{0.30\textwidth}
        \centering
        \includegraphics[width=\linewidth,height=0.76\textheight,keepaspectratio]{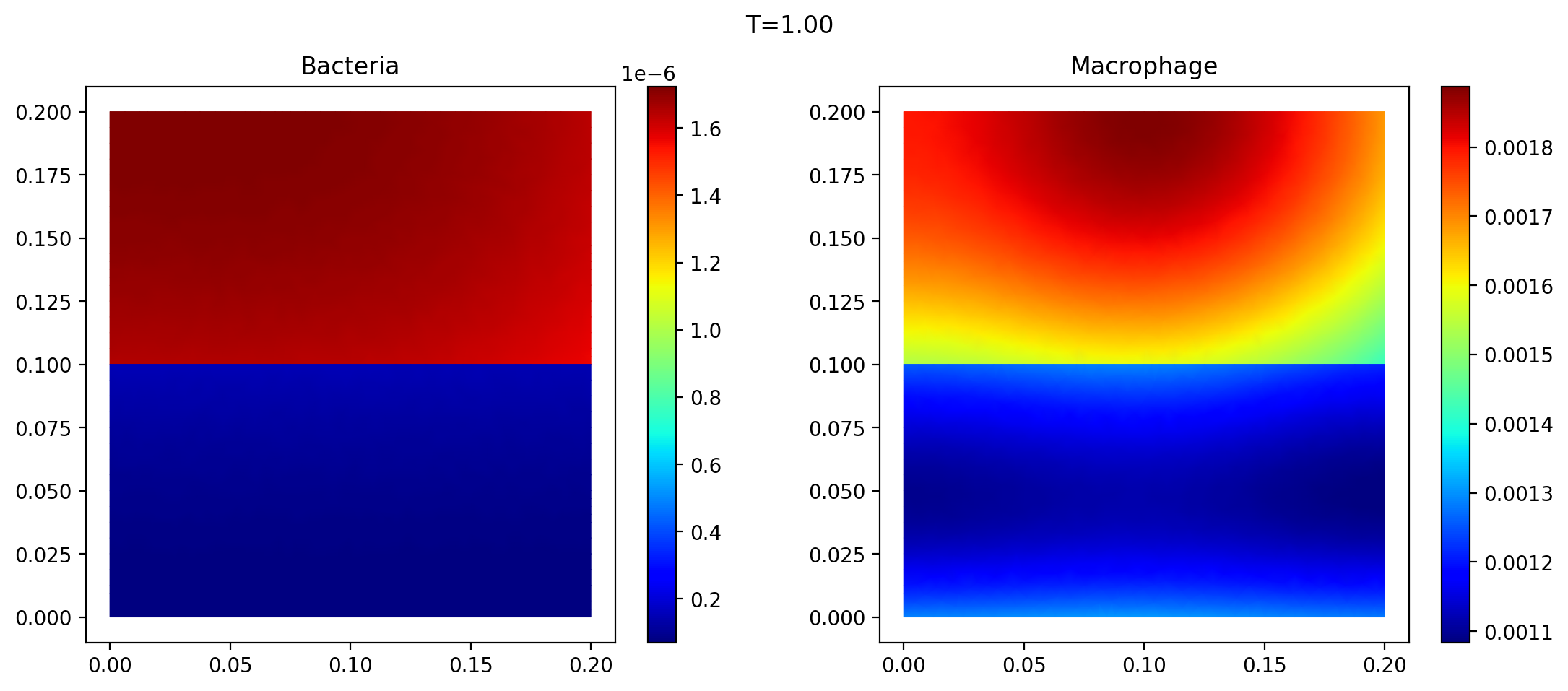}
        \caption{With velocity, $t=1$ d.}
        \label{fig:no756-velocity-t1}
    \end{subfigure}
    \hfill
    \begin{subfigure}{0.30\textwidth}
        \centering
        \includegraphics[width=\linewidth,height=0.76\textheight,keepaspectratio]{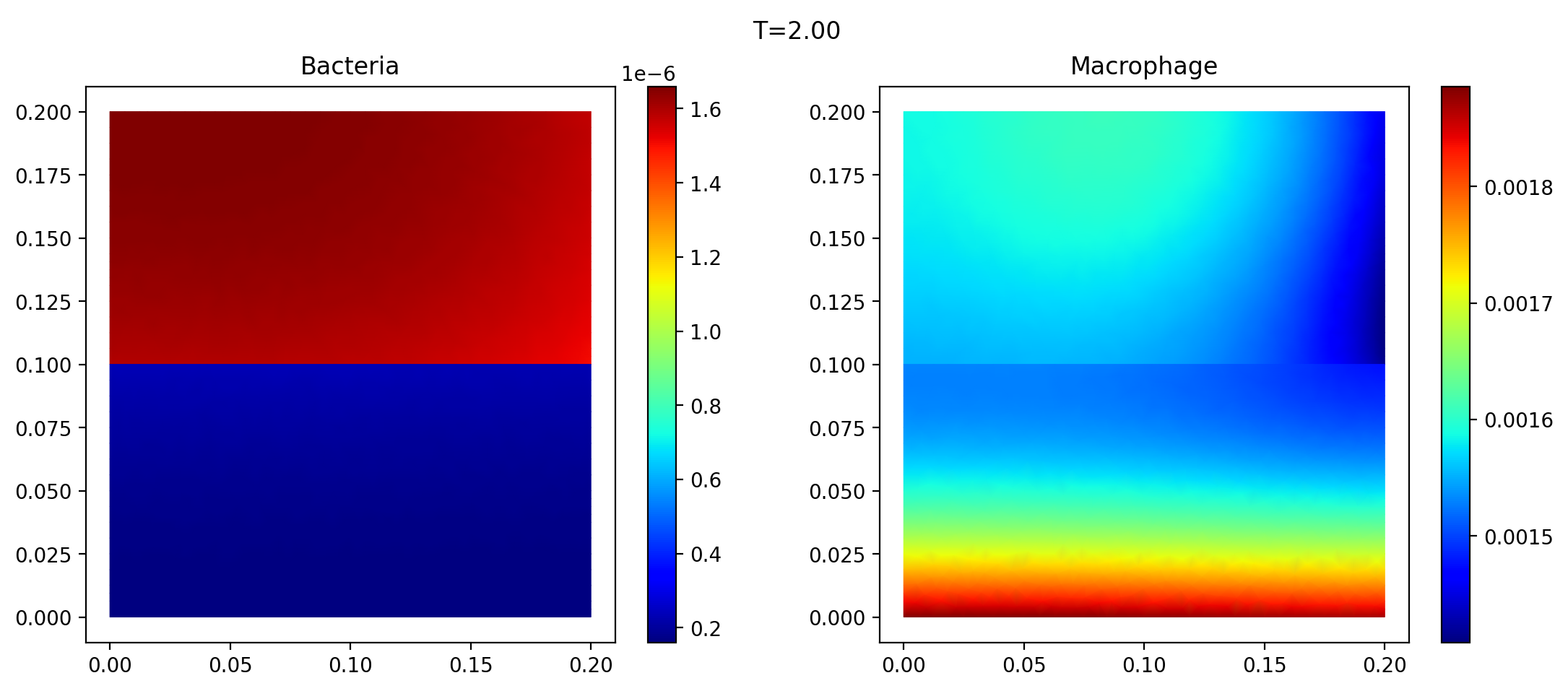}
        \caption{With velocity, $t= 2$ d.}
        \label{fig:no756-velocity-t2}
    \end{subfigure}

    \caption{
    Evolution of the system for a representative parameter set, with and without the mucus velocity field.
    (a) Concentration fitting for bacteria and macrophages, with calibration reference data, mucus-region simulation, and tissue-region simulation shown in each plot.
    (b--d) Spatial distributions without the mucus velocity field  at $t=0$, $1$, and $2$ days.
    (e--g) Spatial distributions with the mucus velocity field  at the same time points.
    Each spatial panel shows bacteria on the left and macrophages on the right.
    }
    \label{fig:no756-velocity-comparison}
\end{figure}

\begin{table}[H]
\centering
\caption{
Mean squared relative errors between the calibrated PDE
predictions and the corresponding ABM observations for $B$ and $M$.
}
\label{tab:relative-error}
\begin{tabular}{ccc}
\hline
Dataset
& $\mathrm{Err}_{B}$
& $\mathrm{Err}_{M}$
\\
\hline
1  & 4.86e-03 & 2.02e-03 \\
2  & 1.10e-02 & 1.77e-03 \\
3  & 6.80e-03 & 6.52e-03 \\
4  & 2.02e-02 & 2.59e-04 \\
5  & 3.92e-02 & 1.44e-03 \\
6  & 2.14e-02 & 4.32e-03 \\
7  & 9.14e-03 & 4.33e-04 \\
8  & 2.14e-03 & 1.15e-03 \\
9  & 3.35e-03 & 4.57e-04 \\
10 & 2.81e-03 & 2.65e-03 \\
\hline
Mean $\pm$ std (\%)
& $1.21 \pm 1.11$
& $0.21 \pm 0.19$
\\
\hline
\end{tabular}
\end{table}

The calibrated PDE model serves as a computationally efficient surrogate for the ABM. While accurately reproducing the average bacterial and macrophage dynamics, the PDE model requires substantially less computational time, thereby making large-scale sensitivity analyses and treatment optimization feasible. To quantify this computational advantage, Table~\ref{tab:computational-time} compares the cost of simulating the system up to $T=14$ days. On average, a calibrated PDE simulation requires approximately 16~s, whereas the corresponding ABM simulation takes approximately 115~s. This substantial reduction in computational cost demonstrates the greater efficiency of the PDE model.

\begin{table}[H]
\centering
\caption{
Comparison of computational times for the PDE model and the ABM over the simulation period up to $T=14$ days.
}
\scriptsize
\label{tab:computational-time}
\begin{tabular}{ccc}
\hline
Run & PDE Computational time (s) & ABM Computational time (s)\\
\hline
1  & 17.714 & 510.166\\
2  & 16.008 & 68.392\\
3  & 18.595 & 154.312\\
4  & 15.130 & 95.443\\
5  & 14.844 & 76.063\\
6  & 17.997 & 62.556\\
7  & 15.100 & 49.178\\
8  & 15.208 & 45.178\\
9  & 15.427 & 45.200\\
10 & 15.230 & 46.918\\
\hline
Mean & 16.125 & 115.344\\
\hline
\end{tabular}
\end{table}

\subsection{Sensitivity analysis}
To better understand the impact of mucus dynamics and other biological processes on infection outcomes, we performed a sensitivity analysis to measure correlations between model parameter values and bacteria and macrophage concentrations. We ran 400 PDE model simulations and varied 16 parameters within the calibrated ranges for each parameter. Partial rank correlation coefficients (PRCCs) and p-values with the Bonferroni correction were found according to the method presented by Marino et al. 2008 \cite{MARINO2008methodology}. We limit our discussion here to correlations with model outputs at the final time point (day 14 of infection) since the observed correlations were largely consistent over time. Full sensitivity analysis results are presented in Supplementary Figure \ref{fig:supp_full_PRCC}.

Several pathogen and host parameters were significantly correlated with total bacteria and macrophage concentrations (mucus and tissue compartments combined) (Figures \ref{fig:Sensitivity Analysis PRCC}, \ref{fig:Sensitivity Analysis Figure}). Some parameter influences are expected. For example, bacteria proliferation rate $\lambda_B$ was positively associated with total bacteria concentration (PRCC = 0.514, p < 0.0001), and baseline macrophage recruitment rate $\alpha_0$ (PRCC = 0.724, p < 0.0001) was positively associated with total macrophage concentration.

Other parameter correlations were less intuitive. Bacteria mucus diffusivity $D_{B_1}$ was negatively associated with total bacteria concentration (PRCC = -0.827, p < 0.0001), bacteria mucus concentration (PRCC = -0.779, p < 0.0001), and bacteria tissue concentration (PRCC = -0.850, p < 0.0001). We hypothesize that increased bacteria diffusion in the mucus led to increased bacteria mucociliary clearance from the simulation space. Indeed, bacteria mucus diffusivity $D_{B_1}$ was positively associated with cumulative bacteria mucociliary clearance (PRCC = 0.783, p < 0.0001). Increased bacteria diffusivity in the mucus appears to create more movement within the mucus, providing more opportunities to reach the mucociliary clearance boundary in the simulation. While not significantly associated with macrophage mucus concentration, bacteria mucus diffusivity $D_{B_1}$ was positively associated with macrophage mucociliary clearance (PRCC = 0.921, p < 0.0001). We hypothesize that macrophages may be undergoing chemotaxis towards bacteria as they are cleared at the mucociliary clearance  boundary.

Initial mucus viscosity $\eta_0$ (PRCC = 0.601, p < 0.0001) was positively associated with total macrophage concentration, while initial macrophage diffusivity $D_M$ was negatively associated (PRCC = -0.417, p < 0.0001). We hypothesize that higher initial viscosity and lower macrophage diffusivity would both contribute to ``trapping'' macrophages in the mucus compartment, leading to greater overall macrophage concentrations. Indeed, initial mucus viscosity $\eta_0$  was positively associated with macrophage mucus concentration (PRCC = 0.753, p < 0.0001), and macrophage diffusivity  $D_M$ was negatively associated (PRCC = -0.706, p < 0.0001). In contrast, initial mucus viscosity $\eta_0$ was negatively associated (PRCC = -0.561, p < 0.0001) with macrophage tissue concentration, and macrophage diffusivity $D_M$ (PRCC = 0.579, p < 0.0001) was positively associated. The importance of this balance between macrophage populations in the mucus vs. tissue is also supported by the positive association between macrophage tissue concentration and mucus/tissue interfacial transfer coefficient for macrophages $\tau_M$ (PRCC = 0.329, p < 0.0001), suggesting that macrophage movement between mucus and tissue compartments could help counteract their entrapment in very viscous mucus. Nonetheless, $\eta_0$ was indeed negatively associated (PRCC = -0.171, p < 0.02) with macrophage mucociliary clearance, although macrophage diffusivity $D_M$ was not found to significantly correlate with macrophage mucociliary clearance.

In agreement with these findings, transfer of bacteria and macrophages between mucus and tissue was also revealed as a key mechanism influencing infection dynamics. Macrophage diffusivity $D_M$ is positively associated with mucus-to-tissue transfer of macrophages (PRCC = 0.560, p < 0.0001), while initial mucus viscosity $\eta_0$ was negatively associated (PRCC = -0.539, p < 0.0001). This suggests that the impact of macrophage diffusivity and mucus viscosity on macrophage dynamics is achieved through its influence on both mucociliary clearance as well as mucus to tissue transfer of macrophages.

These inter-compartmental macrophage dynamics also influence bacterial dynamics. Initial mucus viscosity $\eta_0$ (PRCC = 0.820, p < 0.0001) was positively associated with cumulative mucus-to-tissue transfer of bacteria, and macrophage diffusivity $D_M$ (PRCC = -0.828, p < 0.0001) was negatively associated. These inverse relationships, compared to their impact on macrophage populations, are consistent with lower macrophage numbers in the mucus (from higher diffusivity or lower viscosity) leading to higher bacterial numbers in the mucus, and thus more possible bacterial flux into the tissue.

\begin{figure}[H]
    \centering
    \includegraphics[width=.95\linewidth,height=0.76\textheight,keepaspectratio]{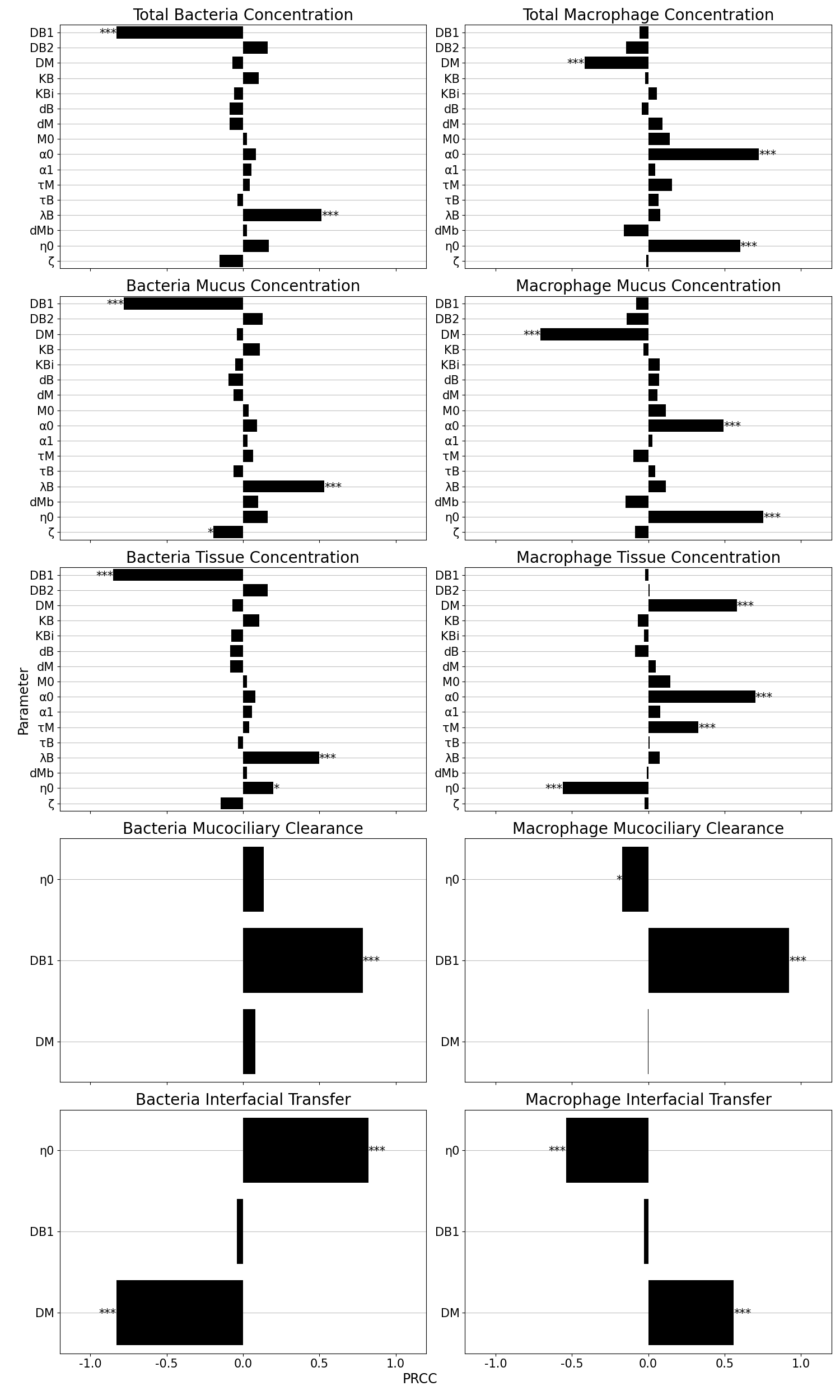}
    \caption{Sensitivity analysis reveals key parameters driving bacteria and macrophage concentrations, mucociliary clearance, and interfacial transfer. * p < 0.05, ** p < 0.01, *** p < 0.001}
    \label{fig:Sensitivity Analysis PRCC}
\end{figure}

These findings suggest that mucociliary clearance-associated mechanisms contribute significantly to bacterial control even in the context of immune responses, and that mucus-associated mechanisms significantly impact tissue compartment dynamics. However, these findings also suggests the possibility of non-monotonic impacts of therapeutics targeting mucus function if the bacterial populations in the mucus vs. tissue compartments are affected differently.

\begin{figure}[H]
    \centering
    \includegraphics[width=0.65\linewidth,height=0.76\textheight,keepaspectratio]{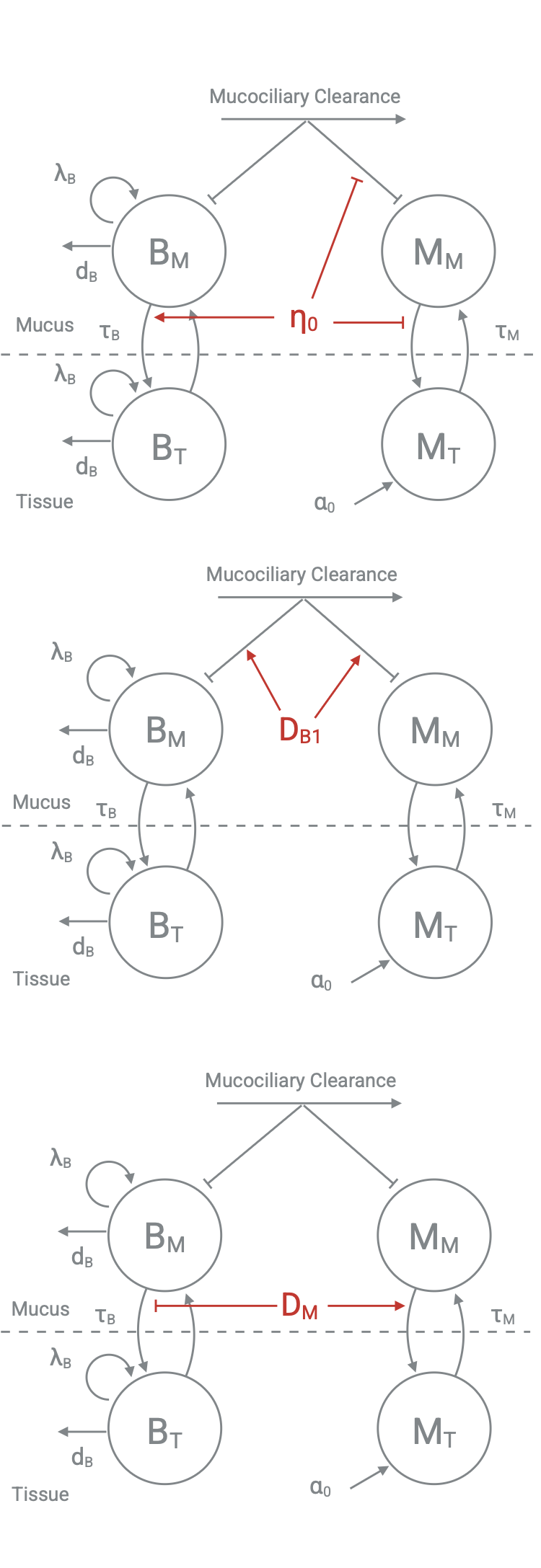}
    \caption{Sensitivity analysis suggests that initial mucus viscosity $\eta_0$ and macrophage diffusivity $D_M$ impact both macrophages and bacteria through mucociliary clearance as well as enabling macrophage killing of bacteria, whereas bacteria diffusivity in mucus $D_{B_1}$ mainly functions by facilitating mucociliary clearance of the bacteria. $B_M$: Mucus bacteria, $M_M$: Mucus macrophages, $B_T$: Tissue bacteria, $M_T$: Tissue macrophages, $\lambda_B$: bacteria proliferation rate, $d_B$: bacteria death rate, $\alpha_0$: baseline macrophage recruitment rate, $\tau_B$: bacteria interfacial transfer coefficient, $\tau_M$: macrophage interfacial transfer coefficient.}
    \label{fig:Sensitivity Analysis Figure}
\end{figure}

\subsection{Treatment simulations}
Since bacterial proliferation, immune responses, and mucociliary clearance appear to contribute to infection progression in complex ways, we explored the potential impact of mucus-thinning therapies in the context of antibacterials. By incorporating perturbations to the viscosity-dependent transport and antibiotic-induced bacterial removal in the PDE model, we further investigate the impact of antibiotic treatment and mucus rheology through local sensitivity analysis.

We ran $13^2=169$ simulations for every parameter set, varying the antibiotic bacteria killing rate $d_{anti}$ and initial mucus viscosity $\eta_0$ for each of the 10 parameter sets identified during initial model calibration. Antibiotic bacteria killing rate $d_{anti}$ was varied on the interval $[10^{-8}, 5]$ to represent antibiotic therapy, and initial mucus viscosity $\eta_0$ was varied on the interval [0.4, 5] to represent mucolytic therapy. From these simulations, we generated heatmaps of final bacteria and macrophage concentrations in the whole simulation space, mucus, and tissue.

Considering the total bacteria heatmaps, all 10 parameter sets demonstrate a clear delineation between antibiotic-dominant and mucolytic-dominant regions of the $d_{anti}$ vs. $\eta_0$ parameter space  (Figure  \ref{fig:Treatment heatmap figure}; Supplementary Figures \ref{fig:ntm_heatmap_supp_48_83} to \ref{fig:ntm_heatmap_supp_808_843}). Antibiotic efficacy is the main driver of final bacteria burden above a threshold of approximately $log(d_{anti}) = -2$. Below that efficacy, mucus viscosity appears to be the dominant determinant of bacterial burden. Total macrophage concentration, as expected, is primarily driven by mucus viscosity rather than antibiotics. Collectively, our simulations fell into 2 main categories: those for which total macrophage concentrations decrease with decreasing viscosity (i.e. mucolytic therapy) (Scenario I) and those for which macrophage concentrations increase with decreasing viscosity (Scenario II).

\paragraph{Scenario I: Macrophage Retention-Dominated Response}
Scenario I (7 out of the 10 parameter sets) demonstrates decreasing macrophage concentration with decreasing mucus viscosity. This is consistent with the global sensitivity analysis results demonstrating a positive association between initial mucus viscosity $\eta_0$ and total final macrophage concentration (See Fig. 4). This indicates that less-viscous mucus from mucolytic therapy potentially leads to lower macrophage concentrations overall. In 5 of these simulations, less-viscous mucus results in macrophages being predominantly located in the tissue compartment, which aligns with the sensitivity analysis finding that viscosity is negatively associated with macrophage transfer to the tissue.

In 5 of the scenario I simulations, there was a non-monotonic impact of mucus viscosity on bacterial burden, with higher bacterial loads predicted for both high and low viscosity. This aligns with the fact that final bacteria burden was not found to be sensitive to initial mucus viscosity $\eta_0$ in the global sensitivity analysis. From the heatmaps of bacteria concentration in the mucus and tissue compartments, as viscosity decreases, bacteria mucus concentrations decrease or remain nearly constant, while bacteria tissue concentrations increase. This suggests that as mucus becomes less viscous, there is an increased shift of bacteria from the mucus to the tissue. Taken together, these results indicate that while mucolytic therapies are effective at clearing bacteria and macrophages from the mucus, they may also drive the remaining bacteria and macrophages into the tissue, where antibacterials would be required to eliminate the infection.

\paragraph{Scenario II: Recruitment-Dominated Response}
Scenario II (3 of the 10 simulations) demonstrates increasing total macrophage concentration with decreasing mucus viscosity. The compartment-specific results reveal that this outcome is primarily driven by mucus macrophages. Among all 10 parameter sets, these 3 simulations had the highest values of baseline macrophage recruitment rate $\alpha_0$ and the highest final macrophage counts (>1000). This, along with the negative association between viscosity and macrophage mucus to tissue transfer, suggests that macrophage recruitment and subsequent transfer to the mucus is compensating for mucociliary clearance of macrophages. However, as with the Scenario I simulations, less-viscous mucus still mainly shifts bacteria into the tissue compartment from the mucus. Therefore, in these Scenario II simulations, there is a seeming misalignment between where the bacteria and macrophages are located.

Overall, these local sensitivity analyses suggest that bactericidal antibiotics would be the preferred method for clearing NTM infection, while mucolytic therapies are effective at clearing bacteria present in the mucus. However, there is a risk of the infection remaining predominantly in the lung tissue if the mucus viscosity decreases too much. Therefore, mucolytic therapies should be optimized to support mucociliary clearance while preventing bacteria dynamics that may lead to predominantly tissue infection.

\begin{figure}
    \centering
    \includegraphics[width=1\textwidth,height=0.76\textheight,keepaspectratio]{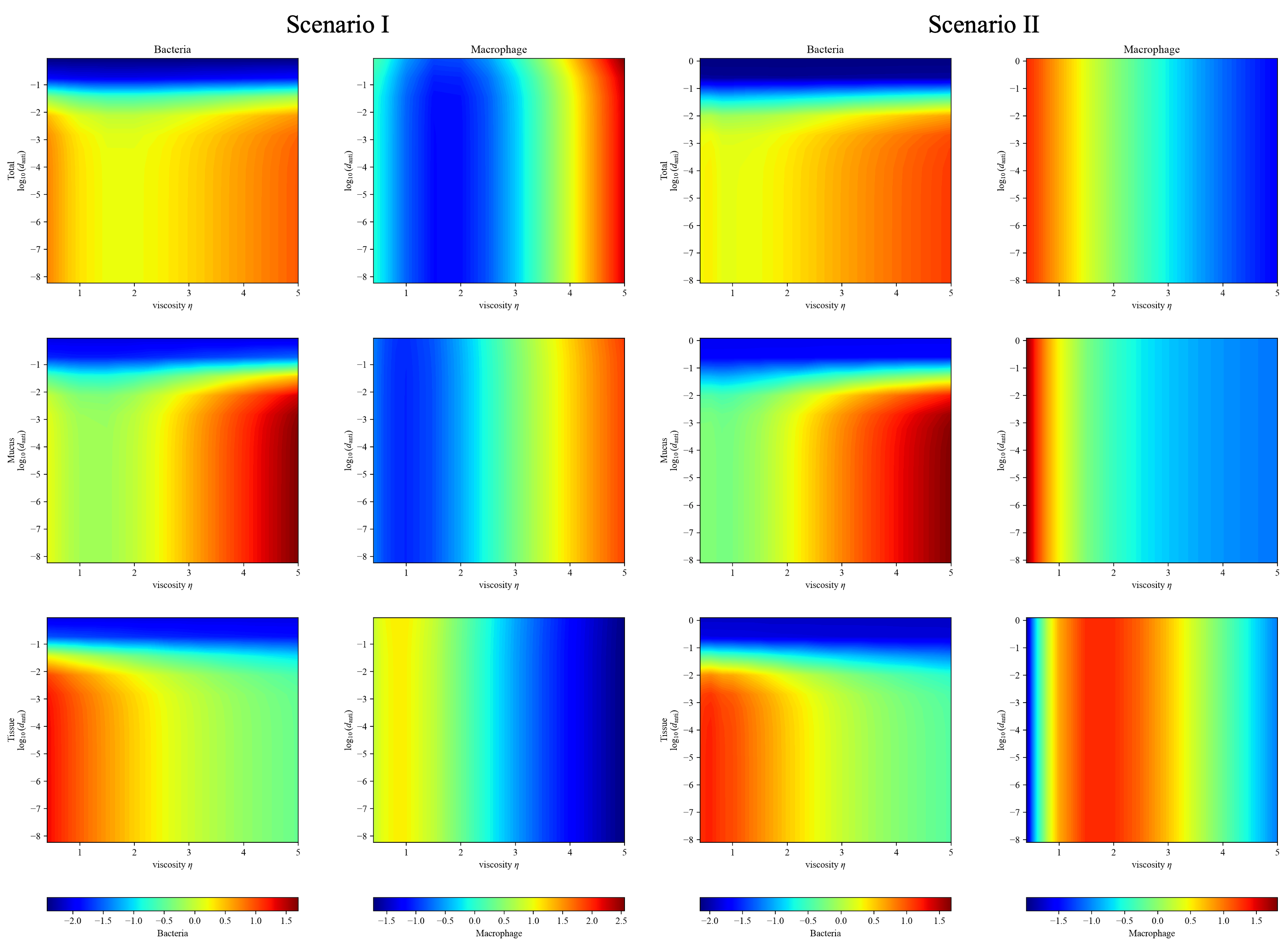}
    \caption{A representative example of a Scenario I simulation with a macrophage retention-dominated response (left) and a representative example of a Scenario II simulation with a recruitment-dominated response (right).}
    \label{fig:Treatment heatmap figure}
\end{figure}

\section{Conclusions}
Non-tuberculous mycobacterial (NTM) infections are an increasingly important clinical challenge for individuals with cystic fibrosis (CF), where impaired mucociliary clearance, abnormal mucus rheology, and dysregulated immune responses create a favorable environment for persistent bacterial colonization and chronic infection \cite{Martiniano2022-NTMandCF}. Despite significant advances in experimental and clinical studies, it remains difficult to explore the relative contributions of bacterial proliferation, immune-cell dynamics, mucus transport, and tissue invasion because these processes occur over multiple spatial and temporal scales and strongly influence one another. Developing predictive mathematical models that integrate these mechanisms is therefore essential for understanding disease progression, identifying the dominant biological processes that govern infection outcomes, and designing more effective therapeutic strategies that could potentially include mucolytics \cite{Sulaiman2025-nonPharmNTM}.

Toward addressing these challenges, we developed a multiscale continuum modeling framework that integrates bacterial dynamics, macrophage-mediated immune responses, mucus transport, and interfacial exchange between the mucus layer and the underlying lung tissue. A key contribution of this paper is the systematic integration of agent-based modeling and partial differential equations. Rather than constructing the continuum model solely from phenomenological assumptions, we calibrated the PDE model using data generated from a previously developed agent-based model, thereby preserving the underlying cellular mechanisms while achieving the computational efficiency needed for sensitivity analysis, uncertainty quantification, and large-scale treatment simulations. This hybrid modeling strategy provides an effective bridge between cellular-scale stochastic interactions and tissue-scale transport dynamics.

The calibrated PDE model successfully reproduces the temporal evolution of bacterial and macrophage populations observed in the ABM across multiple parameter sets. Beyond reproducing the average population dynamics, the PDE framework provides additional spatial information that is difficult to obtain from compartment-based or stochastic models alone. In particular, the model captures the coupled movement of bacteria and immune cells within the mucus layer and their migration into the underlying tissue, allowing us to investigate how physical transport mechanisms interact with immune responses during the early establishment of infection.

The sensitivity analysis demonstrates that infection outcomes are jointly controlled by microbial growth, immune responses, and mucus-associated transport mechanisms. While bacterial proliferation and macrophage recruitment remain important determinants of infection progression, parameters governing mucus viscosity, bacterial diffusivity, macrophage mobility, and interfacial transfer emerge as equally influential regulators of bacterial persistence. These findings suggest that mucociliary clearance is not merely a passive removal mechanism but actively shapes host-pathogen interactions by regulating the spatial redistribution of both bacteria and immune cells. This is in line with other work demonstrating the complex interactions between airway clearance, airway structure, and immune responses in driving NTM disease \cite{Ortega2025-geneticNTM,Park2025-geneticNTMgwas}. Consequently, the progression of pulmonary NTM infection likely depends on the dynamic balance among bacterial replication, immune-mediated clearance, and physical transport within the airway environment.

Our treatment simulations further demonstrate that mucus-targeted interventions may have complex and nonlinear consequences. Reducing mucus viscosity enhances mucociliary transport and facilitates bacterial removal from the mucus layer. However, excessive reduction in mucus viscosity  can also increase bacterial redistribution into the lung tissue, where bacteria become inaccessible to mucociliary clearance and rely more heavily on immune responses and antibiotic treatment for elimination. This spatial redistribution creates competing effects that cannot be captured by non-spatial models. The simulations therefore provide quantitative insights into how mucolytic therapy complements an integrated therapeutic strategy including antibiotics \cite{Burgel2024-CFguidelines}. Combining mucus-thinning agents with sufficiently effective antibacterial treatment may maximize bacterial clearance while minimizing the risk of tissue invasion. These findings also suggest that clinical responses of NTM infections to mucolytic therapies may vary among patients with different immune characteristics and mucus properties.

From a computational perspective, the proposed framework establishes an effective strategy for integrating discrete and continuum models of infectious disease \cite{cilfone2015strategies}. The ABM provides detailed mechanistic information describing stochastic cellular interactions, whereas the PDE model efficiently captures tissue-scale transport and enables systematic exploration of parameter space. The calibration procedure developed here offers a practical methodology for constructing computationally efficient surrogate models while retaining key biological realism. Such hybrid multiscale strategies are broadly applicable to many infectious and inflammatory diseases in which biological processes span multiple spatial and temporal scales \cite{garira2018primer,vodovotz2019agent}.

Several limitations should be acknowledged. The current model focuses primarily on innate immunity and explicitly considers macrophages as the dominant immune cell population. Other important components of the immune response, including neutrophils, dendritic cells, T lymphocytes, cytokine signaling, and adaptive immunity, are not yet represented. The airway geometry is simplified into two coupled compartments and does not account for the highly heterogeneous branching structure of the lung or patient-specific airway morphology. In addition, although the PDE parameters are calibrated against an established ABM, direct calibration using experimental or clinical imaging data will be needed to further strengthen the model. Finally, the present work investigates early infection dynamics and does not explicitly model long-term biofilm formation, bacterial phenotypic adaptation, or chronic tissue remodeling that are characteristic of persistent NTM disease.

These limitations naturally motivate several future directions. The framework can be extended to incorporate additional immune cell populations, inflammatory signaling pathways, and bacterial adaptation mechanisms to study chronic infection and host-pathogen co-evolution. Patient-specific airway geometries reconstructed from CT imaging and individualized mucus properties could enable personalized simulations of disease progression and treatment response. The computational efficiency of the PDE model also makes it well-suited for uncertainty quantification, Bayesian parameter estimation, optimal control, and AI-assisted treatment optimization. More broadly, the hybrid ABM-PDE framework developed in this study provides a general computational paradigm for linking mechanistic cellular models with tissue-scale continuum descriptions. We anticipate that this approach will facilitate the development of predictive digital models for pulmonary infections and provide a quantitative foundation for designing personalized therapeutic strategies.

\section*{Ethics statement}
This study did not involve human participants, human tissue, or live animals. The agent-based model data used for model calibration were obtained from previously published computational studies, and no new ethical approval was required.

\section*{Data and code availability}
The data and code used for the PDE simulations, ABM simulations,
and sensitivity analysis are available on GitHub at
\url{https://github.com/jindong24/ntmsimulation}.

\section*{Author contributions}
J.W., K.K., E.P. and W.H. conceived the study. J.W., P.-C.K., M.C., N.W. and W.H. developed the mathematical model and computational framework. J.W. and K.K. performed the numerical simulations and data analysis. K.K. and E.P. provided the biological interpretation and agent-based model data. E. P. and W.H. supervised the research. J.W., W.H., K.K., and E.P. wrote the initial manuscript draft. All authors contributed to revising the manuscript, approved the final version, and agree to be held accountable for the work.

\section*{Competing interests}
The authors declare that they have no competing interests.

\section*{Acknowledgements}
This material is based upon work supported by the National Science Foundation under Grant No. DMS-1929284 while the authors were in residence at the Institute for Computational and Experimental Research in Mathematics in Providence, RI, during the “Data-driven Predictive Modeling and Simulation for Lung Infections” Collaborate@ICERM program. JW was supported by a Royal Society Newton International Fellowship (NIF$\backslash$R1$\backslash$252881). MC was supported by the National Science Foundation under Grant No. DMS-2602787 and the Army
Research Office award W911NF-23-1-0004. WH was supported by NIGMS 1R35GM146894 and the Huck Chair in AI Mathematical Modeling from Penn State University's Huck Institutes of the Life Sciences.
NW was supported by NSF DMS-2327184.
EP and KK were supported by NIH-NIAID R01-AI172838.

\section*{Use of artificial intelligence}
Generative artificial intelligence tools were used solely to improve the language,
clarity and readability of the manuscript. All scientific content, analyses,
interpretations and conclusions were produced and verified by the authors.

\section*{Disclaimer}
The views expressed in this article are those of the authors and do not necessarily reflect the views of the funding agencies.

\clearpage
\setcounter{section}{0}
\setcounter{equation}{0}
\setcounter{figure}{0}
\setcounter{table}{0}
\renewcommand{\thesection}{S\arabic{section}}
\renewcommand{\theequation}{S\arabic{equation}}
\renewcommand{\thefigure}{S\arabic{figure}}
\renewcommand{\thetable}{S\arabic{table}}
\renewcommand{\theHsection}{supp.\arabic{section}}
\renewcommand{\theHequation}{supp.\arabic{equation}}
\renewcommand{\theHfigure}{supp.\arabic{figure}}
\renewcommand{\theHtable}{supp.\arabic{table}}
\section*{Supplementary Materials}
\addcontentsline{toc}{section}{Supplementary Materials}

\section{Selection of Agent-Based Model (ABM) Simulations for Calibration}
10 runs of the nontuberculous mycobacteria (NTM) airway infection ABM were selected to calibrate the partial differential equation (PDE) model. The 10 runs were chosen to capture the full range of ABM simulation results, with enrichment for edge cases. The plot below shows the final bacteria and macrophage counts for all ABM runs, with the 10 selected runs shown in red dots.

\begin{figure}[H]
    \centering
    \includegraphics[width=1\linewidth,height=0.76\textheight,keepaspectratio]{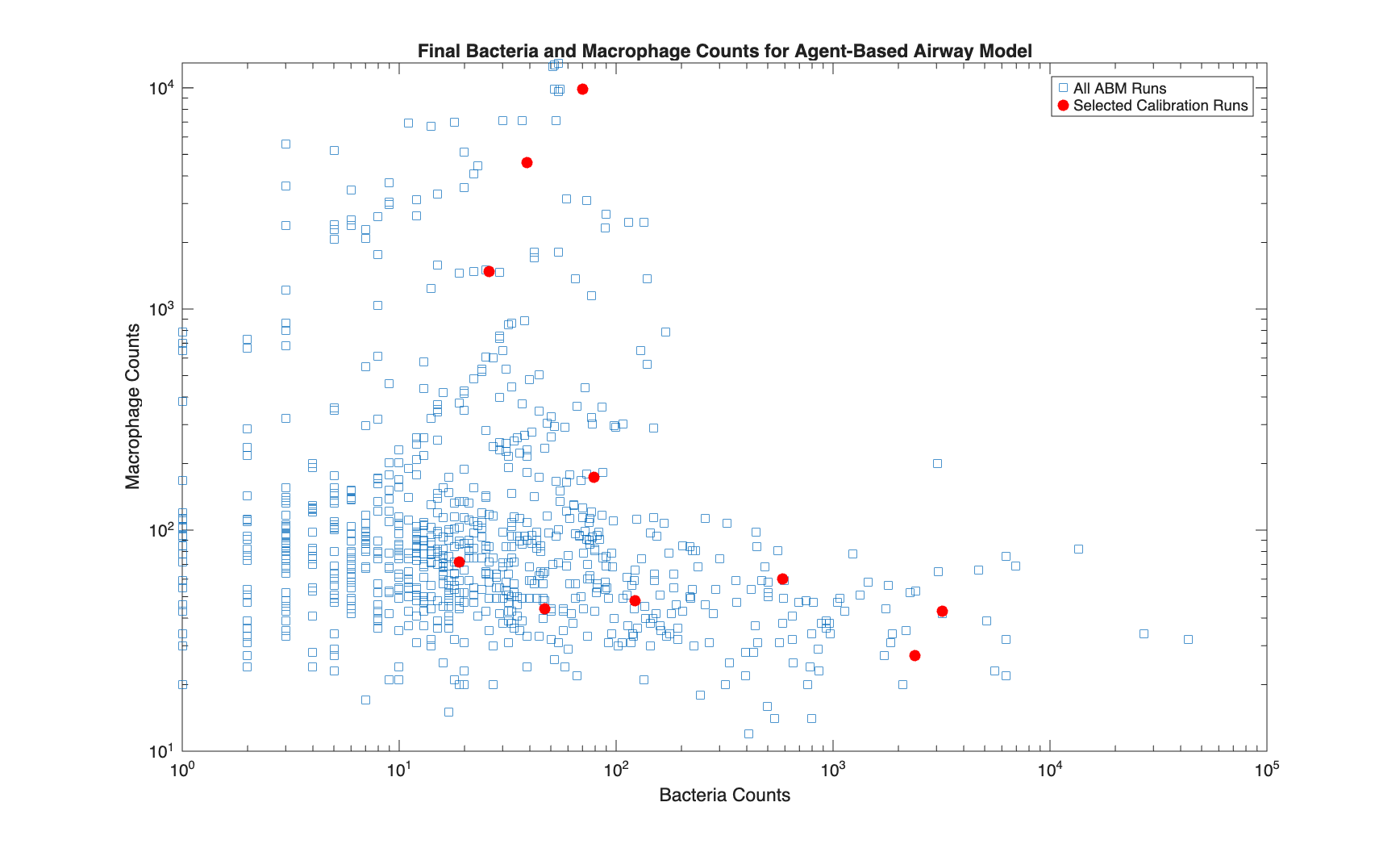}
    \caption{Final bacteria and macrophage counts from the nontuberculous mycobacteria (NTM) airway infection agent-based model. The runs selected for calibration of the partial differential equation model are shown in red.}
    \label{fig:ABM runs scatter plot}
\end{figure}

\section{ABM to Partial Differential Equation (PDE) Model Parameter Conversions}
\subsection{NTM Bacteria Death Rate \texorpdfstring{$d_B$}{dB}}
In the ABM, infected macrophages may kill an intracellular NTM bacterium with a certain probability each 6-minute time step. The following method to convert this bacteria kill probability to a bacteria death rate was adapted from Jones et al. 2017 \cite{jones2017procedure}.
\begin{eqnarray*}
p(t)=1-e^{-rt} \\
1-p(t)=e^{-rt} \\
ln(1-p(t))=-rt \\
r=\frac{-ln(1-p(t)}{t}
\end{eqnarray*}
where p(t) is the probability at time t.

For a 6-minute ABM timestep:

\begin{eqnarray*}
r=\frac{-ln(1-\text{ABM bacteria kill probability})}{\text{6 min}}\times
\frac{\text{60 min}}{\text{1 hr}}\times \frac{\text{24 hr}}{\text{1 day}}\\
=-240ln(1-\text{ABM bacteria kill probability})
\end{eqnarray*}

To test the validity of this parameter conversion, we reran the 10 selected ABM simulations with no bacteria proliferation and compared the bacteria time course outputs with those from an ODE model of bacteria death, shown below.

\begin{equation*}
\frac{dB}{dt}=-240ln(1-\text{ABM bacteria kill probability}) \times B
\end{equation*}

where B is bacteria concentration in $g/cm^3$.

Figure~\ref{fig:ode bac death plots} shows the ODE and ABM bacteria time courses for each of the 10 runs, demonstrating that the ODE model recapitulates the ABM results and supporting the validity of the bacteria death parameter conversion.

\begin{figure}
    \centering
    \includegraphics[width=0.8\linewidth,height=0.76\textheight,keepaspectratio]{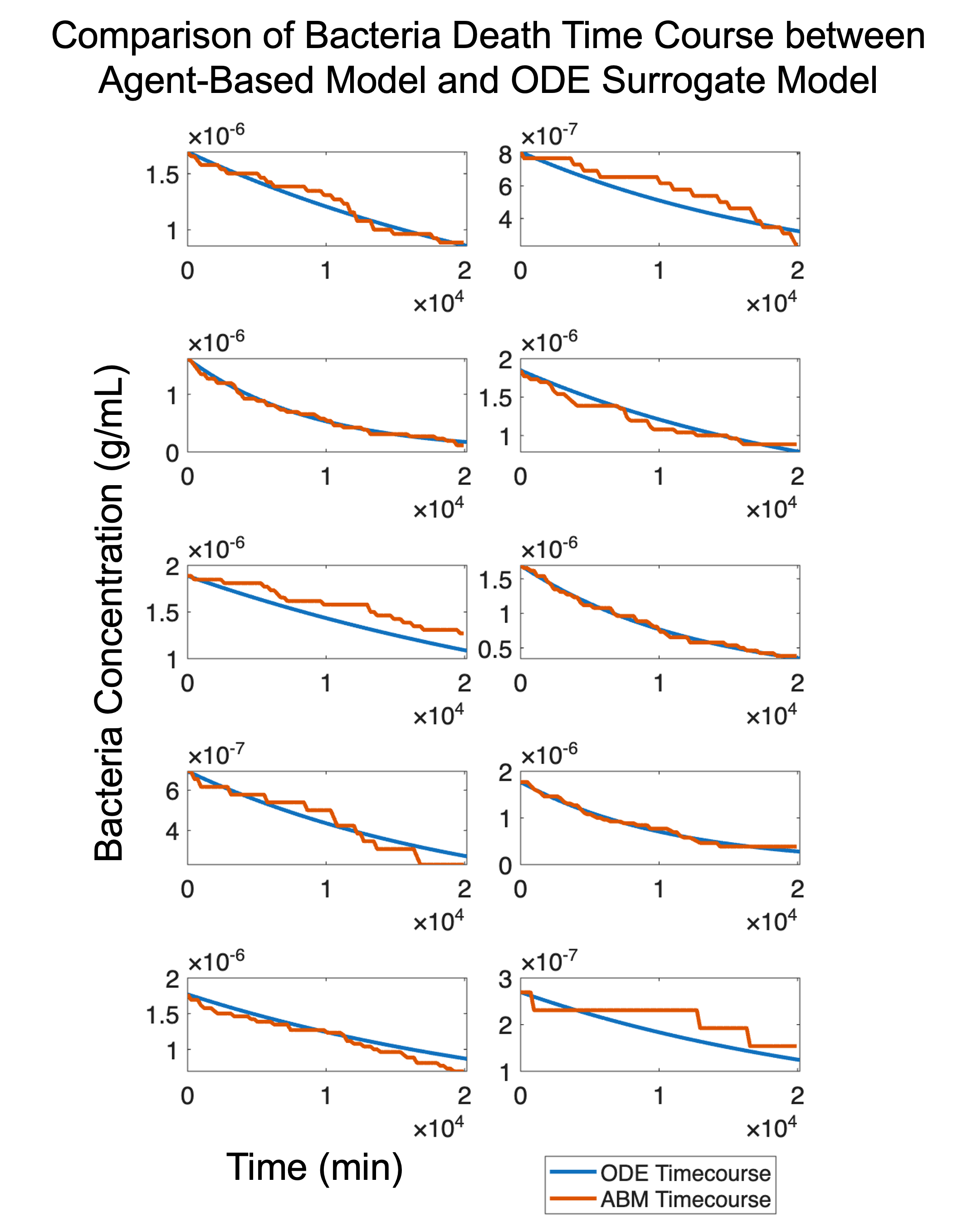}
    \caption{Bacteria concentration time courses between the ordinary differential equation bacteria death model and the adjusted agent-based model with no bacteria replication.}
    \label{fig:ode bac death plots}
\end{figure}

\subsection{Homeostatic Macrophage Recruitment \texorpdfstring{$\alpha_0$}{alpha0}}
In non-infection conditions, the number of newly recruited macrophages must counteract the number of macrophages that undergo apoptosis. Consider the following steady-state case

\begin{equation*}
\frac{dM}{dt}=r_h-d_MM=0
\end{equation*}
where $r_h$ is the macrophage recruitment at homeostasis and $d_M$ is the macrophage death rate. The number of newly recruited macrophages must counteract the concentration of macrophages lost. If a concentration of $M_0$ macrophages is initialized, then

{\small
\begin{equation*}
d_M M_0 =\frac{\text{(Number of newly recruited macrophages)}\text{(Macrophage density)}\text{(Macrophage volume)}}{\text{ABM Simulation Volume}}
\end{equation*}
}

\begin{equation*}
\text{Recruited macrophages per day} =\frac{d_M M_0 \text{(ABM Simulation Volume)}}{\text{(Macrophage density)}\text{(Macrophage volume)}}
\end{equation*}

In order to translate ABM cell agent data to the PDE model, we assume a well-mixed ABM simulation space. In the ABM, new macrophages are recruited at the bottom layer of the ABM simulation environment. Calculating the recruited macrophages per volume at the bottom layer of the ABM environment,

\begin{equation*}
\frac{d_M M_0 \text{(ABM Simulation Volume)}}{\text{(Macrophage density)}\text{(Macrophage volume)}\text{(Simulation Bottom Layer Volume)}}
\end{equation*}

After converting to units of concentration of macrophages,

{\small
\begin{equation*}
\begin{aligned}
&\frac{d_M M_0 \text{(ABM Simulation Volume)}}{\text{(Macrophage density)}\text{(Macrophage volume)}\text{(Simulation Bottom Layer Volume)}}\\
&\qquad \times (\text{Macrophage Density})(\text{Macrophage Volume})
\\&= \frac{d_M M_0 \text{(ABM Simulation Volume)}}{\text{(Simulation Bottom Layer Volume)}}
\\ &= \frac{d_M M_0 (\text{0.2} \times \text{0.2} \times \text{0.006 cm}^3)}{ \text{0.2} \times \text{0.2} \times \text{0.002 cm}^3}
\\ &= 3 \cdot d_M M_0 \text{ g/cm}^3 \cdot \text{day}
\end{aligned}
\end{equation*}
}

Assuming well-mixed conditions, this is equal to the concentration of macrophages per unit length of the ABM simulation environment, yielding the final conversion

\begin{equation*}
3 \cdot d_M M_0 \text{ g/cm}^4 \cdot \text{day}
\end{equation*}

\section{Units of model variables and parameters}
\begin{table}[ht]
\centering
\small
\caption{Variables and parameters in the PDE model and their units.}
\label{tab:pde_units}
\setlength{\tabcolsep}{4pt}
\renewcommand{\arraystretch}{1.05}
\begin{tabular}{lll}
\hline
\textbf{Symbol} & \textbf{Description} & \textbf{Unit} \\
\hline

$\bx$
& Spatial coordinate
& $\mathrm{cm}$ \\

$t$
& Time
& $\mathrm{day}$ \\

$B(\bx,t)$
& Bacterial density
& $\mathrm{g\,cm^{-3}}$ \\

$M(\bx,t)$
& Macrophage density
& $\mathrm{g\,cm^{-3}}$ \\

$m(\bx,t)$
& Mucus density
& $\mathrm{g\,cm^{-3}}$ \\

$\bm u(\bx,t)$
& Advection velocity
& $\mathrm{cm\,day^{-1}}$ \\

\hline

$\lambda_B$
& Bacterial growth rate
& $\mathrm{day^{-1}}$ \\

$K_B$
& Bacterial carrying capacity
& $\mathrm{g\,cm^{-3}}$ \\

$K_{Bi}$
& Half-saturation bacterial density
& $\mathrm{g\,cm^{-3}}$ \\

$d_B$
& Macrophage-mediated bacterial clearance rate
& $\mathrm{cm^3\,g^{-1}\,day^{-1}}$ \\

$d_{Mb}$
& Maximum macrophage burst rate
& $\mathrm{day^{-1}}$ \\

$d_M$
& Macrophage death rate
& $\mathrm{day^{-1}}$ \\

$d_m$
& Mucus degradation rate
& $\mathrm{day^{-1}}$ \\

\hline

$D_B( D_1,D_2)$
& Bacterial diffusion coefficients
& $\mathrm{cm^2\,day^{-1}}$ \\

$D_M$
& Macrophage diffusion coefficient
& $\mathrm{cm^2\,day^{-1}}$ \\

$D_m$
& Mucus diffusion coefficient
& $\mathrm{cm^2\,day^{-1}}$ \\

$\delta$
& Chemotaxis coefficient
& $\mathrm{cm^5\,g^{-1}\,day^{-1}}$ \\

\hline

$\gamma_B$
& Bacterial outflow coefficient
& $\mathrm{cm\,day^{-1}}$ \\

$\gamma_M$
& Macrophage outflow coefficient
& $\mathrm{cm\,day^{-1}}$ \\

$\gamma_m$
& Mucus outflow coefficient
& $\mathrm{cm\,day^{-1}}$ \\

$\tau_B$
& Bacterial interface permeability coefficient
& $\mathrm{cm^{-1}}$ \\

$\tau_M$
& Macrophage interface permeability coefficient
& $\mathrm{cm^{-1}}$ \\

$M_{\rm ref}$
& Reference macrophage density
& $\mathrm{g\,cm^{-3}}$ \\

$m_0$
& Reference mucus density
& $\mathrm{g\,cm^{-3}}$ \\

\hline

$\alpha(B)$
& Macrophage recruitment coefficient
& $\mathrm{cm\,day^{-1}}$ \\

$\alpha_0$
& Baseline macrophage recruitment coefficient
& $\mathrm{cm\,day^{-1}}$ \\

$\alpha_1$
& Bacteria-dependent recruitment coefficient
& $\mathrm{cm^4\,g^{-1}\,day^{-1}}$ \\

$\beta(MB)$
& Mucus production coefficient
& $\mathrm{cm\,day^{-1}}$ \\

$\beta_0$
& Baseline mucus production coefficient
& $\mathrm{cm\,day^{-1}}$ \\

$\beta_1$
& Inflammation-dependent mucus production coefficient
& $\mathrm{cm^7\,g^{-2}\,day^{-1}}$ \\

\hline

$\eta$
&  viscosity factor
& Dimensionless \\

$\eta_0$
& Reference viscosity
& Dimensionless \\

$d_{\rm anti}$
& Treatment-induced bacterial clearance rate
& $\mathrm{day^{-1}}$ \\

$\zeta$
& Treatment-strength parameter
& Dimensionless \\

\hline
\end{tabular}
\end{table}
\section{Calibration parameters}\label{apdx:calibration}
The model parameters selected for calibration are summarized in Table~\ref{tab:calibration-params}.
\begin{table}[H]
\centering
\begin{tabular}{ccccccc}
\hline
$D_B$ & $\lambda_B$ & $K_B$ & $K_{Bi}$ & $d_{Mb}$ & $d_B$ & $\tau_B$ \\
\hline
$D_M$ & $\delta$ & $d_M$ & $\tau_M$ & $M_0$ & $\alpha_0$ & $\alpha_1$ \\
\hline

\end{tabular}
\caption{List of model parameters subject to calibration.}
\label{tab:calibration-params}
\end{table}

\section{Evolution Results Based on Calibration}

This section presents the system evolution for different parameter sets. As in the main text, the results compare the ABM-derived data with the calibrated PDE simulations in the mucus region and the corresponding PDE predictions in the tissue region. The spatial evolution of the PDE system is shown at selected time points, both with and without the mucus velocity field, in Figures~\ref{fig:appendix-set048-velocity-comparison}--\ref{fig:appendix-set843-velocity-comparison}.

\begin{figure}[H]
    \centering
    \begin{subfigure}{0.78\textwidth}
        \centering
        \includegraphics[width=\linewidth,height=0.76\textheight,keepaspectratio]{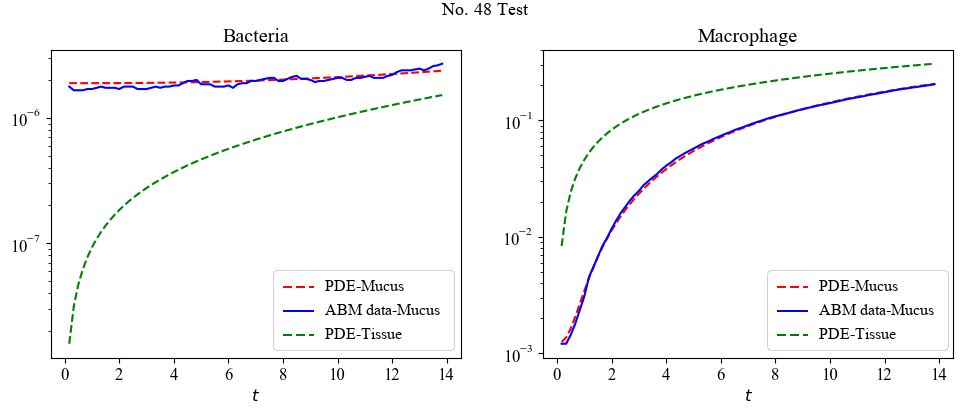}
        \caption{Concentration fitting.}
        \label{fig:appendix-set048-calibration}
    \end{subfigure}

    \vspace{0.5em}
    \begin{subfigure}{0.30\textwidth}
        \centering
        \includegraphics[width=\linewidth,height=0.76\textheight,keepaspectratio]{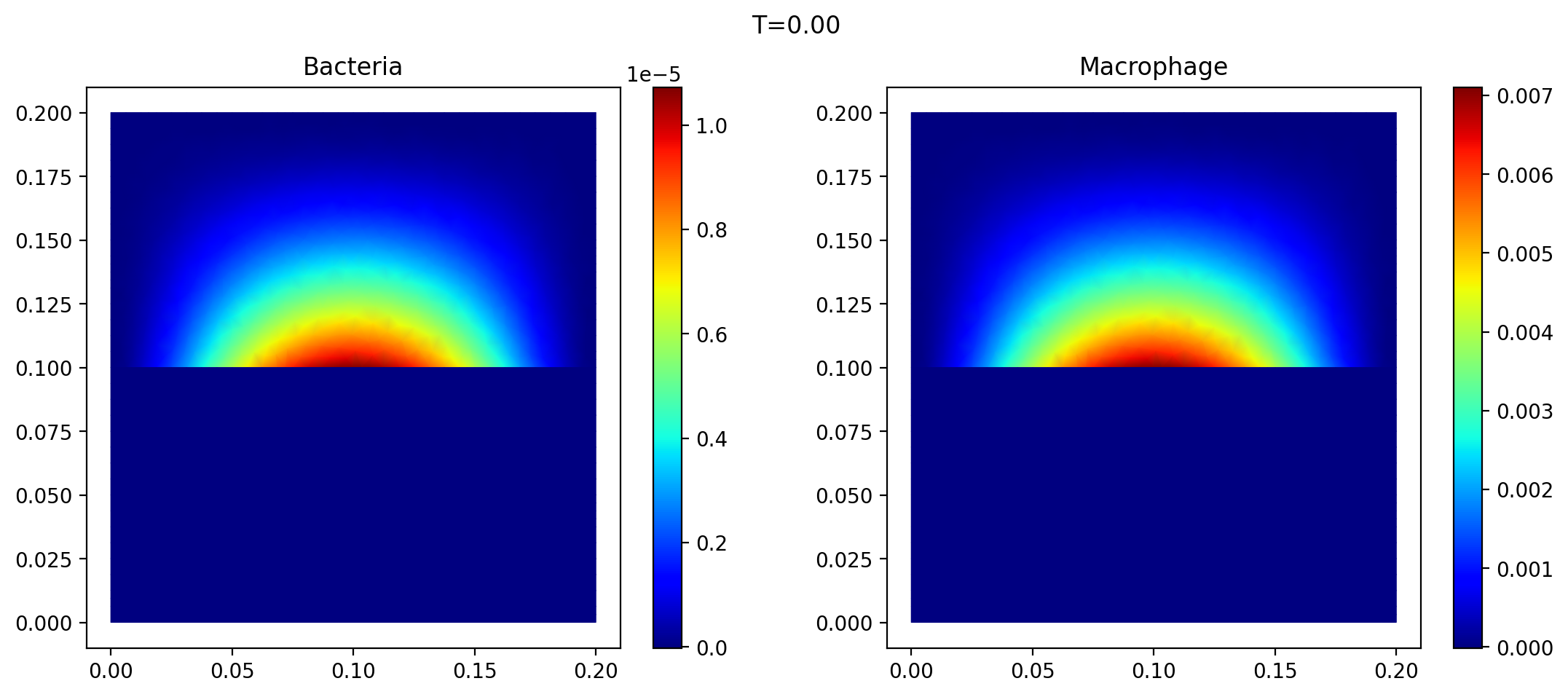}
        \caption{No velocity, $t=0$ day.}
        \label{fig:appendix-set048-no-velocity-t0}
    \end{subfigure}
    \hfill
    \begin{subfigure}{0.30\textwidth}
        \centering
        \includegraphics[width=\linewidth,height=0.76\textheight,keepaspectratio]{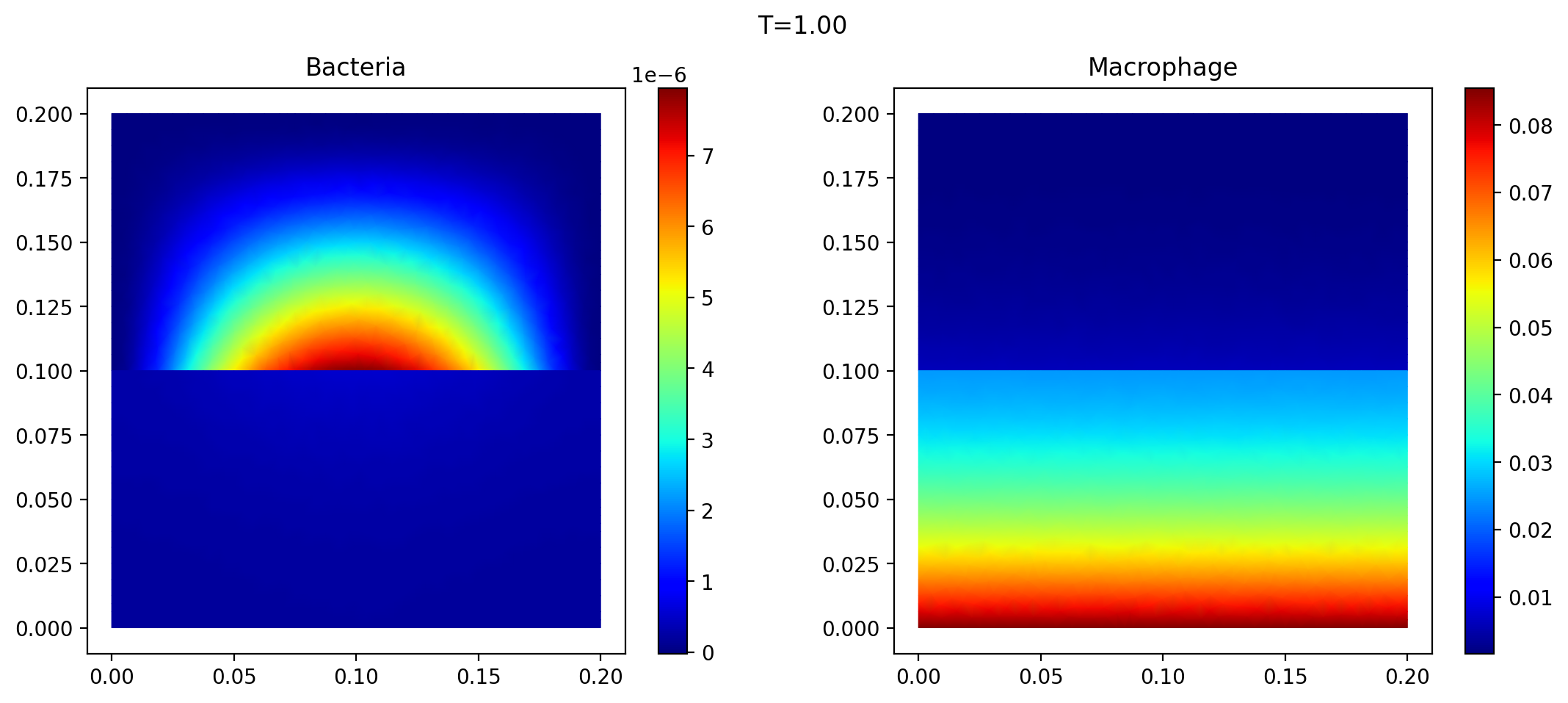}
        \caption{No velocity, $t=1$ day.}
        \label{fig:appendix-set048-no-velocity-t1}
    \end{subfigure}
    \hfill
    \begin{subfigure}{0.30\textwidth}
        \centering
        \includegraphics[width=\linewidth,height=0.76\textheight,keepaspectratio]{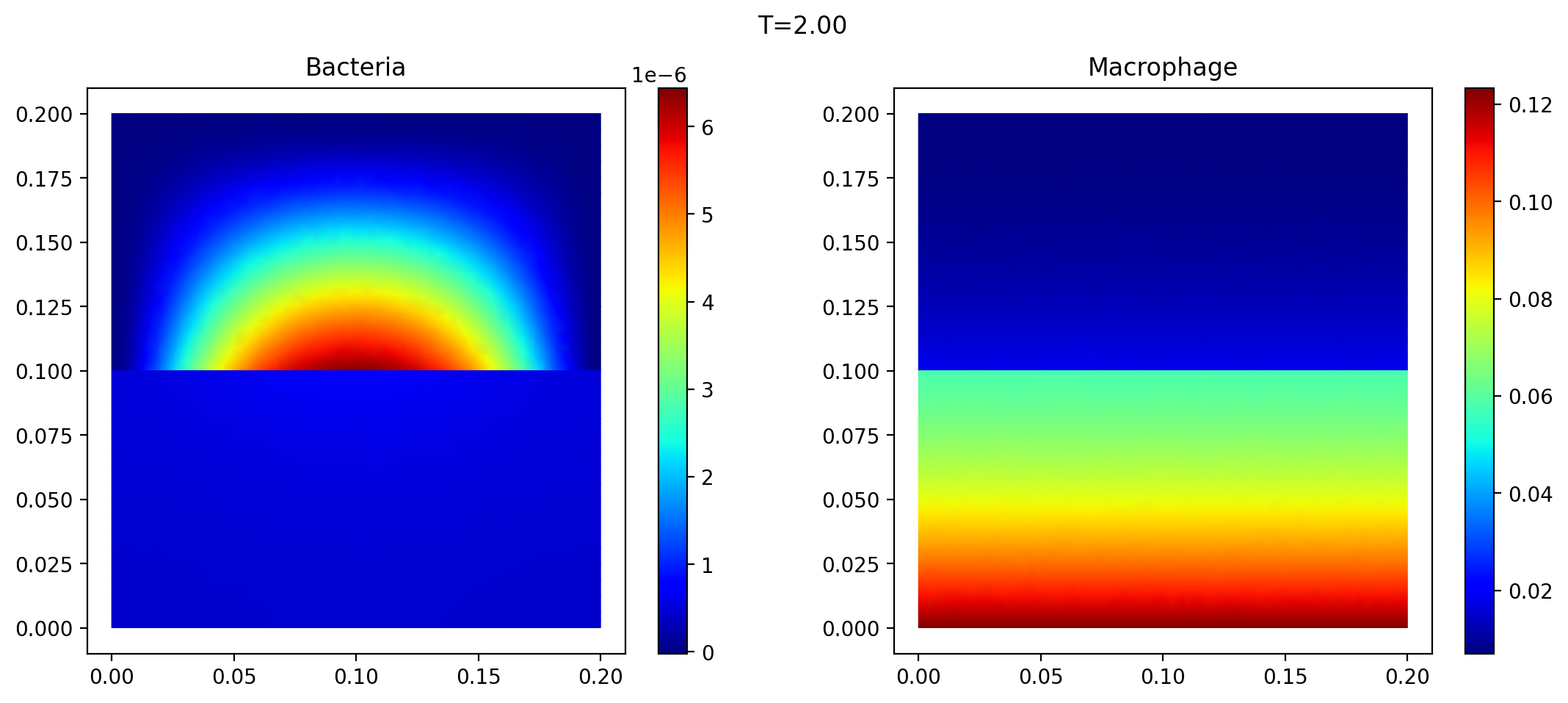}
        \caption{No velocity, $t=2$ day.}
        \label{fig:appendix-set048-no-velocity-t2}
    \end{subfigure}

    \vspace{0.5em}
    \begin{subfigure}{0.30\textwidth}
        \centering
        \includegraphics[width=\linewidth,height=0.76\textheight,keepaspectratio]{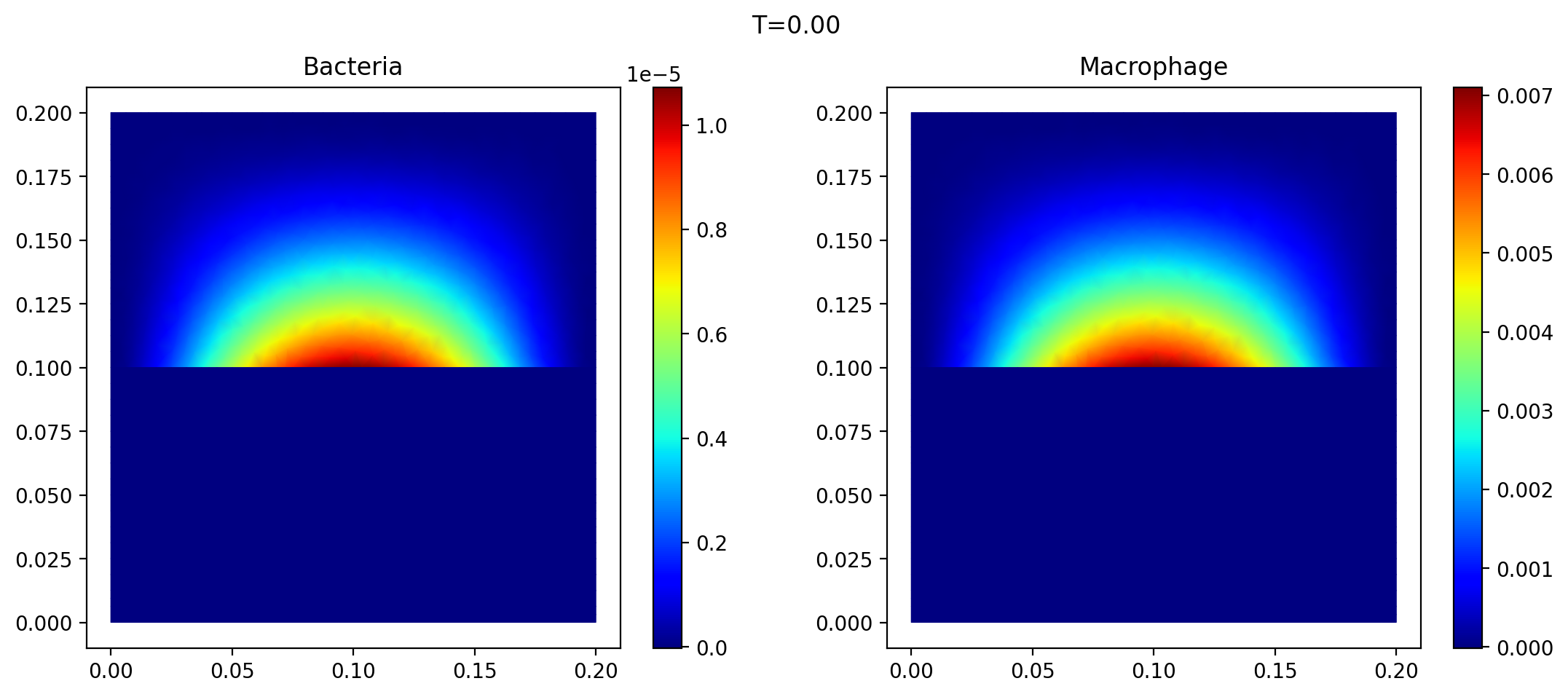}
        \caption{With velocity, $t=0$ day.}
        \label{fig:appendix-set048-velocity-t0}
    \end{subfigure}
    \hfill
    \begin{subfigure}{0.30\textwidth}
        \centering
        \includegraphics[width=\linewidth,height=0.76\textheight,keepaspectratio]{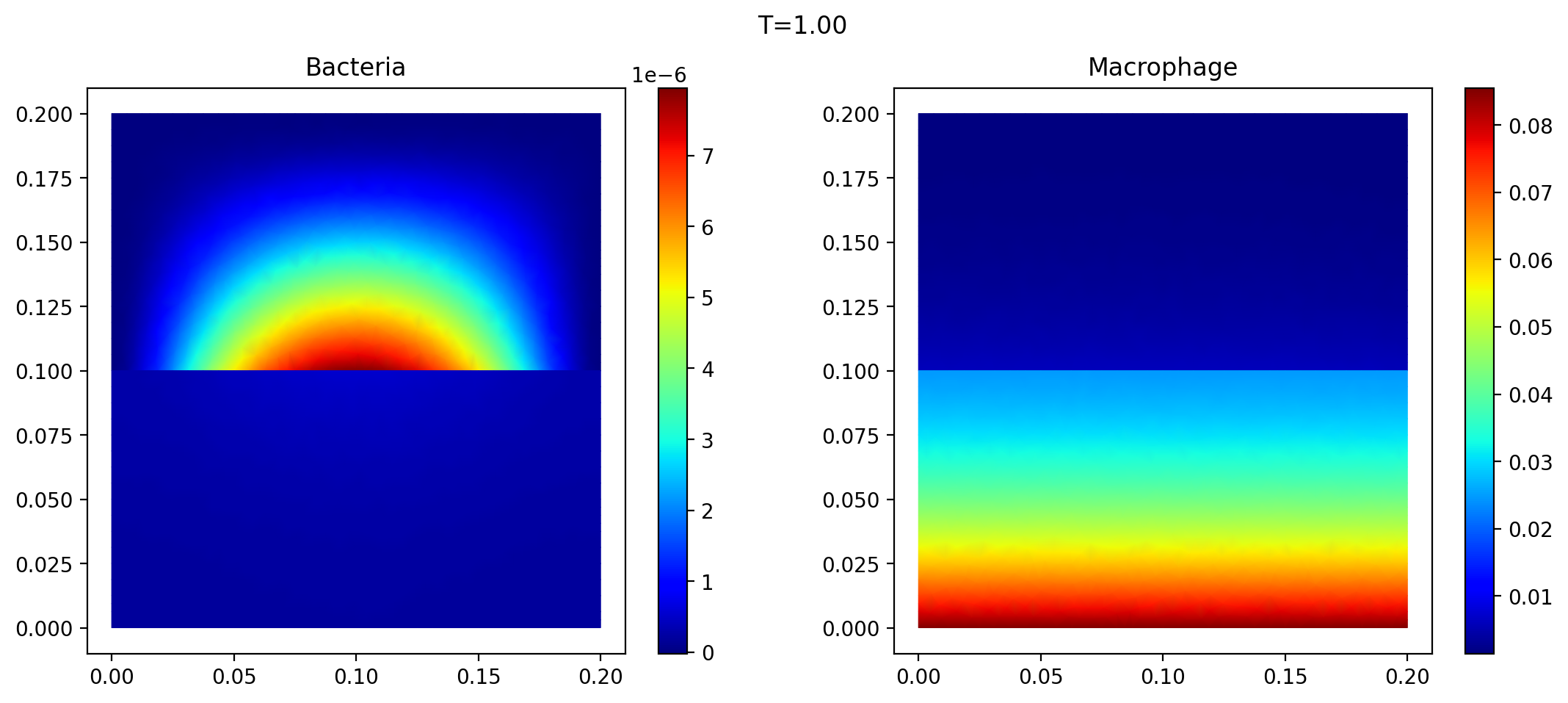}
        \caption{With velocity, $t=1$ day.}
        \label{fig:appendix-set048-velocity-t1}
    \end{subfigure}
    \hfill
    \begin{subfigure}{0.30\textwidth}
        \centering
        \includegraphics[width=\linewidth,height=0.76\textheight,keepaspectratio]{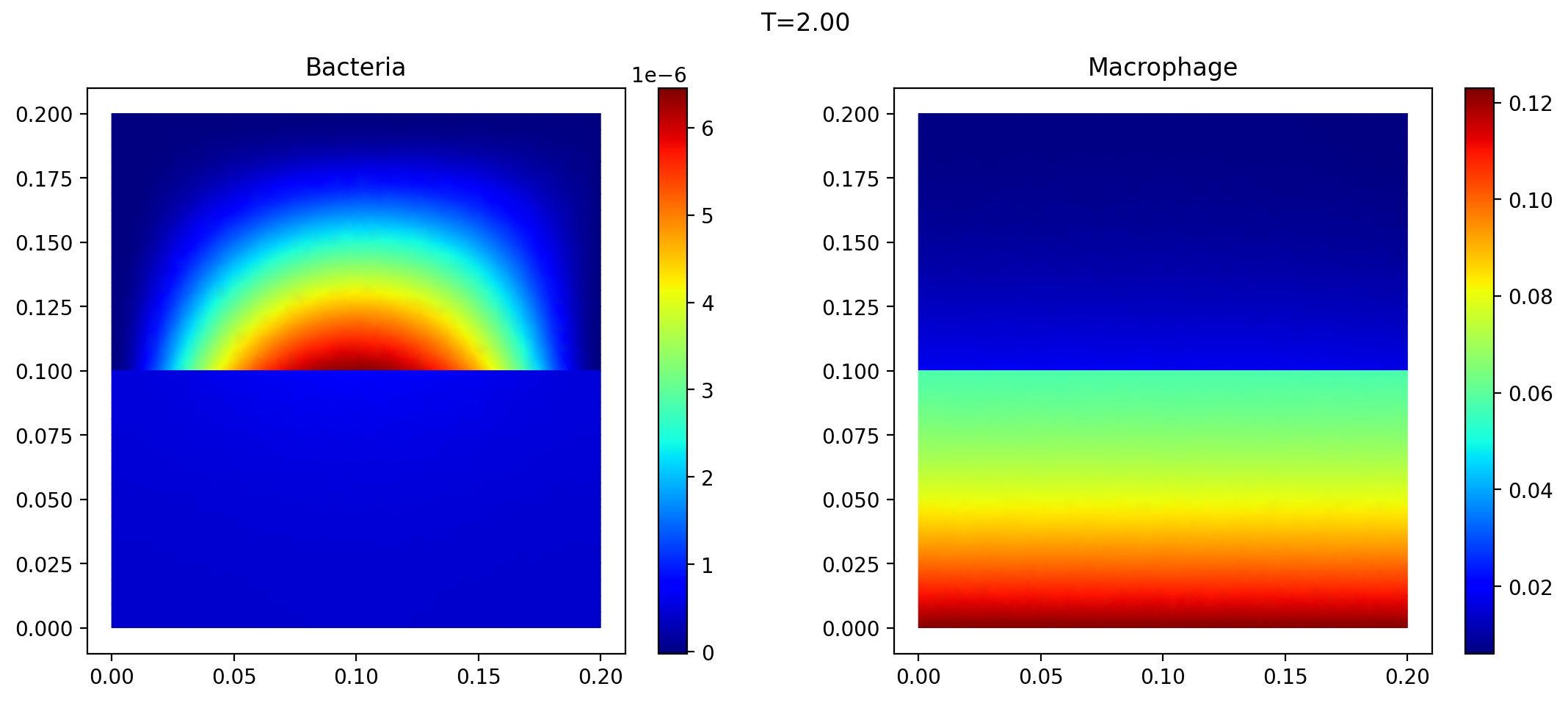}
        \caption{With velocity, $t= 2$ day.}
        \label{fig:appendix-set048-velocity-t2}
    \end{subfigure}

    \caption{
    Evolution of the system for parameter set No.~48, with and without the mucus velocity field.
    (a) Concentration fitting for bacteria and macrophages, with calibration reference data, mucus-region simulation, and tissue-region simulation shown in each plot.
    (b--d) Spatial distributions without mucus velocity field driven at $t=0$, $1$, and $2$ day.
    (e--g) Spatial distributions with the mucus velocity field driven at the same time points.
    Each spatial panel shows bacteria on the left and macrophages on the right.
    }
    \label{fig:appendix-set048-velocity-comparison}
\end{figure}
 \clearpage

\begin{figure}[H]
    \centering
    \begin{subfigure}{0.78\textwidth}
        \centering
        \includegraphics[width=\linewidth,height=0.76\textheight,keepaspectratio]{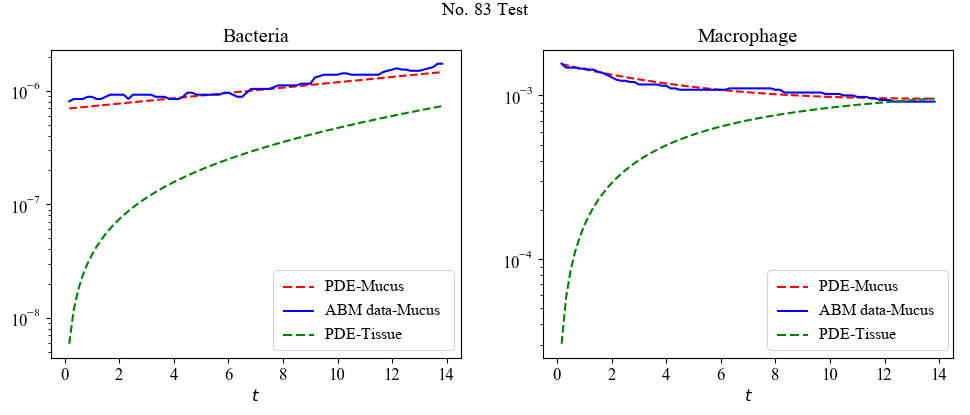}
        \caption{Concentration fitting.}
        \label{fig:appendix-set083-calibration}
    \end{subfigure}

    \vspace{0.5em}
    \begin{subfigure}{0.30\textwidth}
        \centering
        \includegraphics[width=\linewidth,height=0.76\textheight,keepaspectratio]{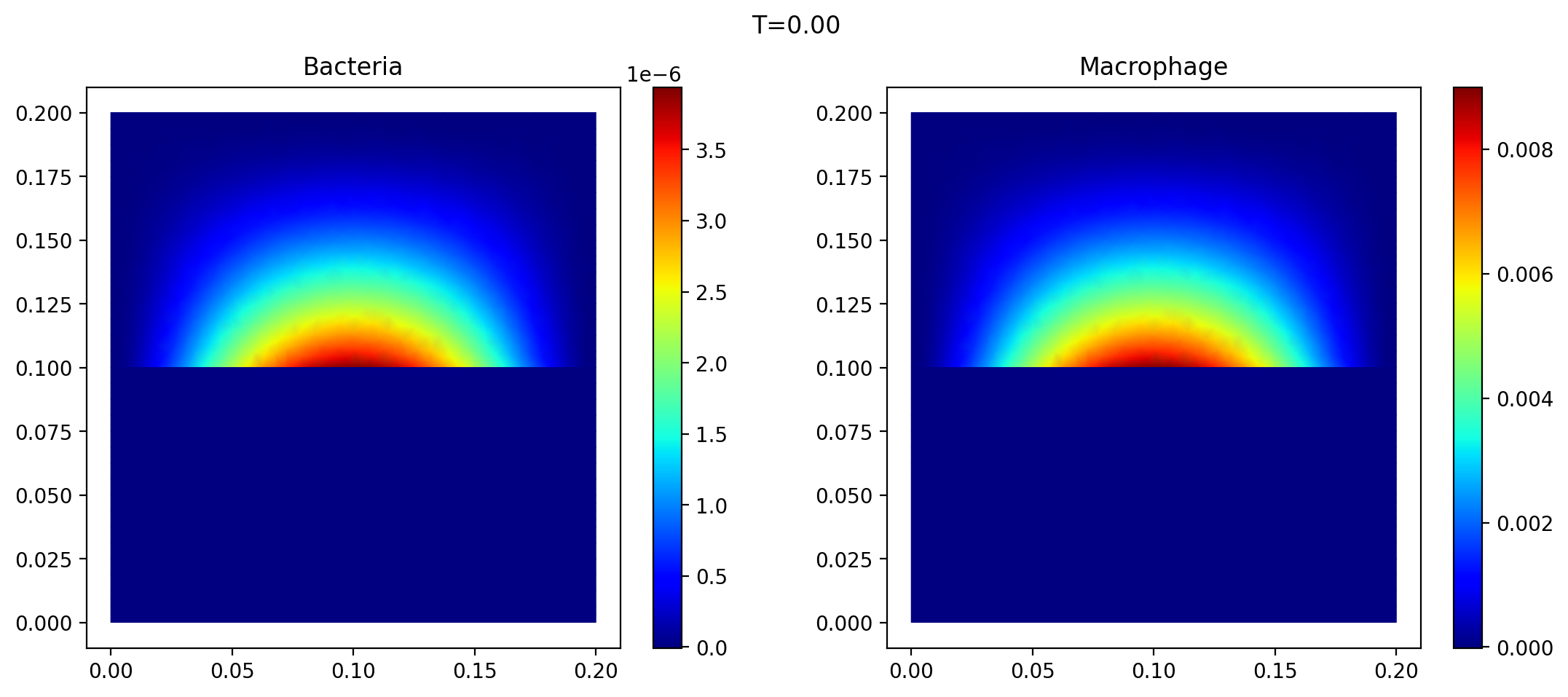}
        \caption{No velocity, $t=0$ day.}
        \label{fig:appendix-set083-no-velocity-t0}
    \end{subfigure}
    \hfill
    \begin{subfigure}{0.30\textwidth}
        \centering
        \includegraphics[width=\linewidth,height=0.76\textheight,keepaspectratio]{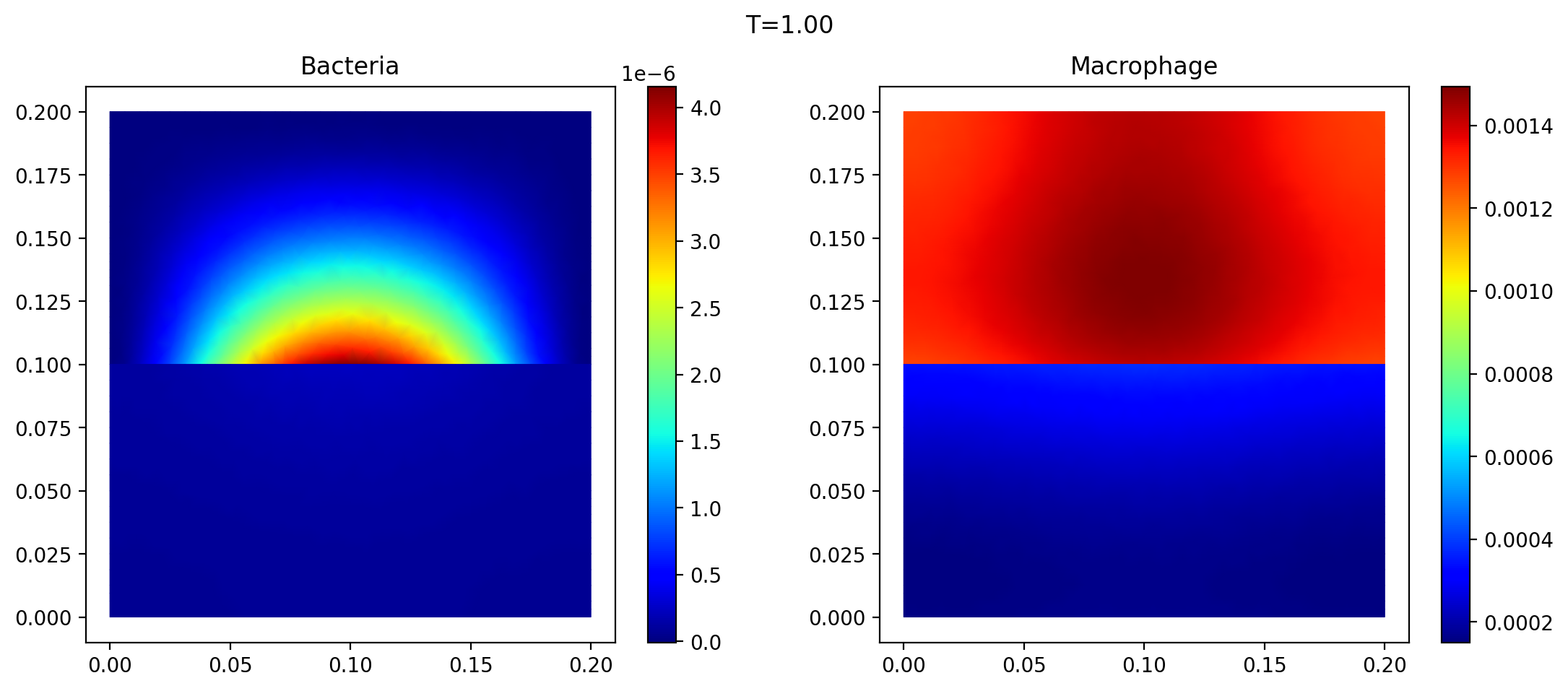}
        \caption{No velocity, $t=1$ day.}
        \label{fig:appendix-set083-no-velocity-t1}
    \end{subfigure}
    \hfill
    \begin{subfigure}{0.30\textwidth}
        \centering
        \includegraphics[width=\linewidth,height=0.76\textheight,keepaspectratio]{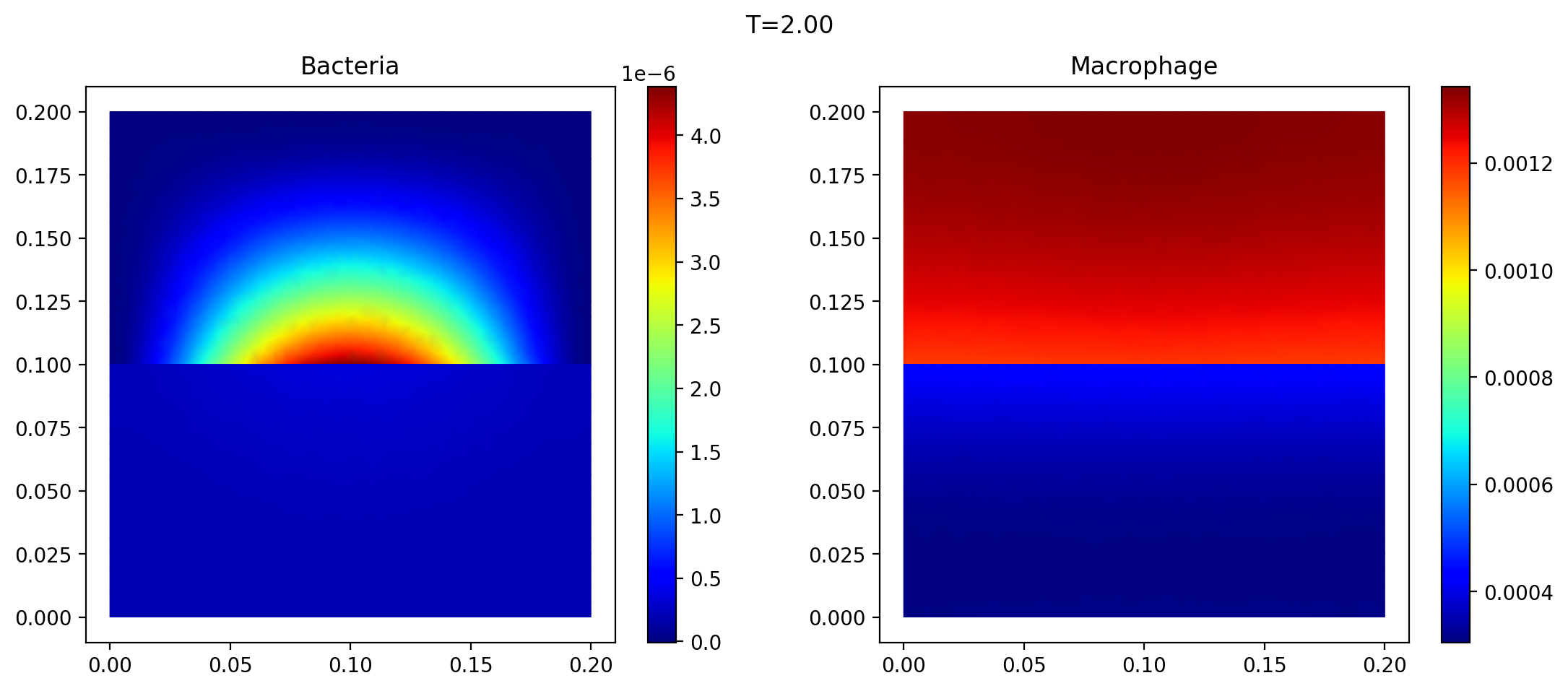}
        \caption{No velocity, $t=2$ day.}
        \label{fig:appendix-set083-no-velocity-t2}
    \end{subfigure}

    \vspace{0.5em}
    \begin{subfigure}{0.30\textwidth}
        \centering
        \includegraphics[width=\linewidth,height=0.76\textheight,keepaspectratio]{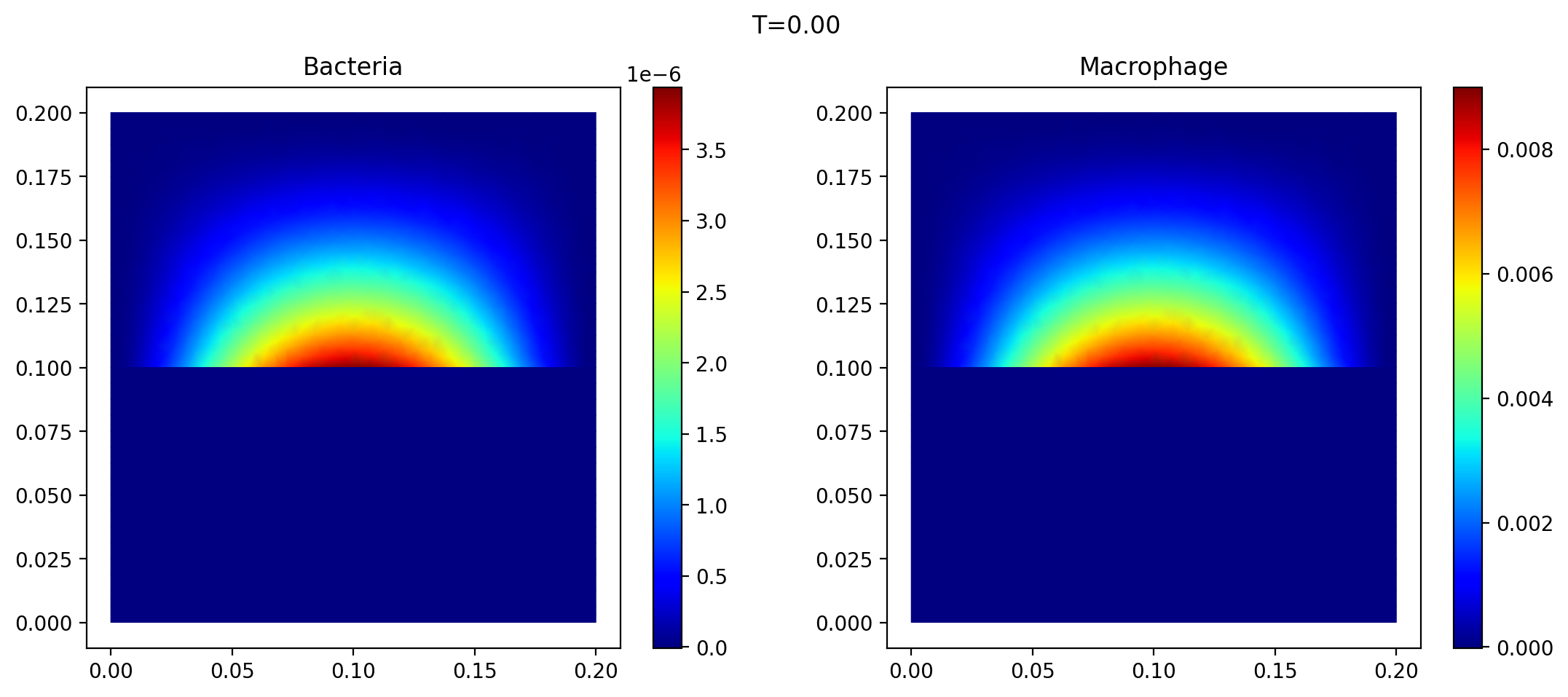}
        \caption{With velocity, $t=0$ day.}
        \label{fig:appendix-set083-velocity-t0}
    \end{subfigure}
    \hfill
    \begin{subfigure}{0.30\textwidth}
        \centering
        \includegraphics[width=\linewidth,height=0.76\textheight,keepaspectratio]{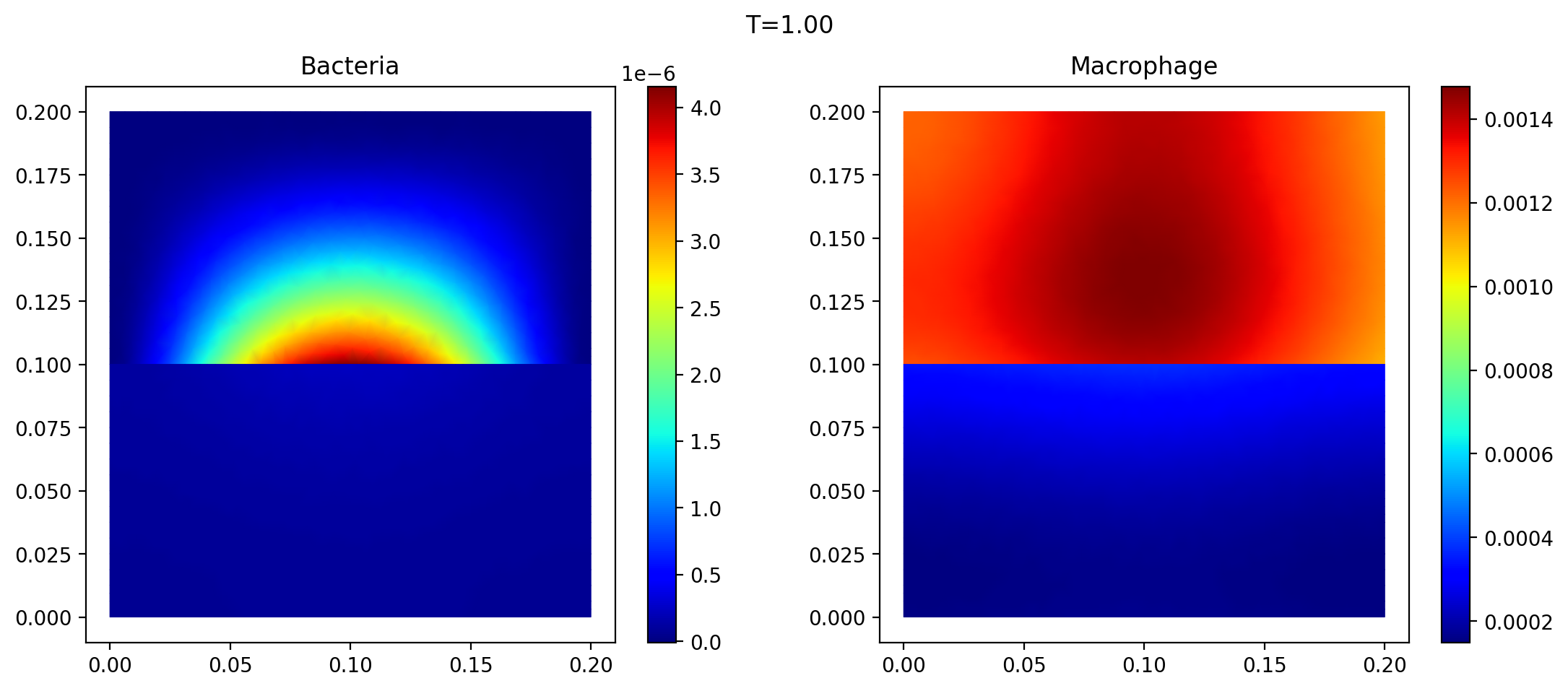}
        \caption{With velocity, $t=1$ day.}
        \label{fig:appendix-set083-velocity-t1}
    \end{subfigure}
    \hfill
    \begin{subfigure}{0.30\textwidth}
        \centering
        \includegraphics[width=\linewidth,height=0.76\textheight,keepaspectratio]{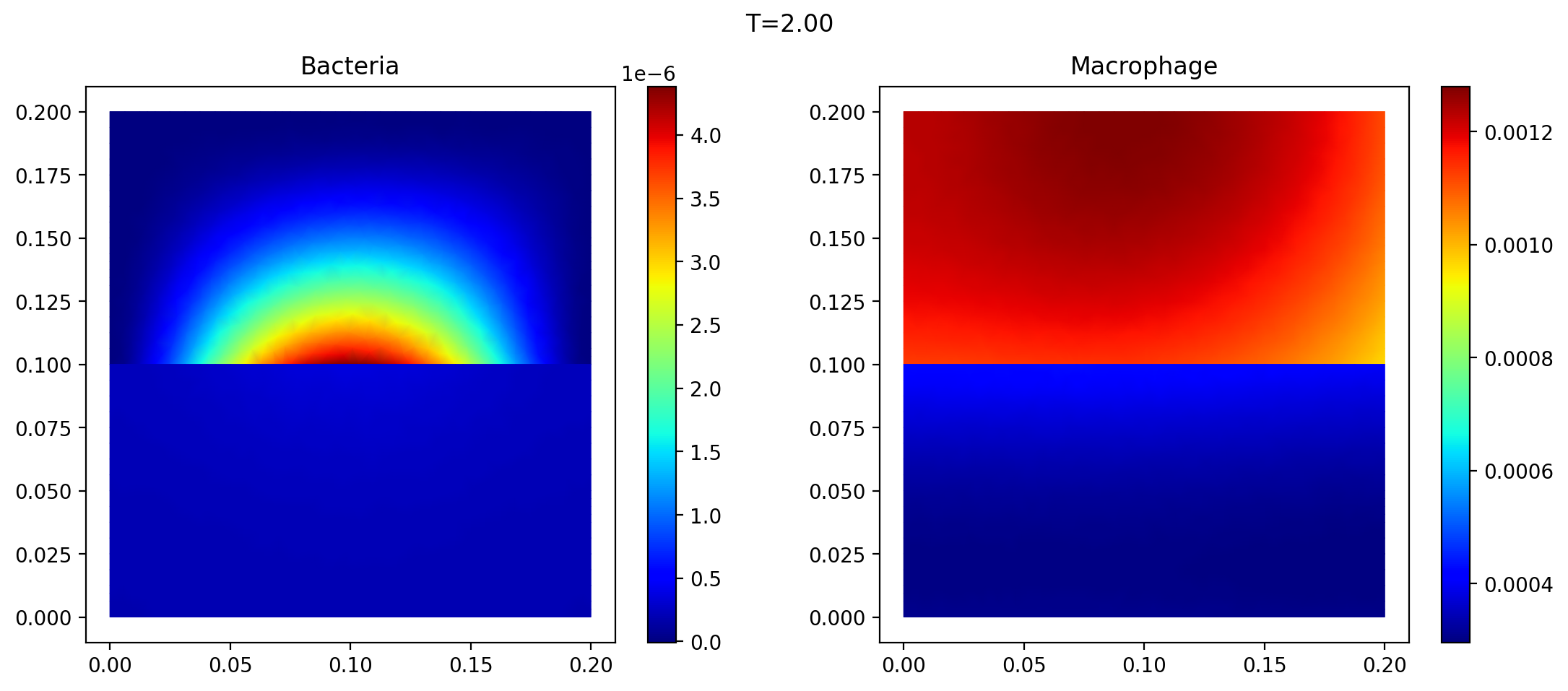}
        \caption{With velocity, $t= 2$ day.}
        \label{fig:appendix-set083-velocity-t2}
    \end{subfigure}

    \caption{
    Evolution of the system for parameter set No.~83, with and without the mucus velocity field.
    (a) Concentration fitting for bacteria and macrophages, with calibration reference data, mucus-region simulation, and tissue-region simulation shown in each plot.
    (b--d) Spatial distributions without mucus velocity field driven at $t=0$, $1$, and $2$ day.
    (e--g) Spatial distributions with the mucus velocity field driven at the same time points.
    Each spatial panel shows bacteria on the left and macrophages on the right.
    }
    \label{fig:appendix-set083-velocity-comparison}
\end{figure}
 \clearpage

\begin{figure}[H]
    \centering
    \begin{subfigure}{0.78\textwidth}
        \centering
        \includegraphics[width=\linewidth,height=0.76\textheight,keepaspectratio]{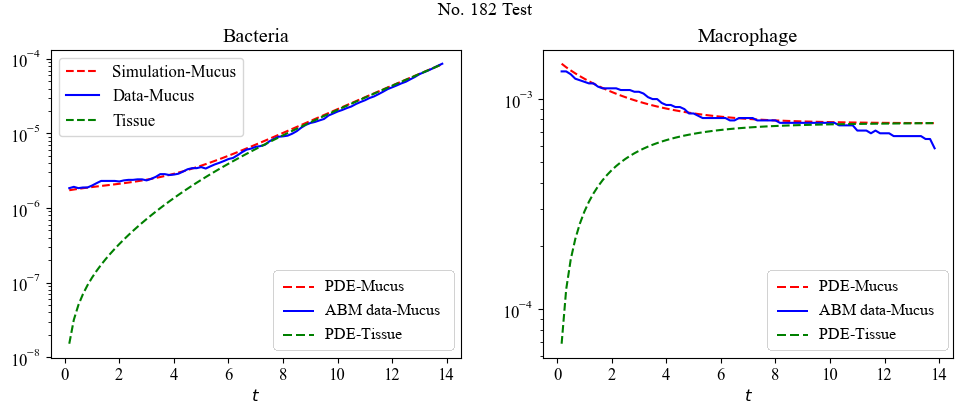}
        \caption{Concentration fitting.}
        \label{fig:appendix-set182-calibration}
    \end{subfigure}

    \vspace{0.5em}
    \begin{subfigure}{0.30\textwidth}
        \centering
        \includegraphics[width=\linewidth,height=0.76\textheight,keepaspectratio]{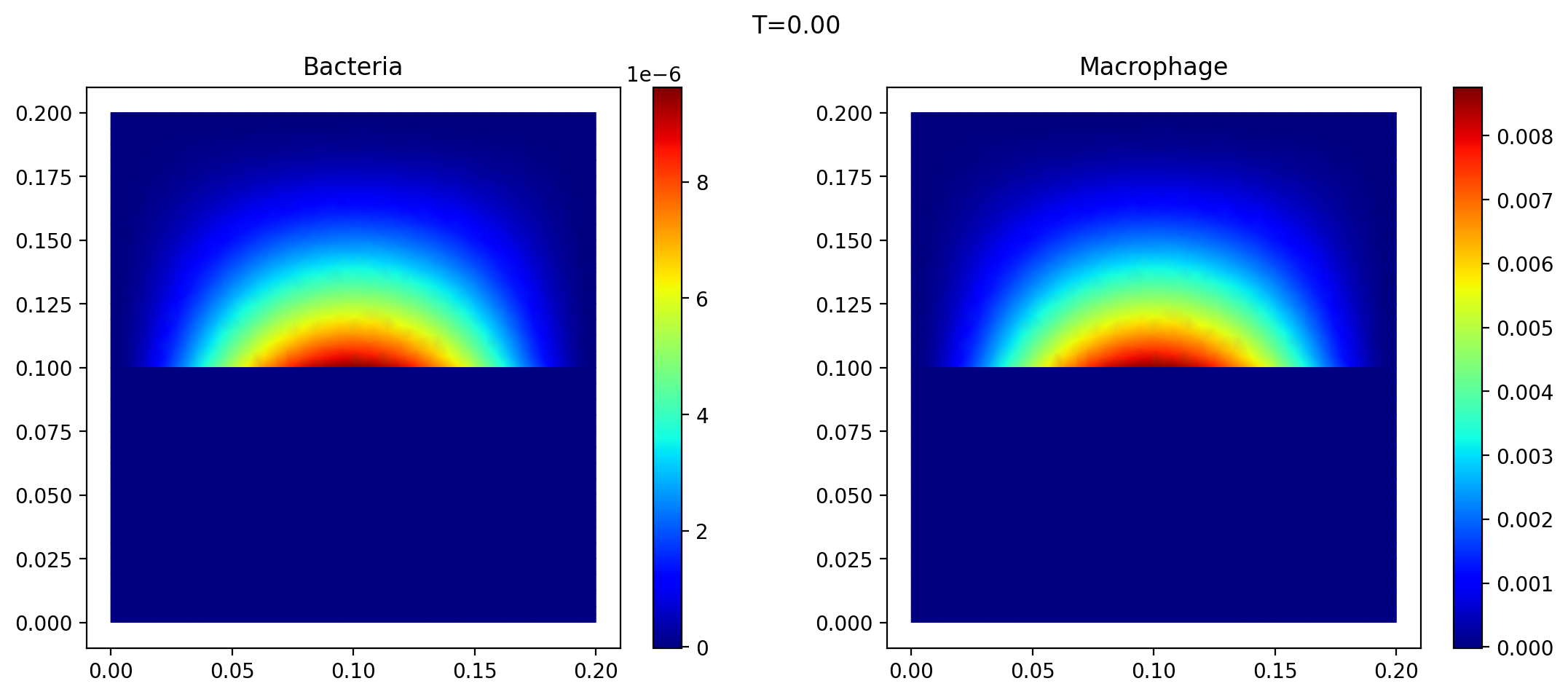}
        \caption{No velocity, $t=0$ day.}
        \label{fig:appendix-set182-no-velocity-t0}
    \end{subfigure}
    \hfill
    \begin{subfigure}{0.30\textwidth}
        \centering
        \includegraphics[width=\linewidth,height=0.76\textheight,keepaspectratio]{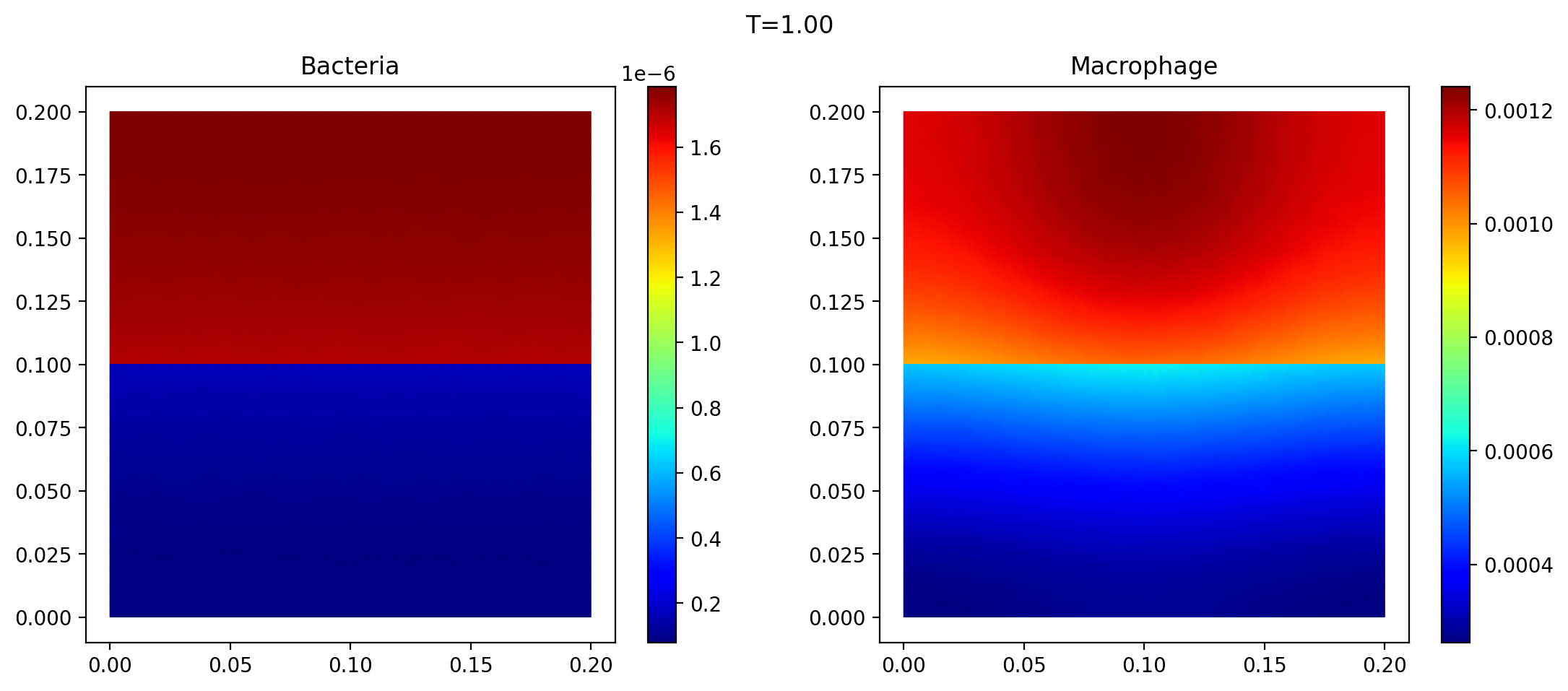}
        \caption{No velocity, $t=1$ day.}
        \label{fig:appendix-set182-no-velocity-t1}
    \end{subfigure}
    \hfill
    \begin{subfigure}{0.30\textwidth}
        \centering
        \includegraphics[width=\linewidth,height=0.76\textheight,keepaspectratio]{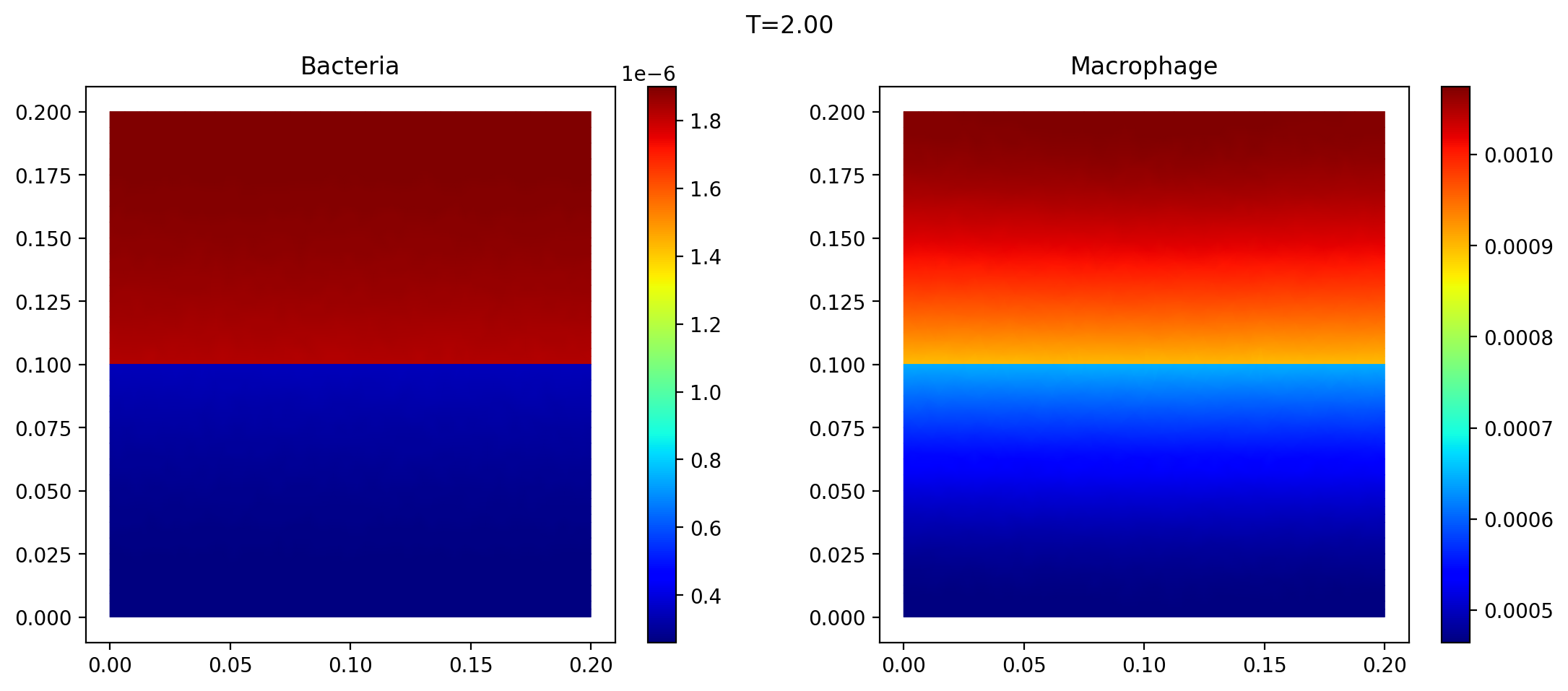}
        \caption{No velocity, $t=2$ day.}
        \label{fig:appendix-set182-no-velocity-t2}
    \end{subfigure}

    \vspace{0.5em}
    \begin{subfigure}{0.30\textwidth}
        \centering
        \includegraphics[width=\linewidth,height=0.76\textheight,keepaspectratio]{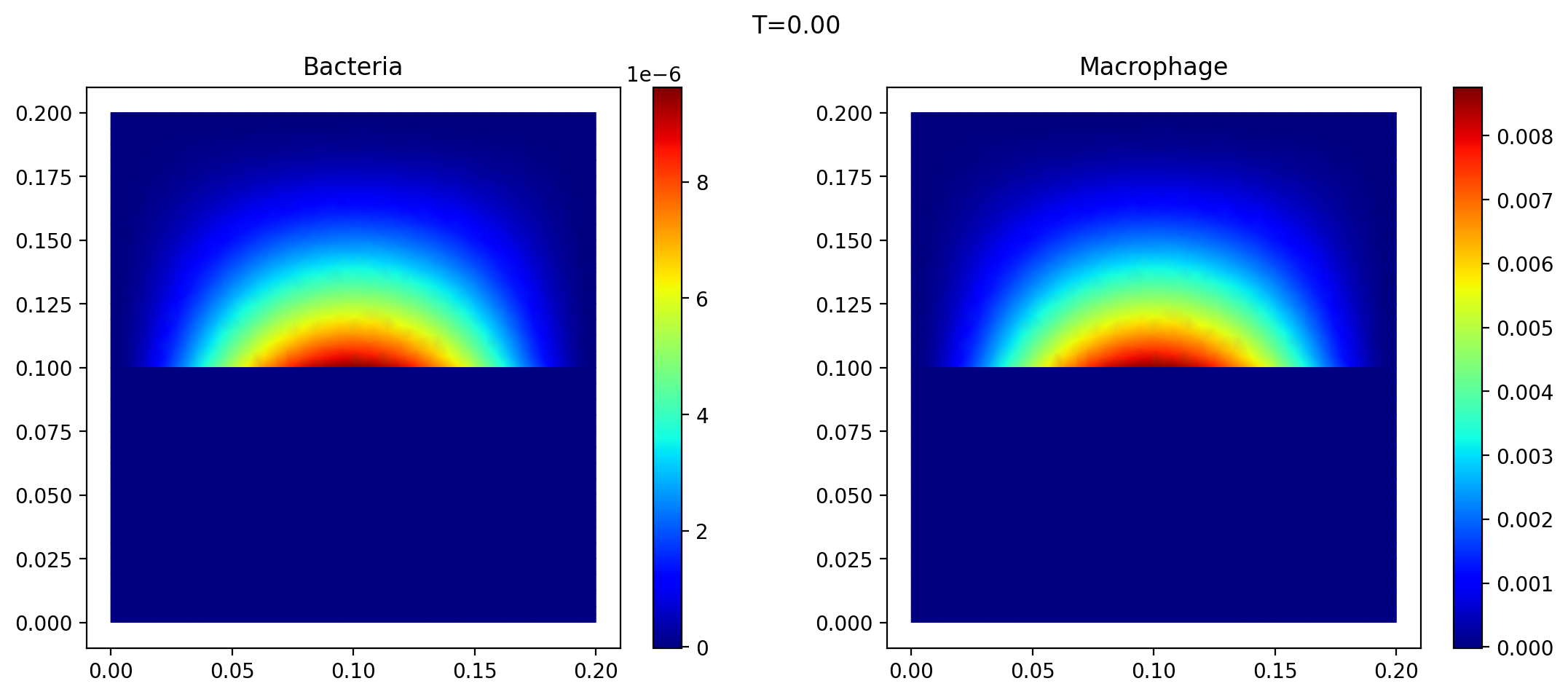}
        \caption{With velocity, $t=0$ day.}
        \label{fig:appendix-set182-velocity-t0}
    \end{subfigure}
    \hfill
    \begin{subfigure}{0.30\textwidth}
        \centering
        \includegraphics[width=\linewidth,height=0.76\textheight,keepaspectratio]{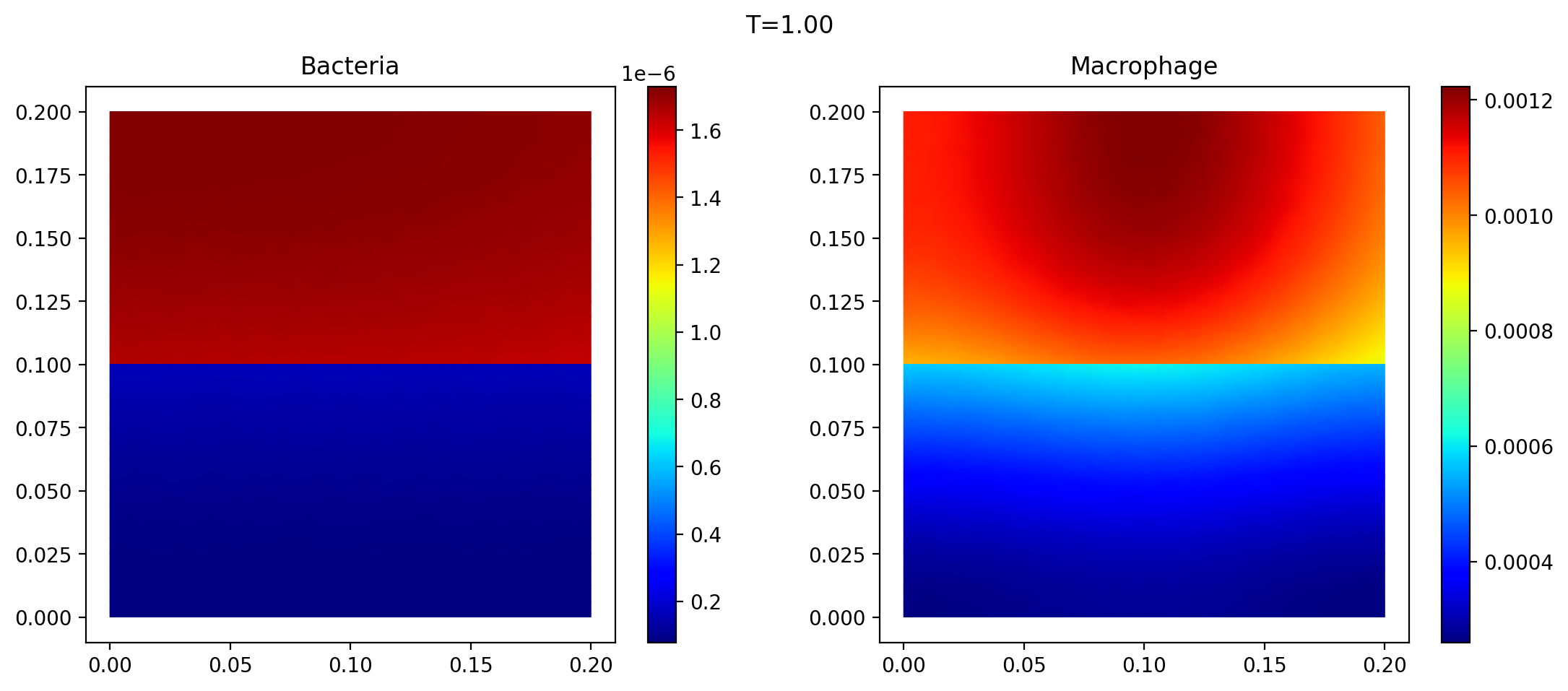}
        \caption{With velocity, $t=1$ day.}
        \label{fig:appendix-set182-velocity-t1}
    \end{subfigure}
    \hfill
    \begin{subfigure}{0.30\textwidth}
        \centering
        \includegraphics[width=\linewidth,height=0.76\textheight,keepaspectratio]{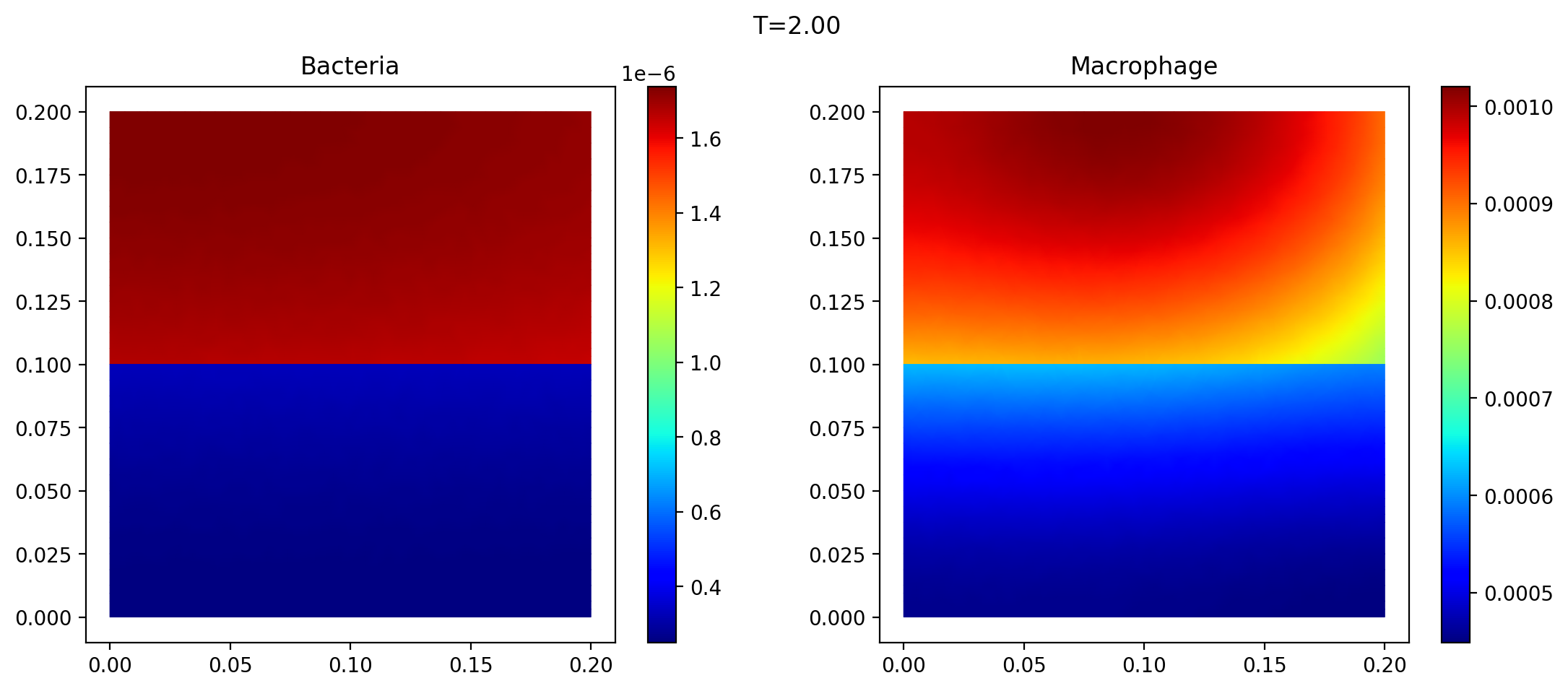}
        \caption{With velocity, $t= 2$ day.}
        \label{fig:appendix-set182-velocity-t2}
    \end{subfigure}

    \caption{
    Evolution of the system for parameter set No.~182, with and without the mucus velocity field.
    (a) Concentration fitting for bacteria and macrophages, with calibration reference data, mucus-region simulation, and tissue-region simulation shown in each plot.
    (b--d) Spatial distributions without mucus velocity field driven at $t=0$, $1$, and $2$ day.
    (e--g) Spatial distributions with the mucus velocity field driven at the same time points.
    Each spatial panel shows bacteria on the left and macrophages on the right.
    }
    \label{fig:appendix-set182-velocity-comparison}
\end{figure}
 \clearpage

\begin{figure}[H]
    \centering
    \begin{subfigure}{0.78\textwidth}
        \centering
        \includegraphics[width=\linewidth,height=0.76\textheight,keepaspectratio]{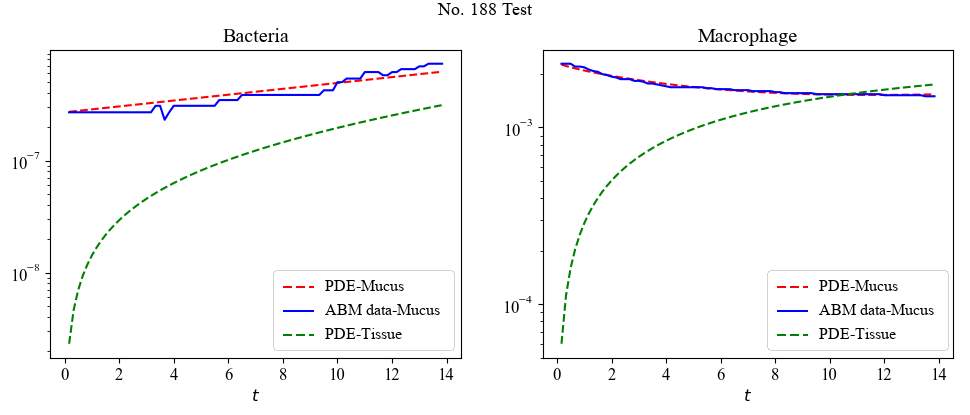}
        \caption{Concentration fitting.}
        \label{fig:appendix-set188-calibration}
    \end{subfigure}

    \vspace{0.5em}
    \begin{subfigure}{0.30\textwidth}
        \centering
        \includegraphics[width=\linewidth,height=0.76\textheight,keepaspectratio]{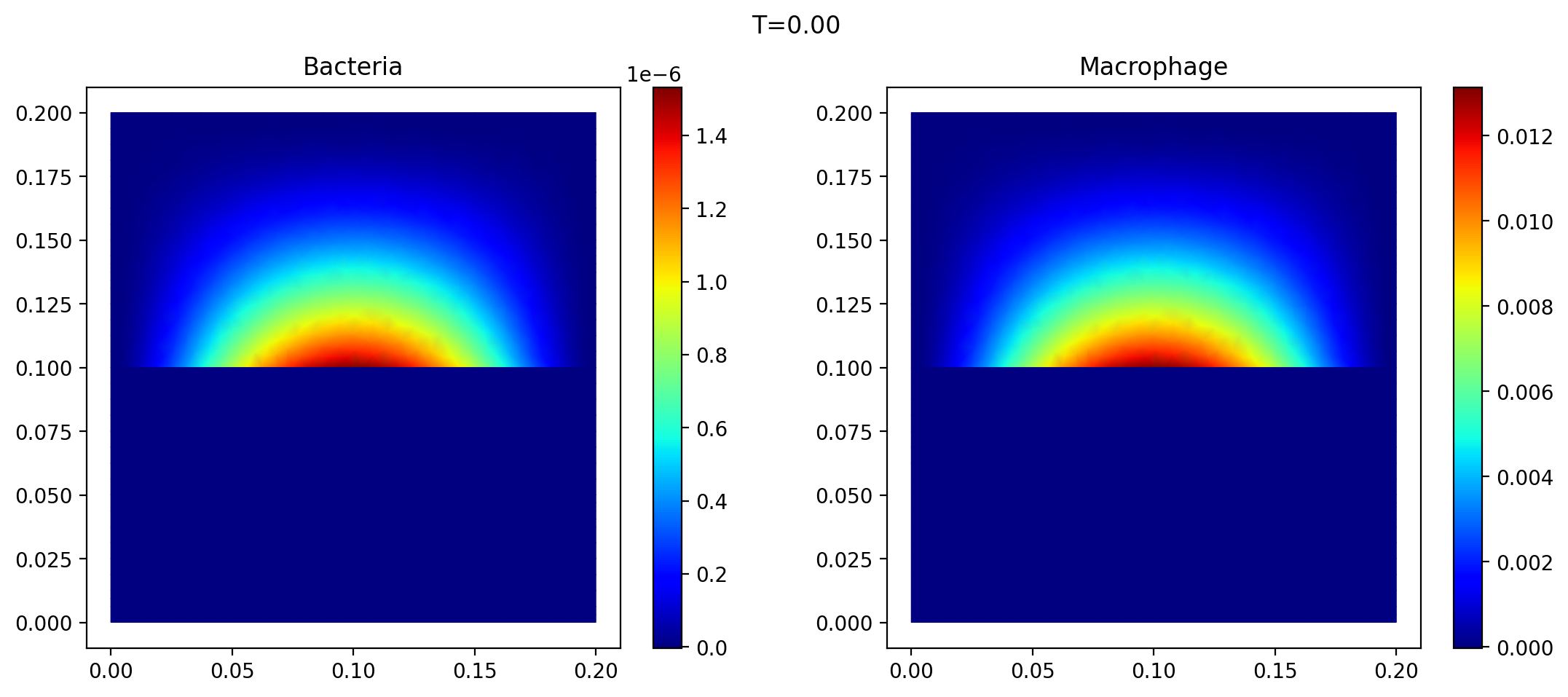}
        \caption{No velocity, $t=0$ day.}
        \label{fig:appendix-set188-no-velocity-t0}
    \end{subfigure}
    \hfill
    \begin{subfigure}{0.30\textwidth}
        \centering
        \includegraphics[width=\linewidth,height=0.76\textheight,keepaspectratio]{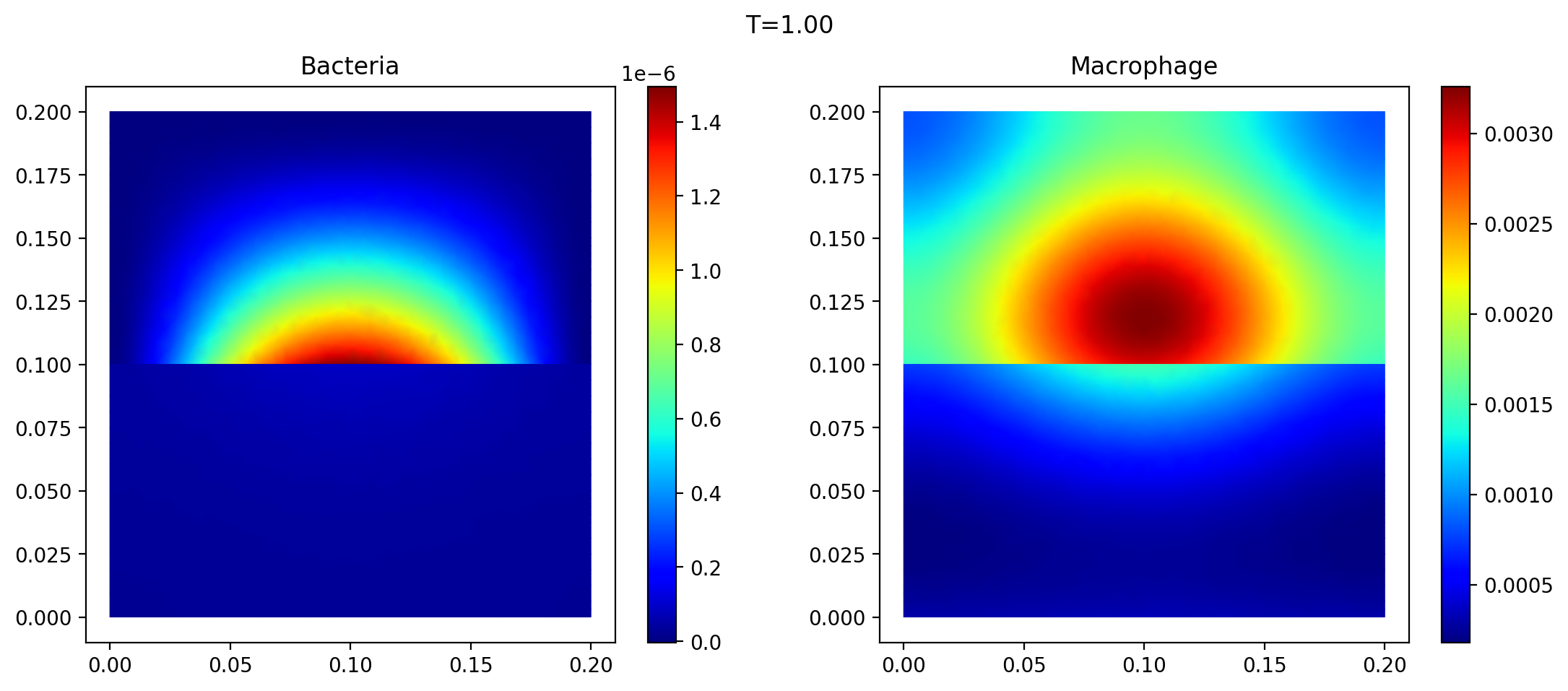}
        \caption{No velocity, $t=1$ day.}
        \label{fig:appendix-set188-no-velocity-t1}
    \end{subfigure}
    \hfill
    \begin{subfigure}{0.30\textwidth}
        \centering
        \includegraphics[width=\linewidth,height=0.76\textheight,keepaspectratio]{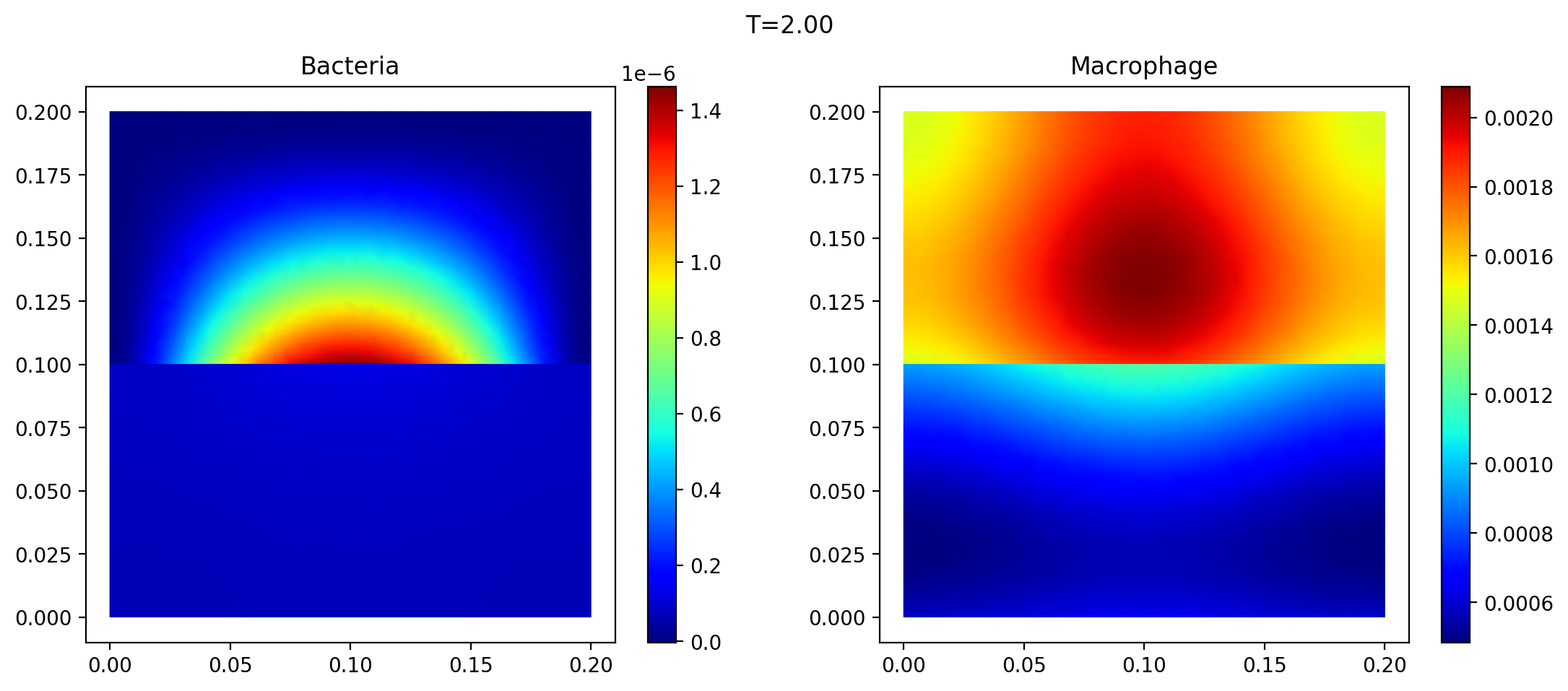}
        \caption{No velocity, $t=2$ day.}
        \label{fig:appendix-set188-no-velocity-t2}
    \end{subfigure}

    \vspace{0.5em}
    \begin{subfigure}{0.30\textwidth}
        \centering
        \includegraphics[width=\linewidth,height=0.76\textheight,keepaspectratio]{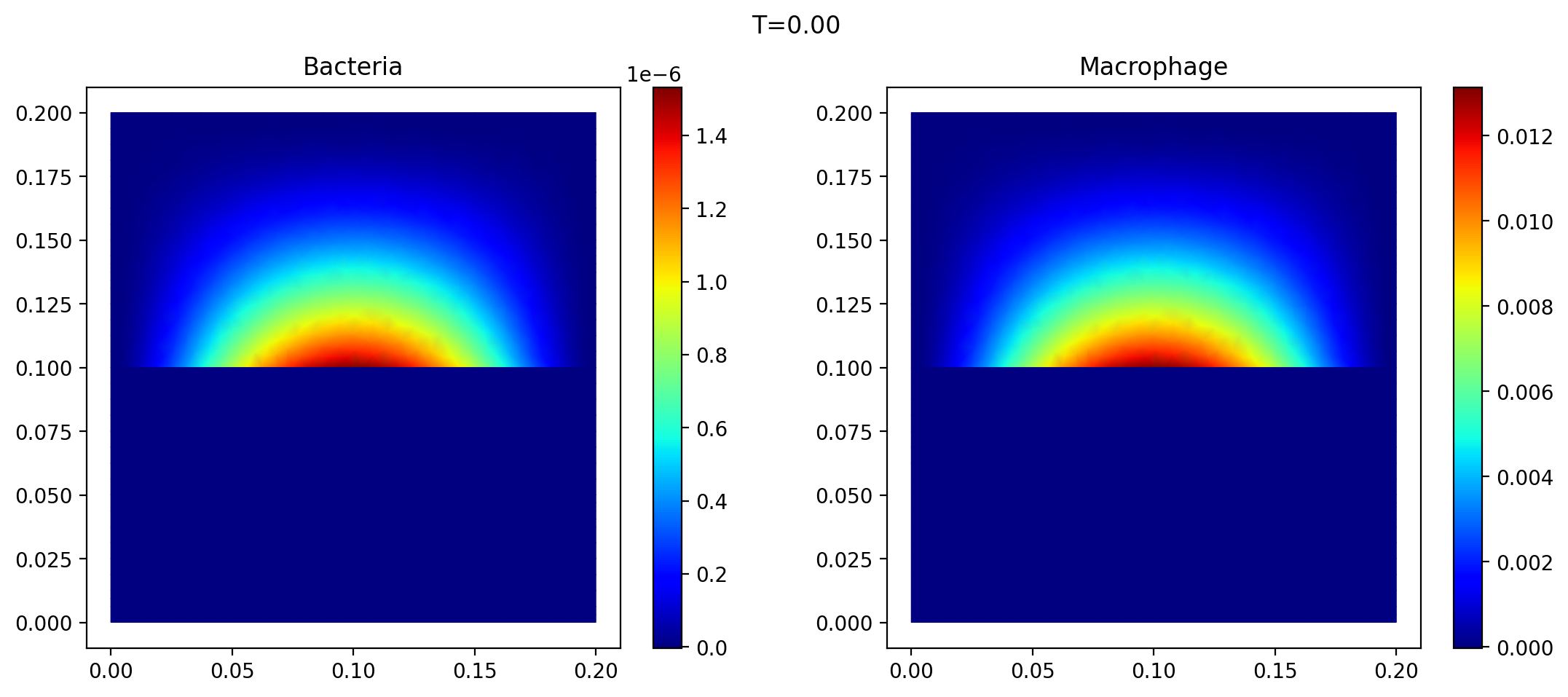}
        \caption{With velocity, $t=0$ day.}
        \label{fig:appendix-set188-velocity-t0}
    \end{subfigure}
    \hfill
    \begin{subfigure}{0.30\textwidth}
        \centering
        \includegraphics[width=\linewidth,height=0.76\textheight,keepaspectratio]{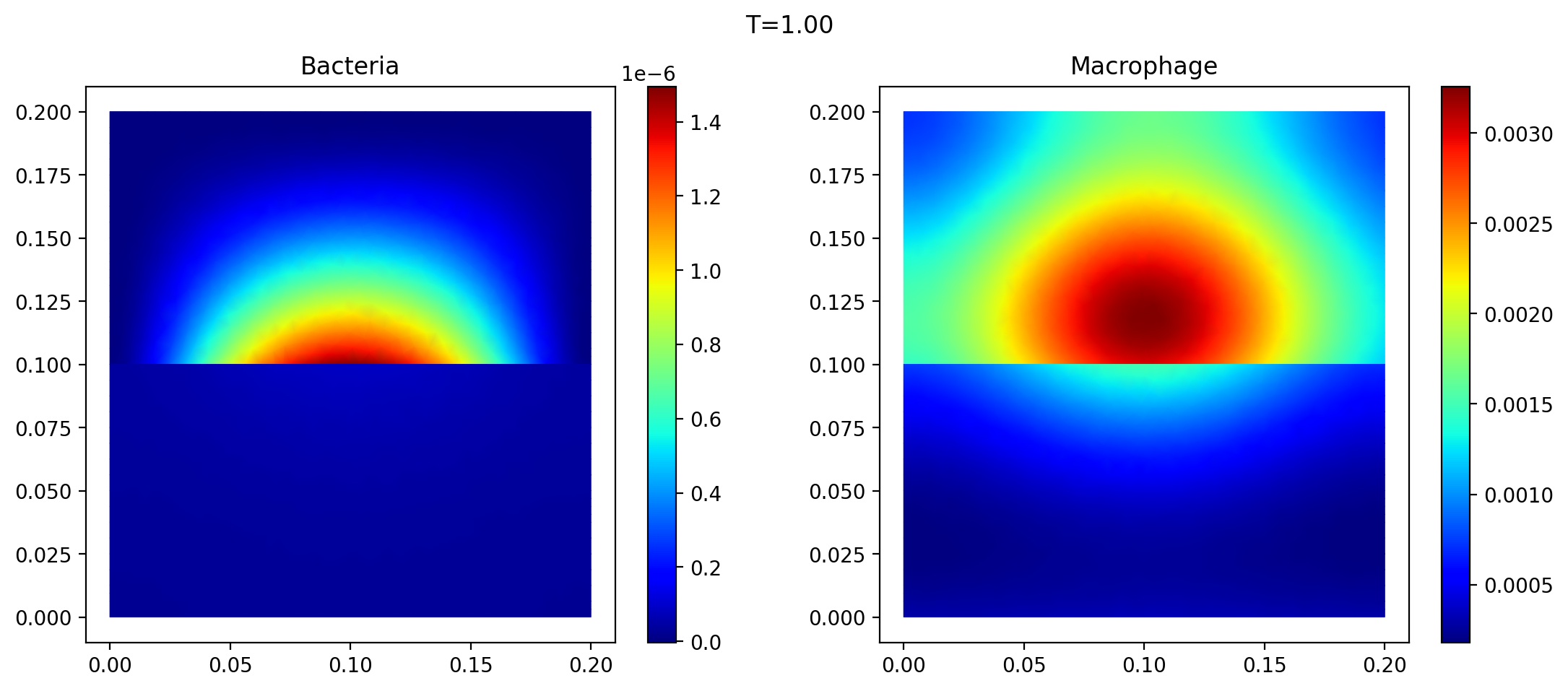}
        \caption{With velocity, $t=1$ day.}
        \label{fig:appendix-set188-velocity-t1}
    \end{subfigure}
    \hfill
    \begin{subfigure}{0.30\textwidth}
        \centering
        \includegraphics[width=\linewidth,height=0.76\textheight,keepaspectratio]{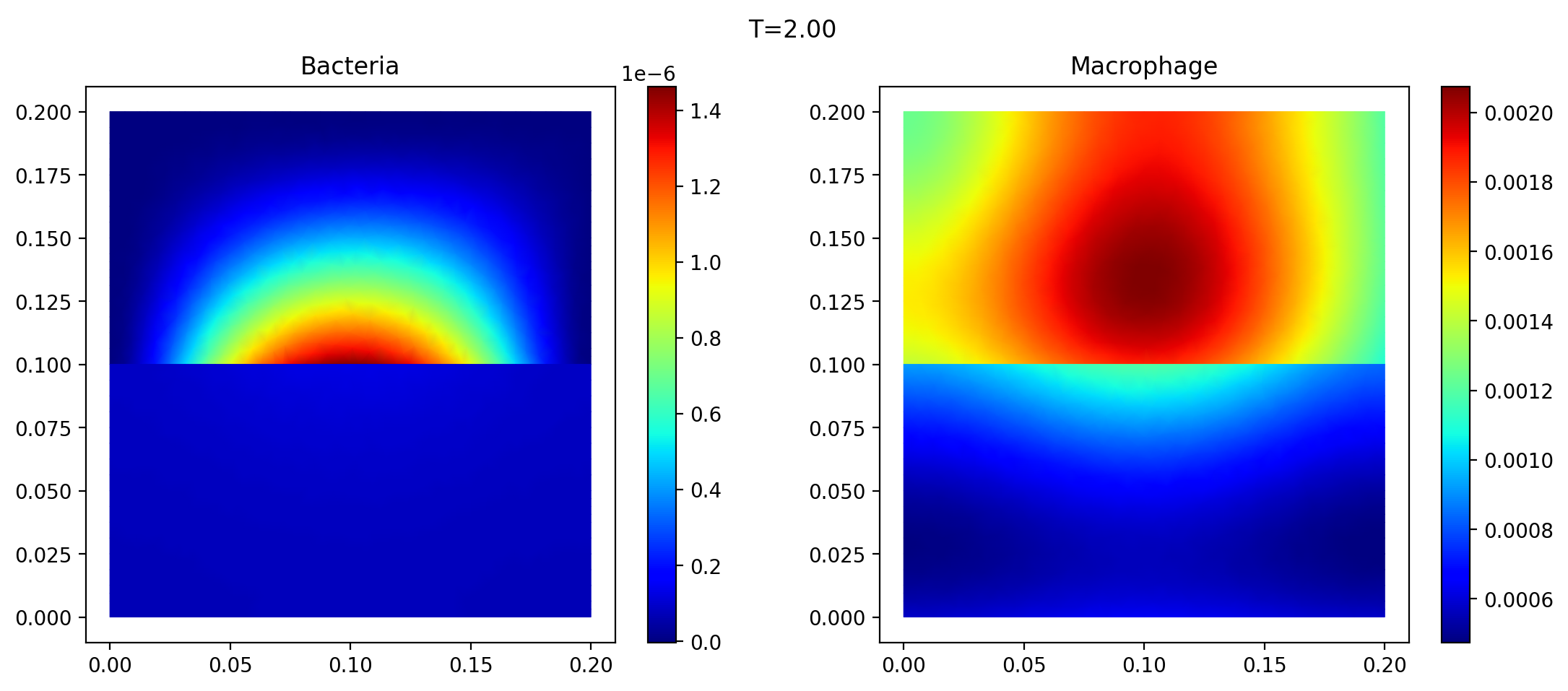}
        \caption{With velocity, $t= 2$ day.}
        \label{fig:appendix-set188-velocity-t2}
    \end{subfigure}

    \caption{
    Evolution of the system for parameter set No.~188, with and without the mucus velocity field.
    (a) Concentration fitting for bacteria and macrophages, with calibration reference data, mucus-region simulation, and tissue-region simulation shown in each plot.
    (b--d) Spatial distributions without mucus velocity field driven at $t=0$, $1$, and $2$ day.
    (e--g) Spatial distributions with the mucus velocity field driven at the same time points.
    Each spatial panel shows bacteria on the left and macrophages on the right.
    }
    \label{fig:appendix-set188-velocity-comparison}
\end{figure}
 \clearpage

\begin{figure}[H]
    \centering
    \begin{subfigure}{0.78\textwidth}
        \centering
        \includegraphics[width=\linewidth,height=0.76\textheight,keepaspectratio]{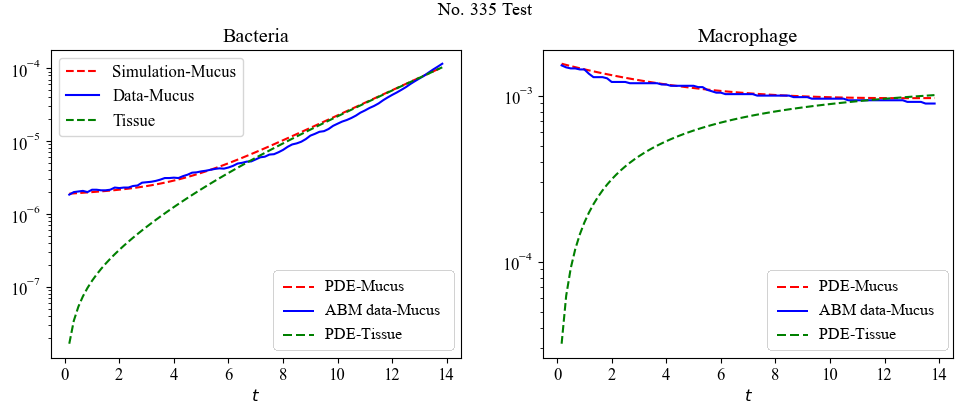}
        \caption{Concentration fitting.}
        \label{fig:appendix-set335-calibration}
    \end{subfigure}

    \vspace{0.5em}
    \begin{subfigure}{0.30\textwidth}
        \centering
        \includegraphics[width=\linewidth,height=0.76\textheight,keepaspectratio]{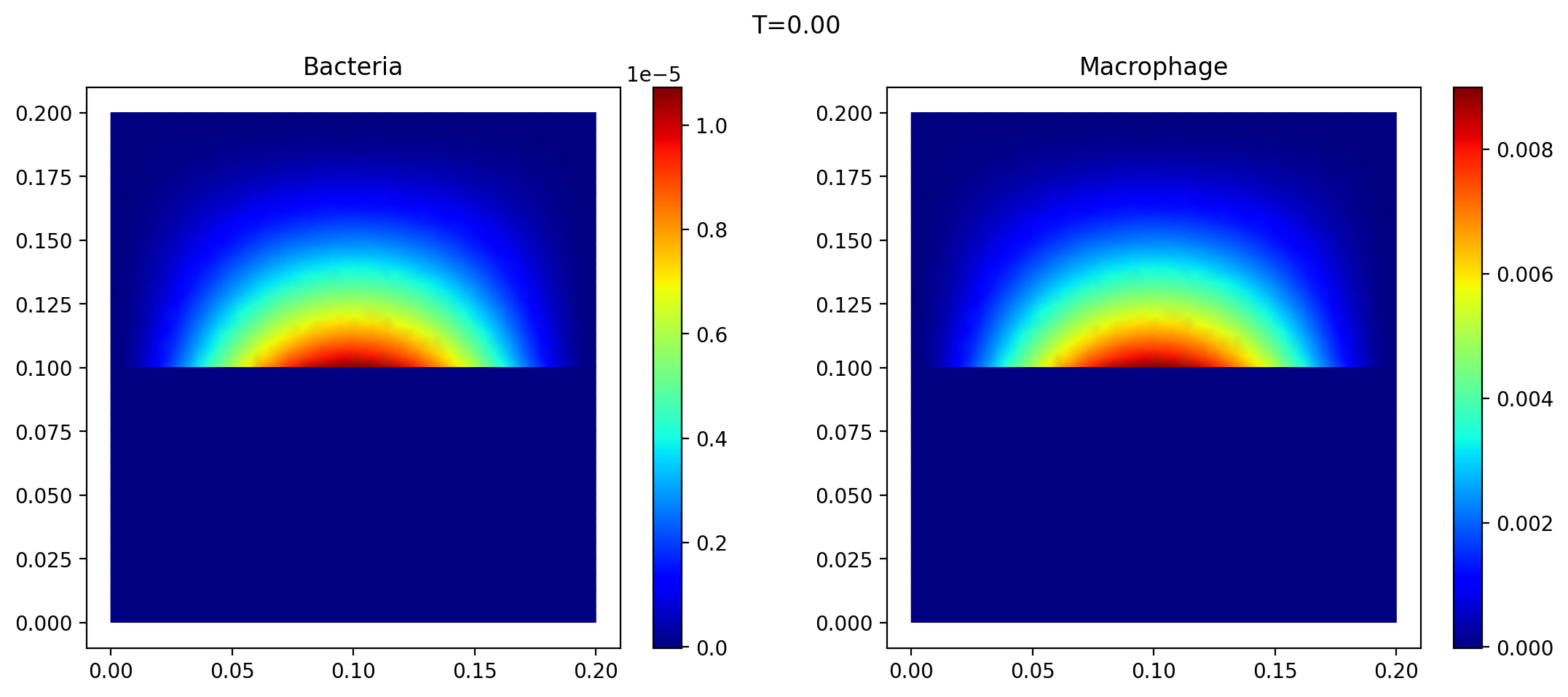}
        \caption{No velocity, $t=0$ day.}
        \label{fig:appendix-set335-no-velocity-t0}
    \end{subfigure}
    \hfill
    \begin{subfigure}{0.30\textwidth}
        \centering
        \includegraphics[width=\linewidth,height=0.76\textheight,keepaspectratio]{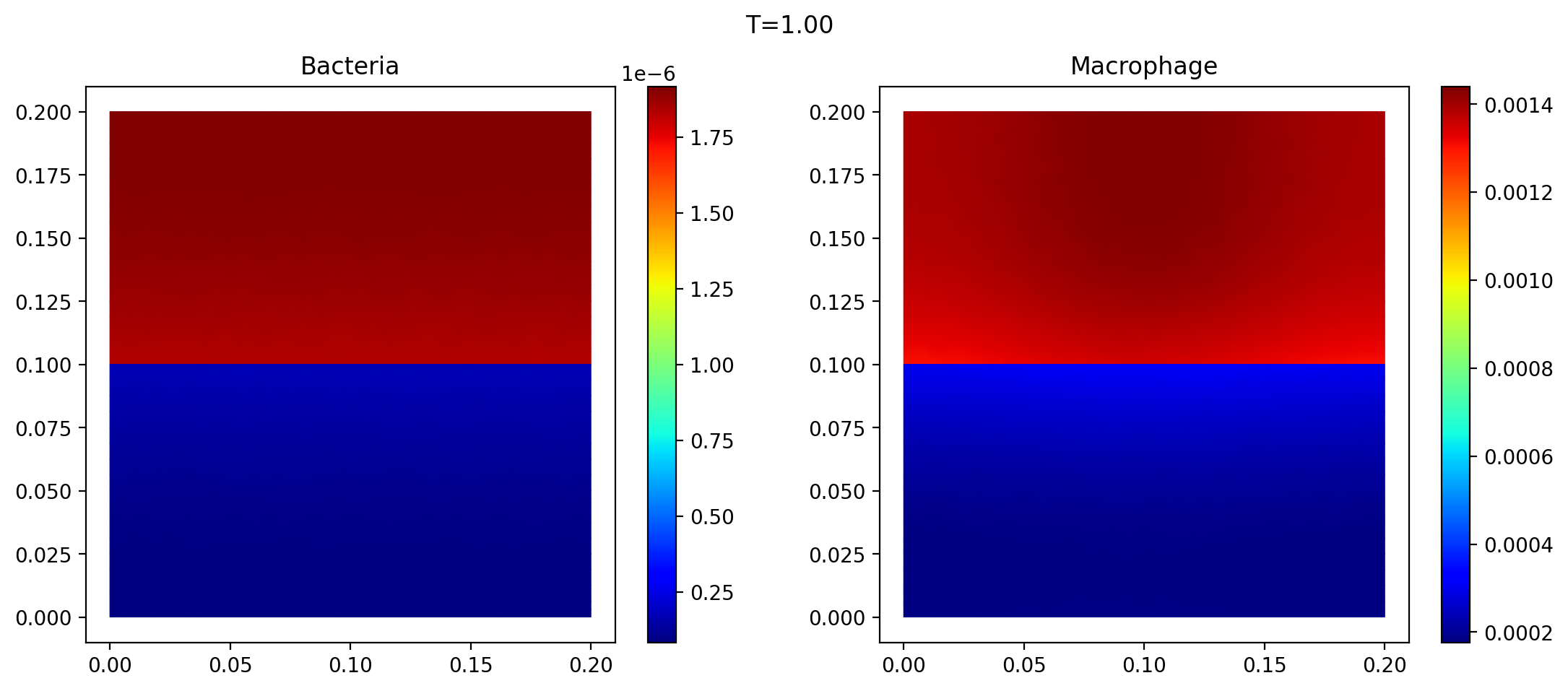}
        \caption{No velocity, $t=1$ day.}
        \label{fig:appendix-set335-no-velocity-t1}
    \end{subfigure}
    \hfill
    \begin{subfigure}{0.30\textwidth}
        \centering
        \includegraphics[width=\linewidth,height=0.76\textheight,keepaspectratio]{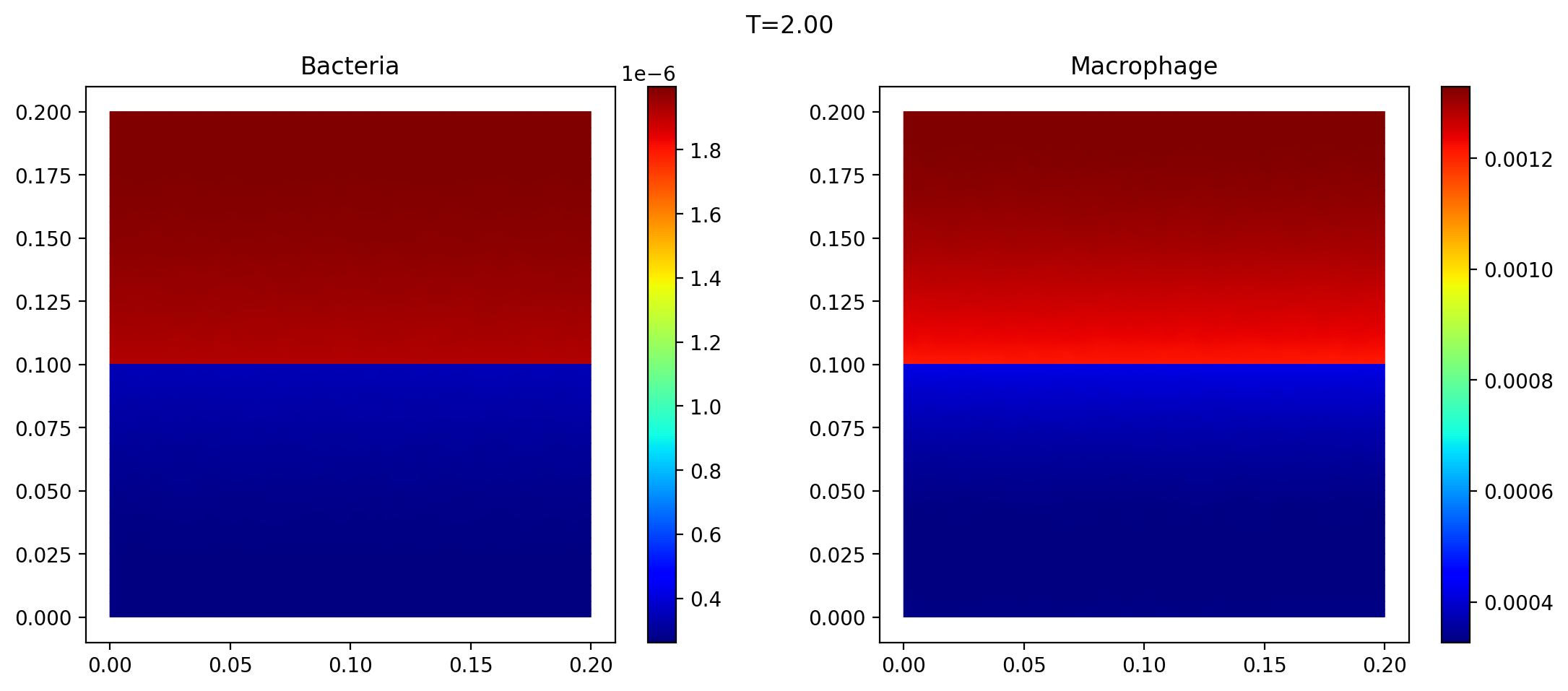}
        \caption{No velocity, $t=2$ day.}
        \label{fig:appendix-set335-no-velocity-t2}
    \end{subfigure}

    \vspace{0.5em}
    \begin{subfigure}{0.30\textwidth}
        \centering
        \includegraphics[width=\linewidth,height=0.76\textheight,keepaspectratio]{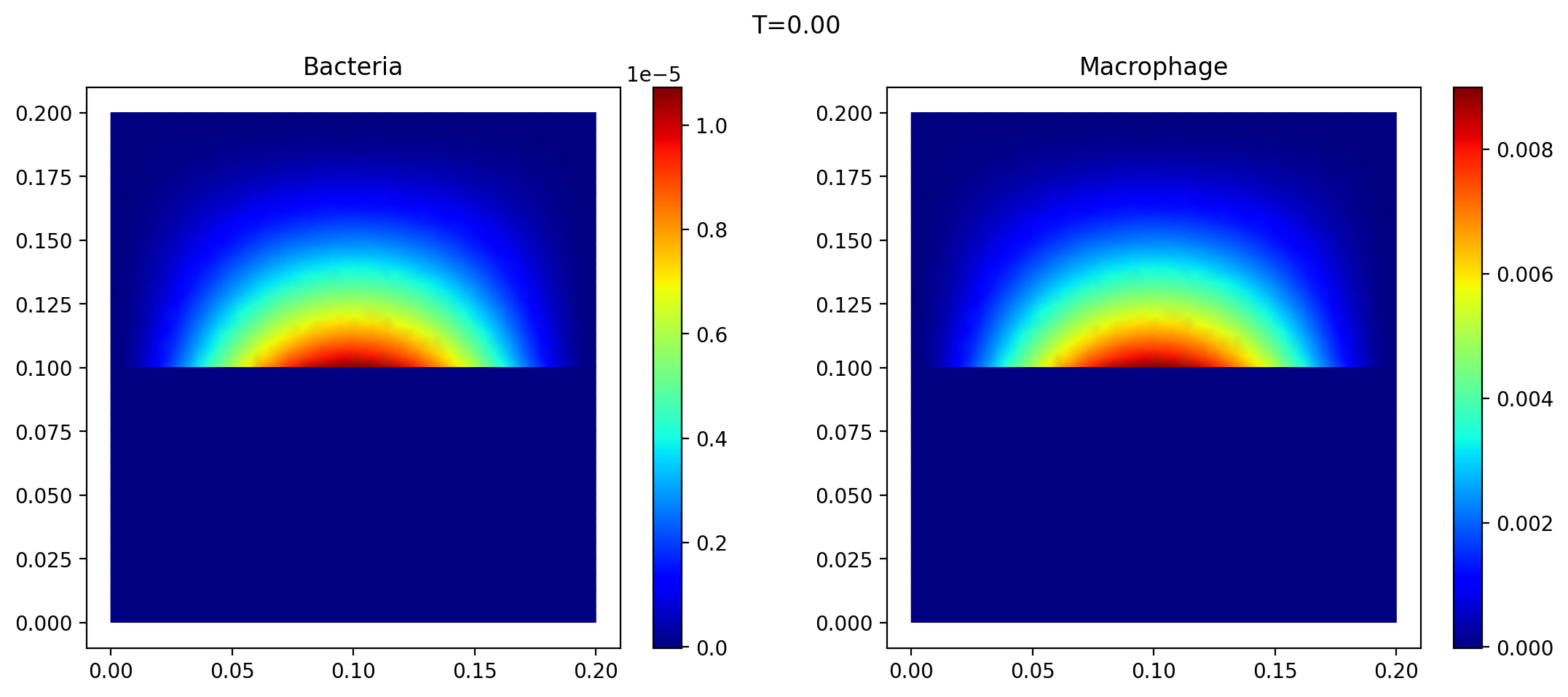}
        \caption{With velocity, $t=0$ day.}
        \label{fig:appendix-set335-velocity-t0}
    \end{subfigure}
    \hfill
    \begin{subfigure}{0.30\textwidth}
        \centering
        \includegraphics[width=\linewidth,height=0.76\textheight,keepaspectratio]{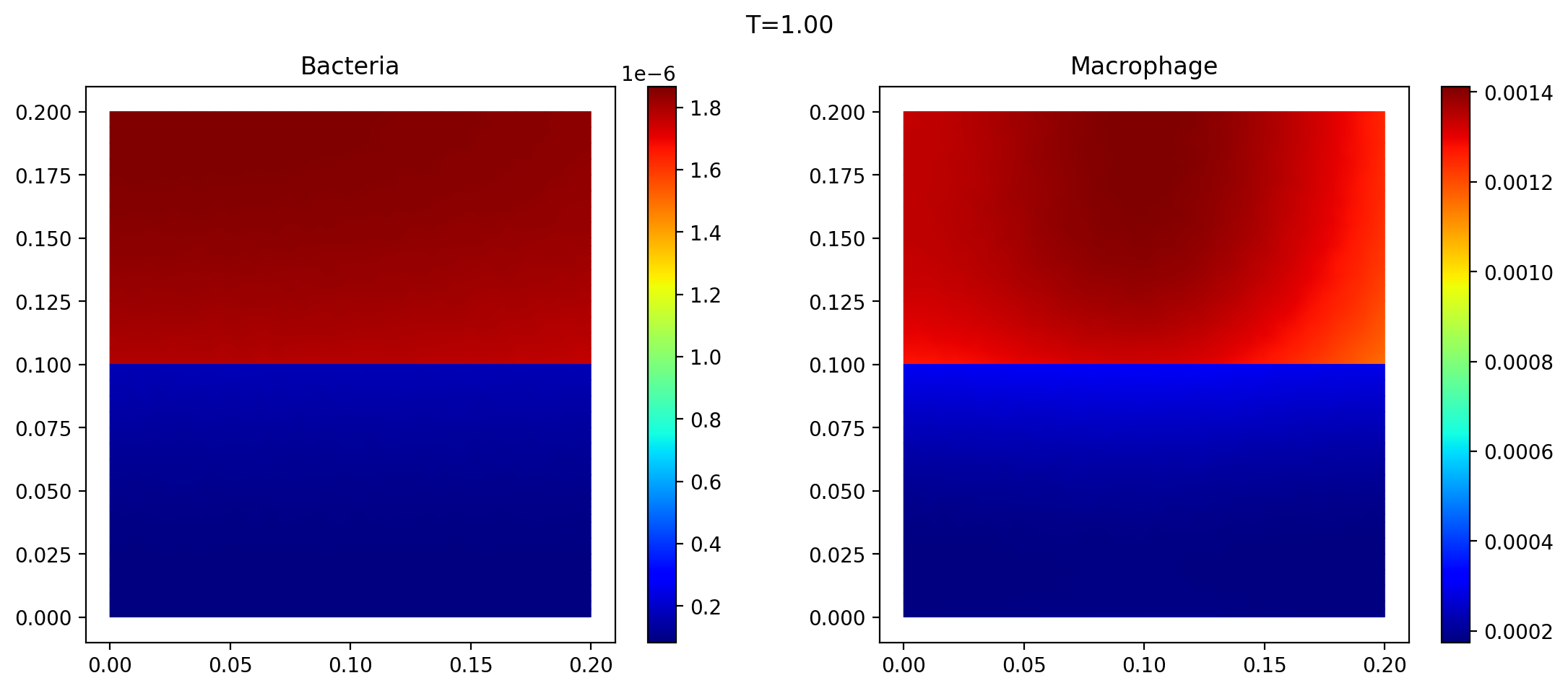}
        \caption{With velocity, $t=1$ day.}
        \label{fig:appendix-set335-velocity-t1}
    \end{subfigure}
    \hfill
    \begin{subfigure}{0.30\textwidth}
        \centering
        \includegraphics[width=\linewidth,height=0.76\textheight,keepaspectratio]{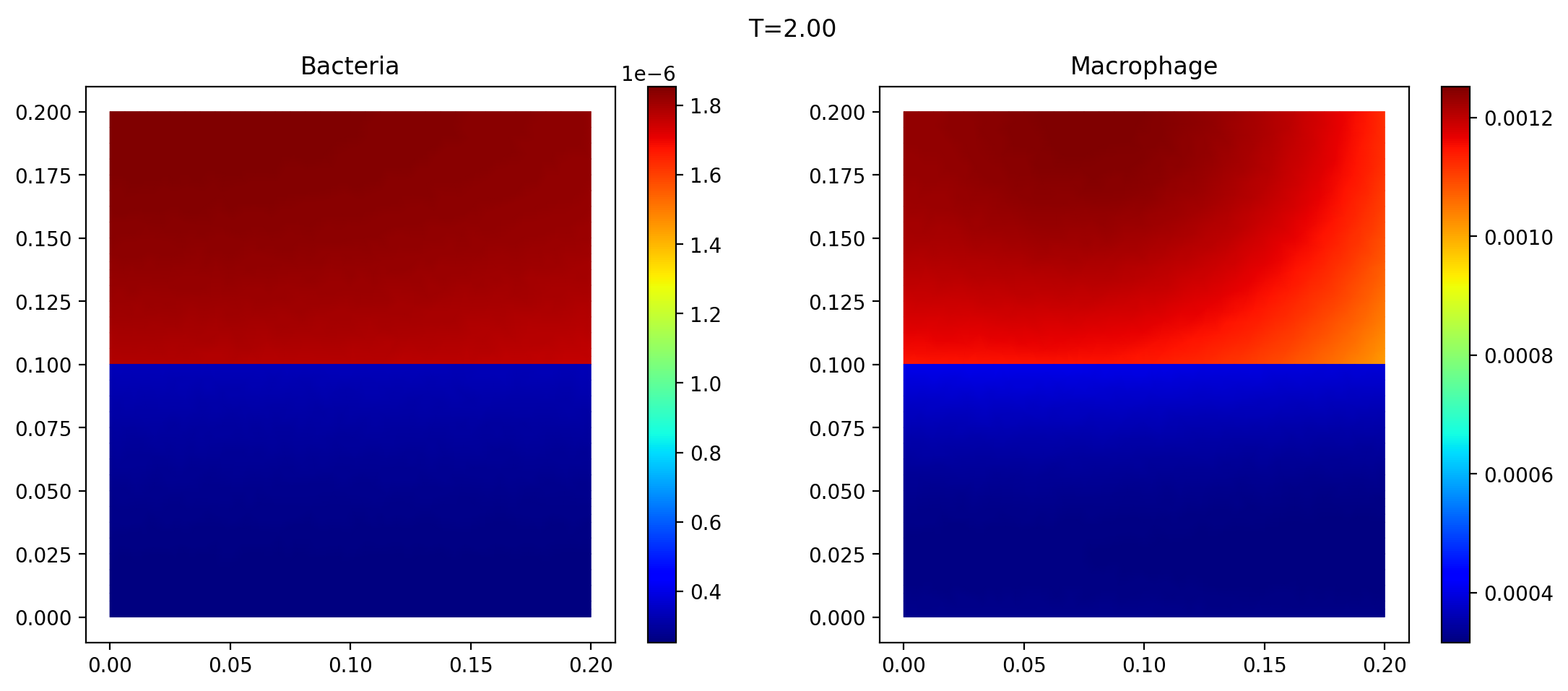}
        \caption{With velocity, $t= 2$ day.}
        \label{fig:appendix-set335-velocity-t2}
    \end{subfigure}

    \caption{
    Evolution of the system for parameter set No.~335, with and without the mucus velocity field.
    (a) Concentration fitting for bacteria and macrophages, with calibration reference data, mucus-region simulation, and tissue-region simulation shown in each plot.
    (b--d) Spatial distributions without mucus velocity field driven at $t=0$, $1$, and $2$ day.
    (e--g) Spatial distributions with the mucus velocity field driven at the same time points.
    Each spatial panel shows bacteria on the left and macrophages on the right.
    }
    \label{fig:appendix-set335-velocity-comparison}
\end{figure}
 \clearpage

\begin{figure}[H]
    \centering
    \begin{subfigure}{0.78\textwidth}
        \centering
        \includegraphics[width=\linewidth,height=0.76\textheight,keepaspectratio]{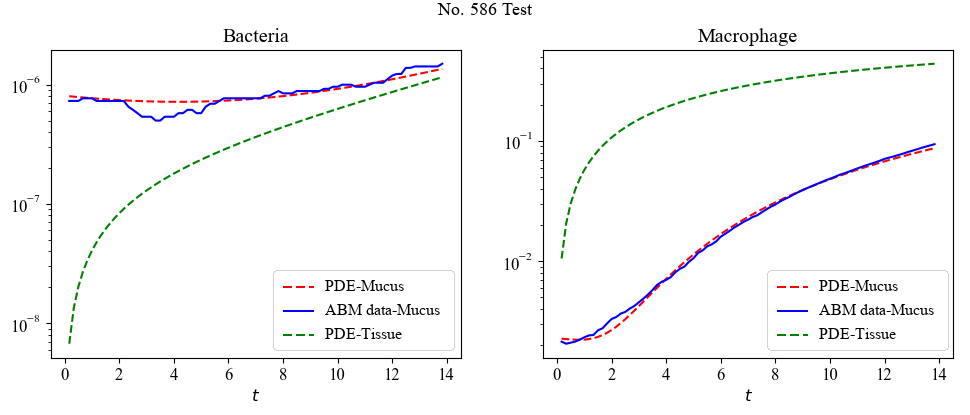}
        \caption{Concentration fitting.}
        \label{fig:appendix-set586-calibration}
    \end{subfigure}

    \vspace{0.5em}
    \begin{subfigure}{0.30\textwidth}
        \centering
        \includegraphics[width=\linewidth,height=0.76\textheight,keepaspectratio]{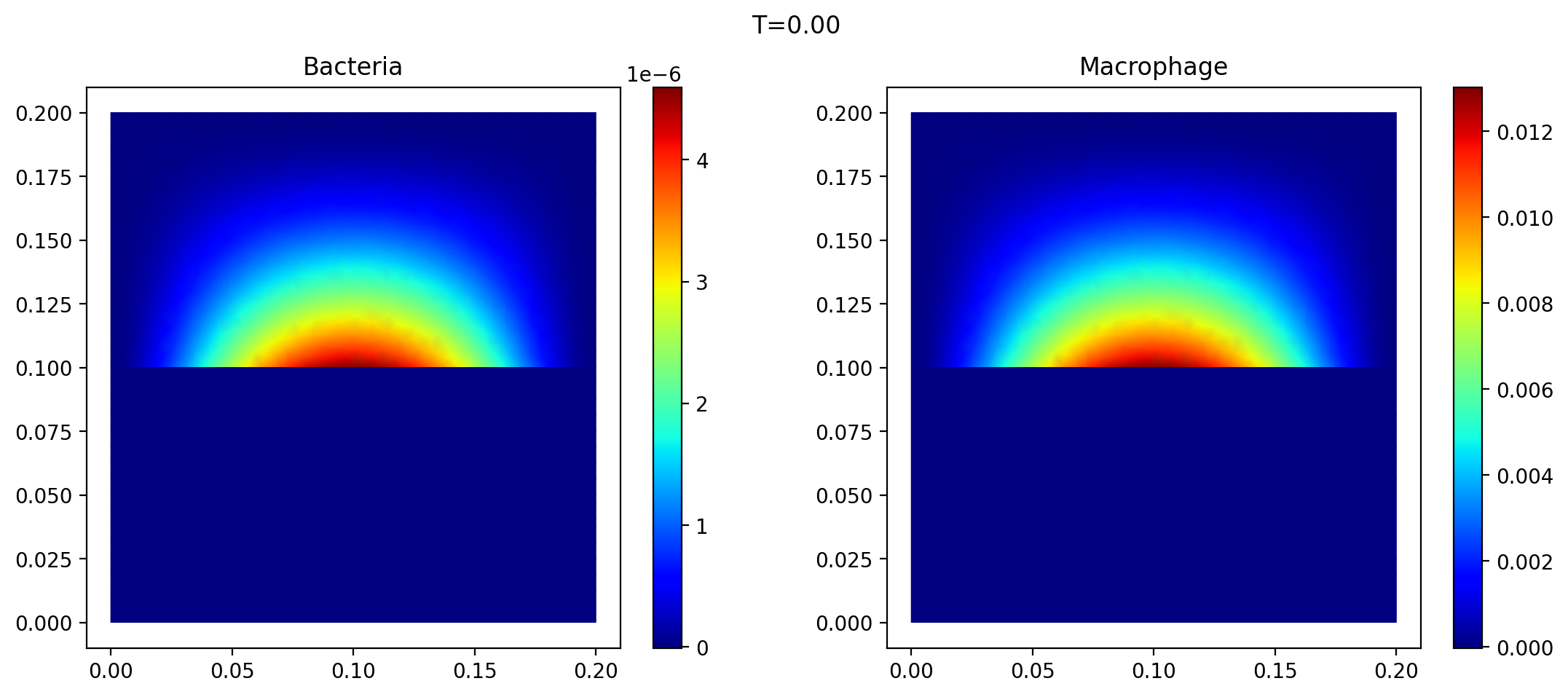}
        \caption{No velocity, $t=0$ day.}
        \label{fig:appendix-set586-no-velocity-t0}
    \end{subfigure}
    \hfill
    \begin{subfigure}{0.30\textwidth}
        \centering
        \includegraphics[width=\linewidth,height=0.76\textheight,keepaspectratio]{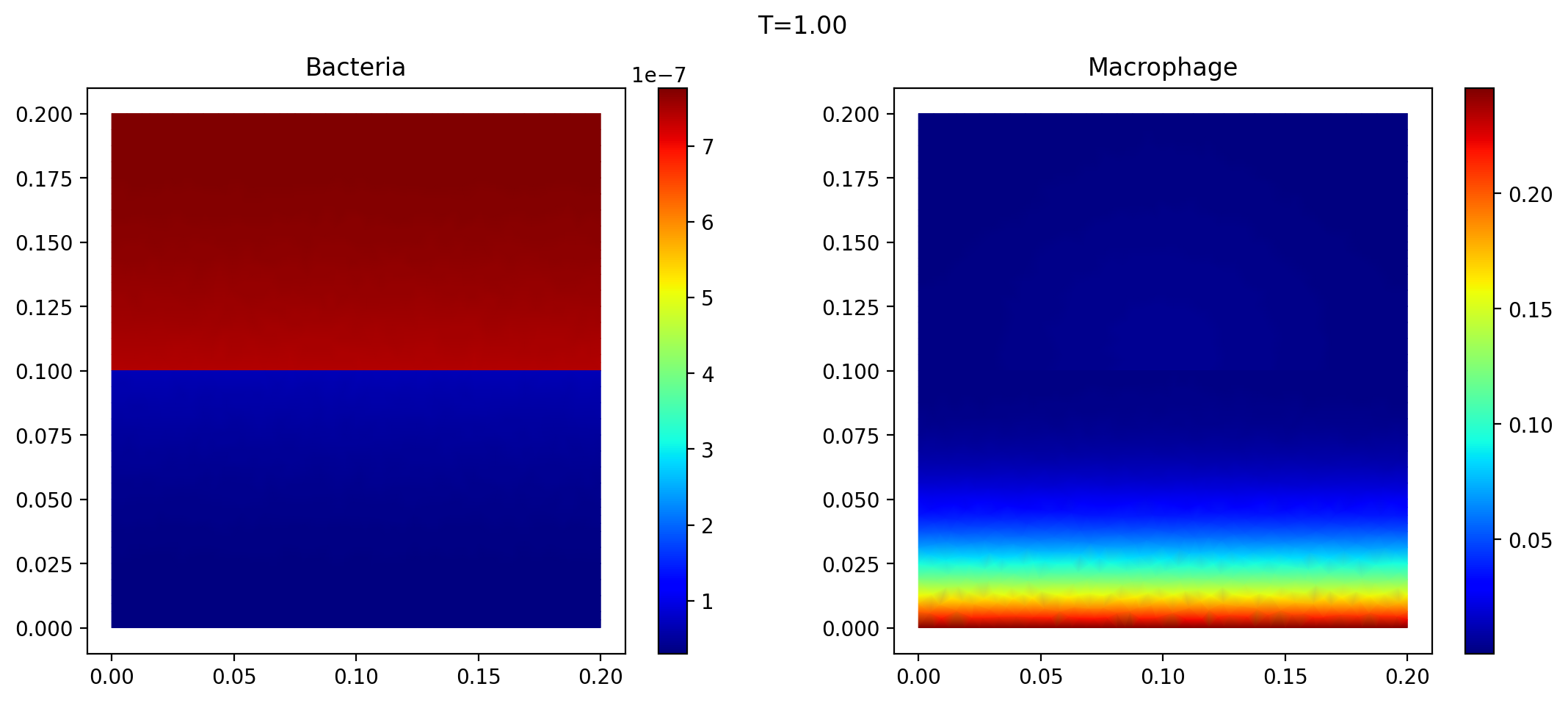}
        \caption{No velocity, $t=1$ day.}
        \label{fig:appendix-set586-no-velocity-t1}
    \end{subfigure}
    \hfill
    \begin{subfigure}{0.30\textwidth}
        \centering
        \includegraphics[width=\linewidth,height=0.76\textheight,keepaspectratio]{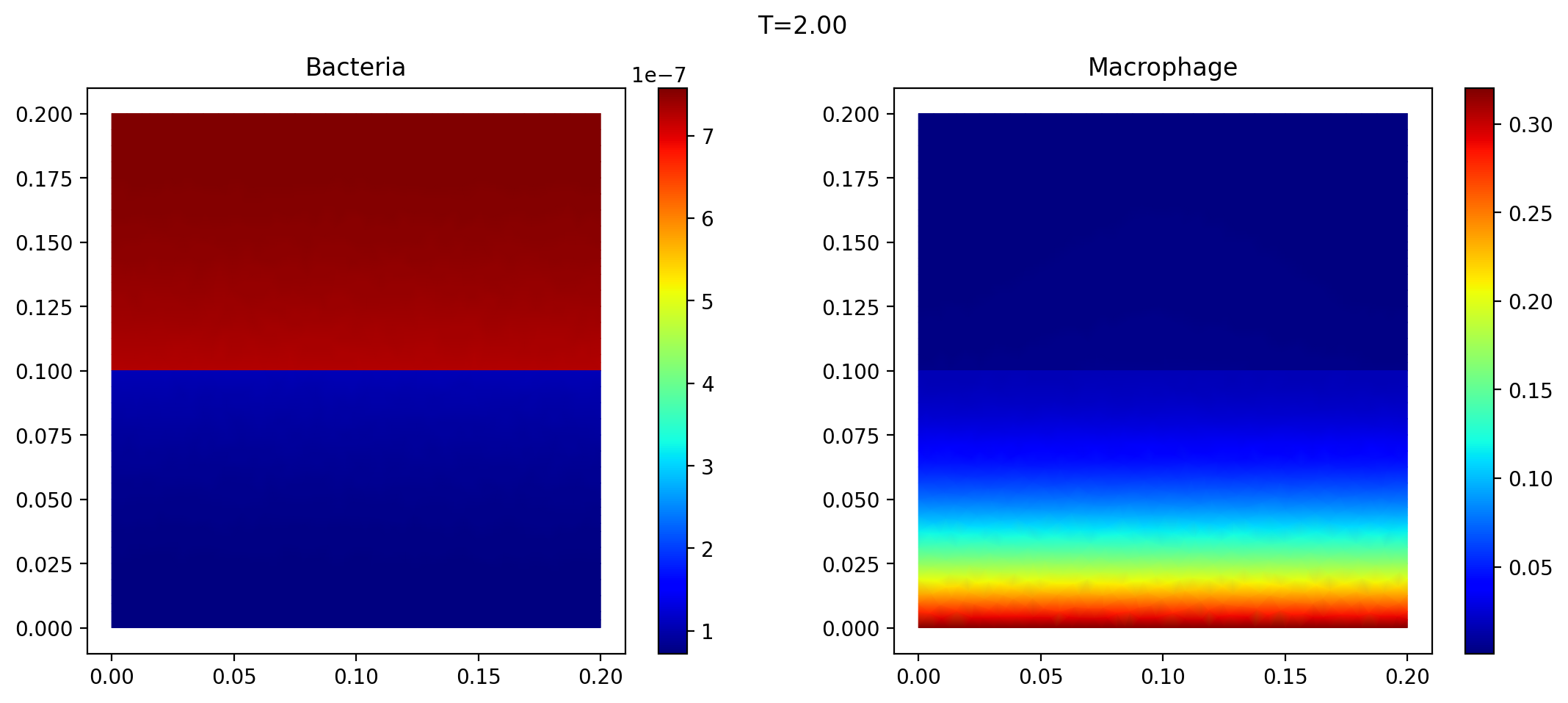}
        \caption{No velocity, $t=2$ day.}
        \label{fig:appendix-set586-no-velocity-t2}
    \end{subfigure}

    \vspace{0.5em}
    \begin{subfigure}{0.30\textwidth}
        \centering
        \includegraphics[width=\linewidth,height=0.76\textheight,keepaspectratio]{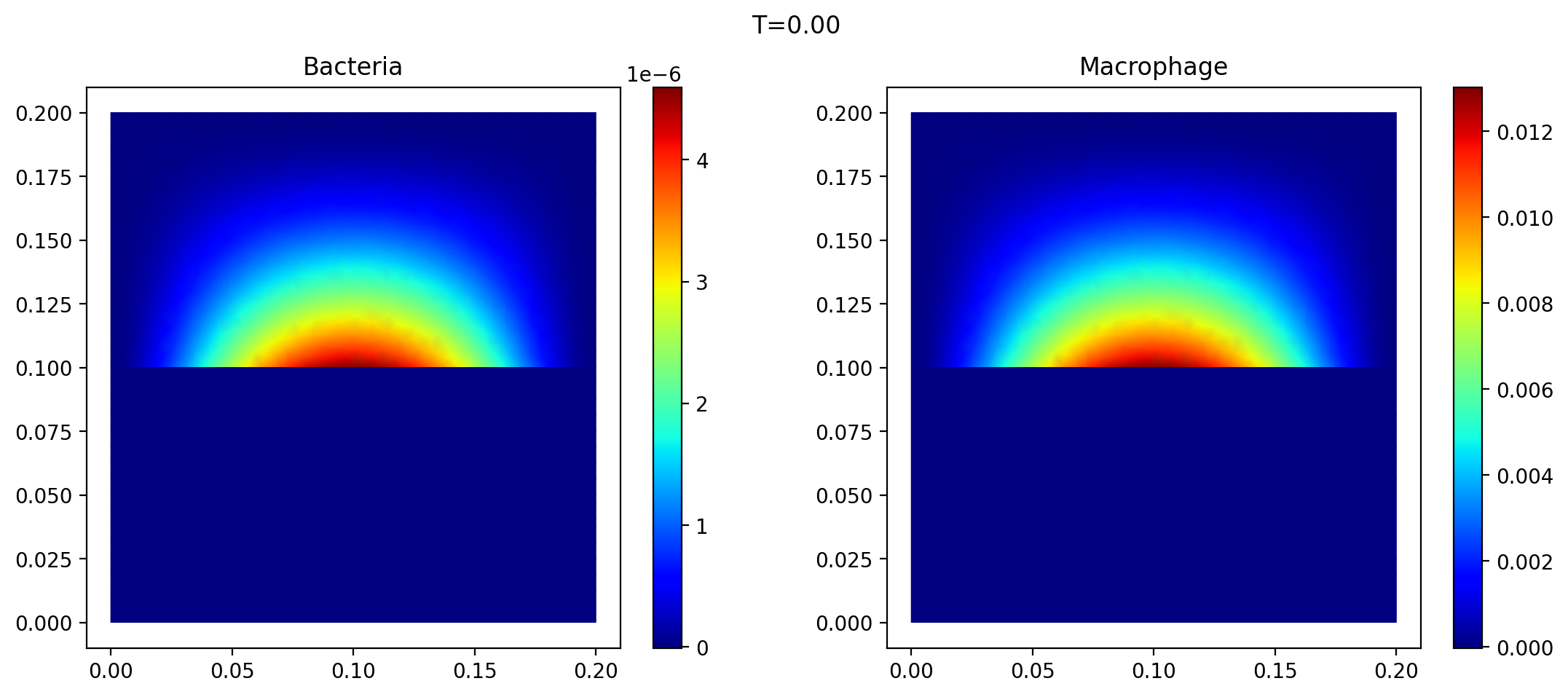}
        \caption{With velocity, $t=0$ day.}
        \label{fig:appendix-set586-velocity-t0}
    \end{subfigure}
    \hfill
    \begin{subfigure}{0.30\textwidth}
        \centering
        \includegraphics[width=\linewidth,height=0.76\textheight,keepaspectratio]{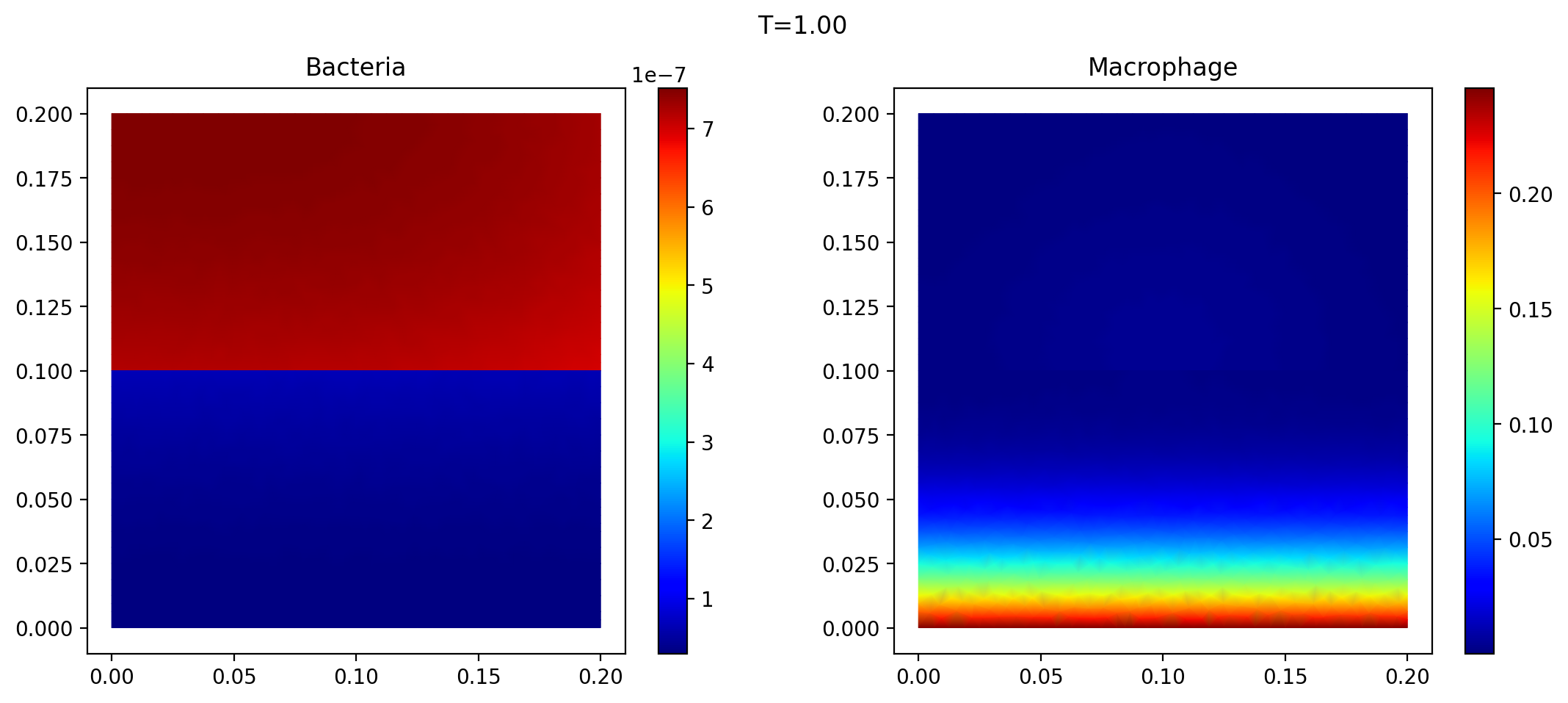}
        \caption{With velocity, $t=1$ day.}
        \label{fig:appendix-set586-velocity-t1}
    \end{subfigure}
    \hfill
    \begin{subfigure}{0.30\textwidth}
        \centering
        \includegraphics[width=\linewidth,height=0.76\textheight,keepaspectratio]{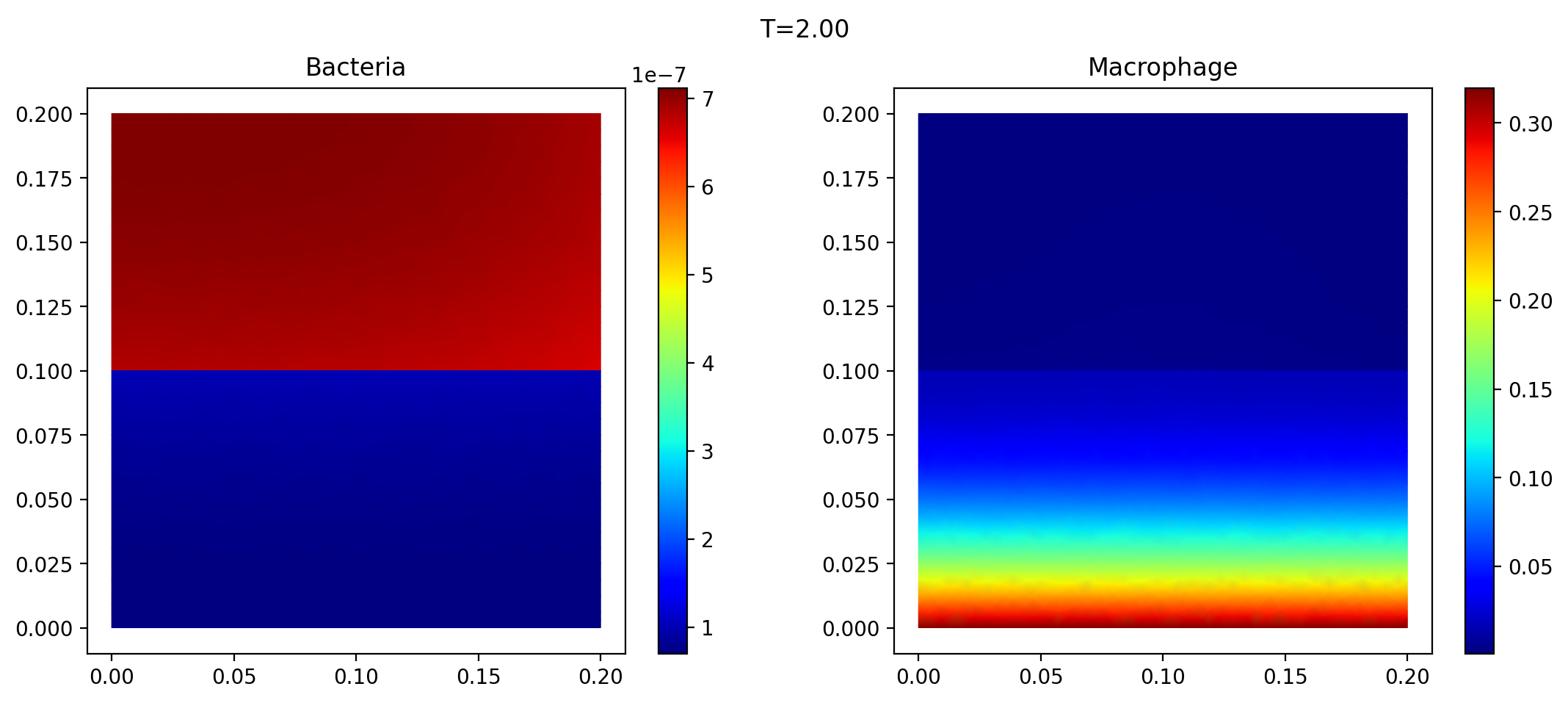}
        \caption{With velocity, $t= 2$ day.}
        \label{fig:appendix-set586-velocity-t2}
    \end{subfigure}

    \caption{
    Evolution of the system for parameter set No.~586, with and without the mucus velocity field.
    (a) Concentration fitting for bacteria and macrophages, with calibration reference data, mucus-region simulation, and tissue-region simulation shown in each plot.
    (b--d) Spatial distributions without mucus velocity field driven at $t=0$, $1$, and $2$ day.
    (e--g) Spatial distributions with the mucus velocity field driven at the same time points.
    Each spatial panel shows bacteria on the left and macrophages on the right.
    }
    \label{fig:appendix-set586-velocity-comparison}
\end{figure}
 \clearpage

\begin{figure}[H]
    \centering
    \begin{subfigure}{0.78\textwidth}
        \centering
        \includegraphics[width=\linewidth,height=0.76\textheight,keepaspectratio]{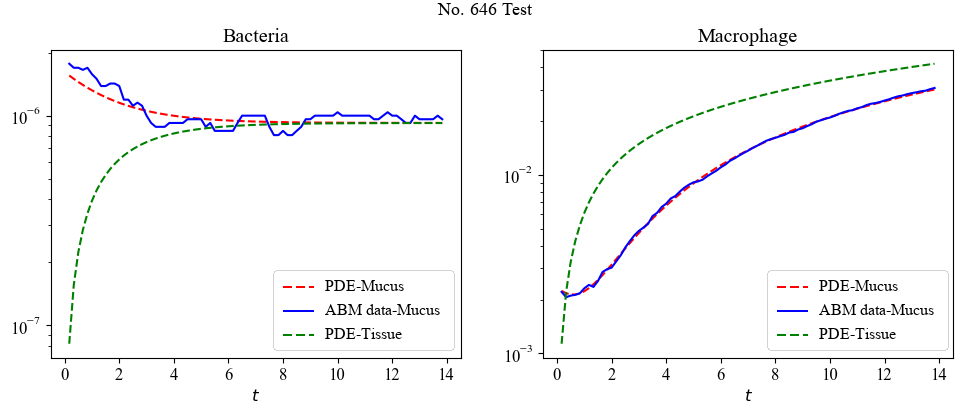}
        \caption{Concentration fitting.}
        \label{fig:appendix-set646-calibration}
    \end{subfigure}

    \vspace{0.5em}
    \begin{subfigure}{0.30\textwidth}
        \centering
        \includegraphics[width=\linewidth,height=0.76\textheight,keepaspectratio]{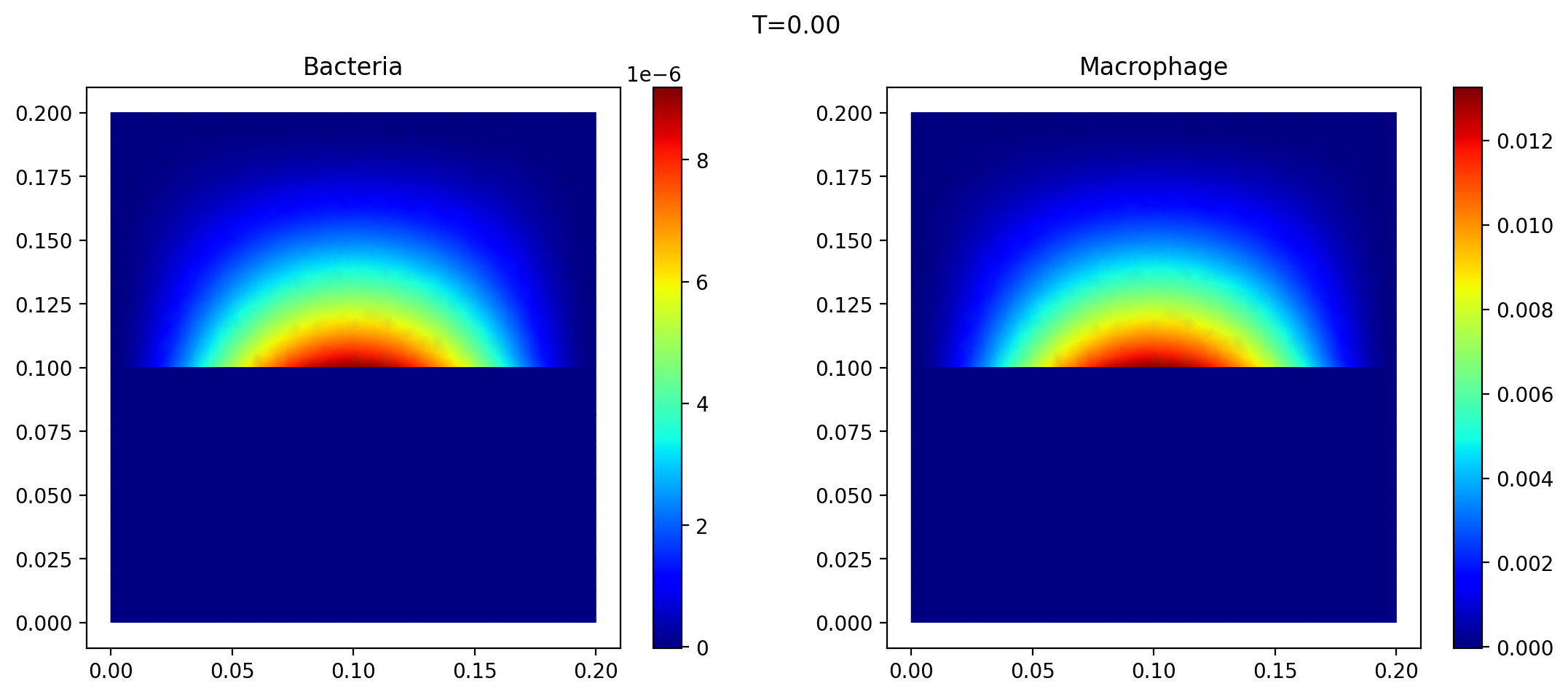}
        \caption{No velocity, $t=0$ day.}
        \label{fig:appendix-set646-no-velocity-t0}
    \end{subfigure}
    \hfill
    \begin{subfigure}{0.30\textwidth}
        \centering
        \includegraphics[width=\linewidth,height=0.76\textheight,keepaspectratio]{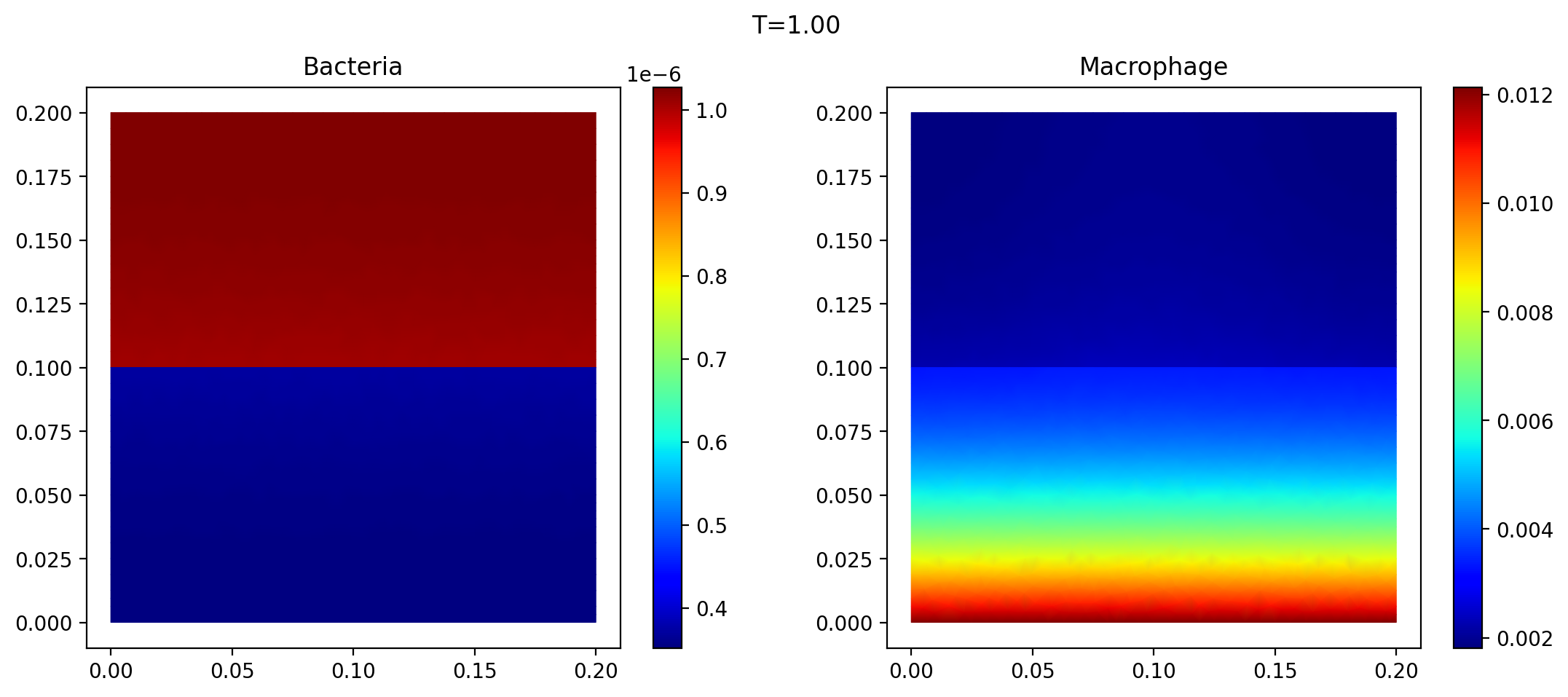}
        \caption{No velocity, $t=1$ day.}
        \label{fig:appendix-set646-no-velocity-t1}
    \end{subfigure}
    \hfill
    \begin{subfigure}{0.30\textwidth}
        \centering
        \includegraphics[width=\linewidth,height=0.76\textheight,keepaspectratio]{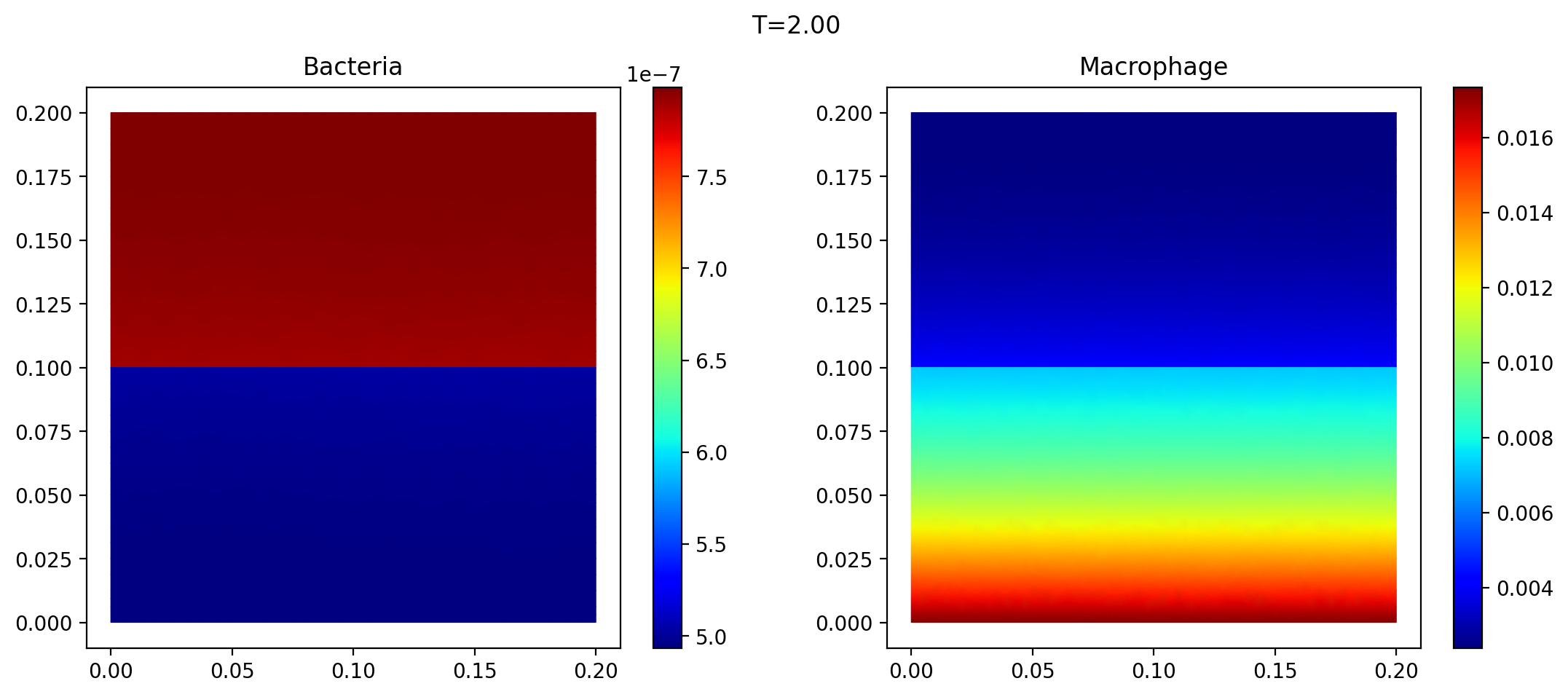}
        \caption{No velocity, $t=2$ day.}
        \label{fig:appendix-set646-no-velocity-t2}
    \end{subfigure}

    \vspace{0.5em}
    \begin{subfigure}{0.30\textwidth}
        \centering
        \includegraphics[width=\linewidth,height=0.76\textheight,keepaspectratio]{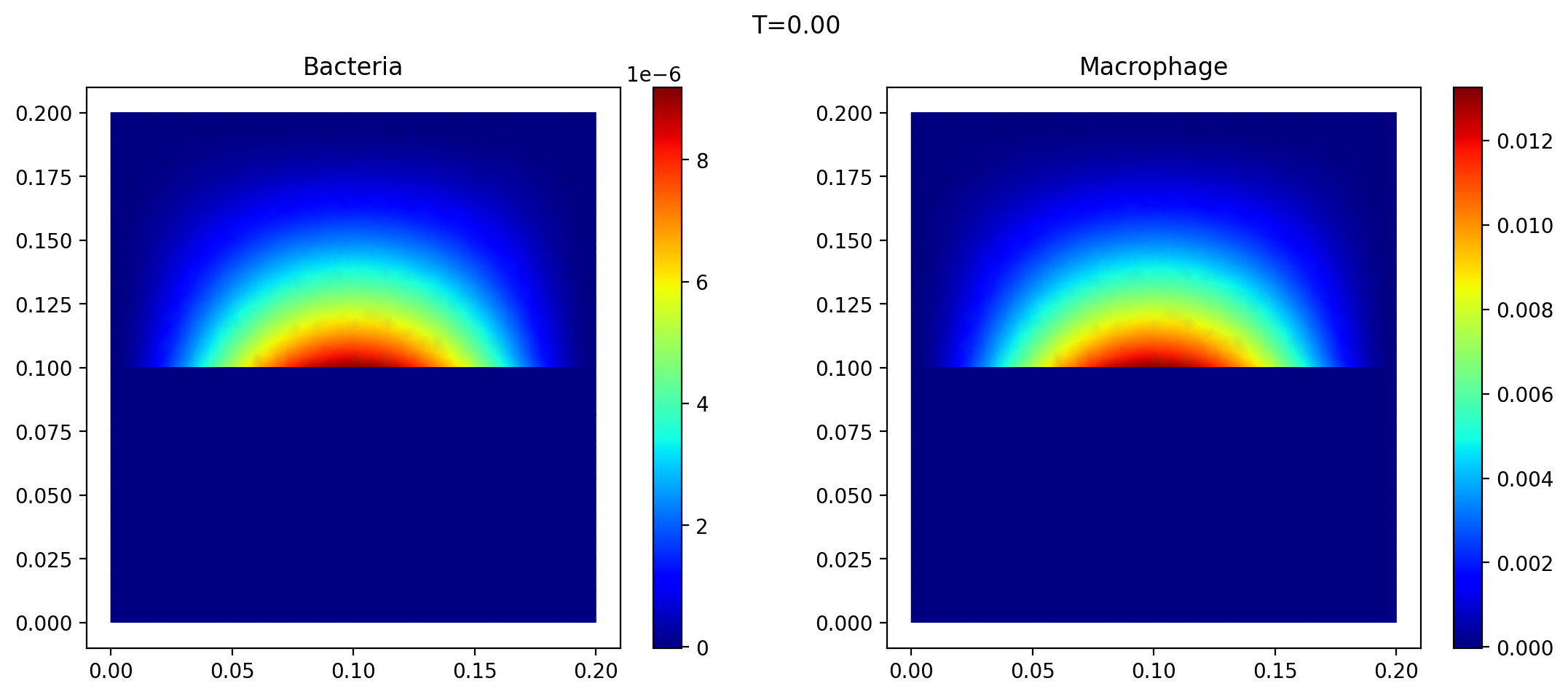}
        \caption{With velocity, $t=0$ day.}
        \label{fig:appendix-set646-velocity-t0}
    \end{subfigure}
    \hfill
    \begin{subfigure}{0.30\textwidth}
        \centering
        \includegraphics[width=\linewidth,height=0.76\textheight,keepaspectratio]{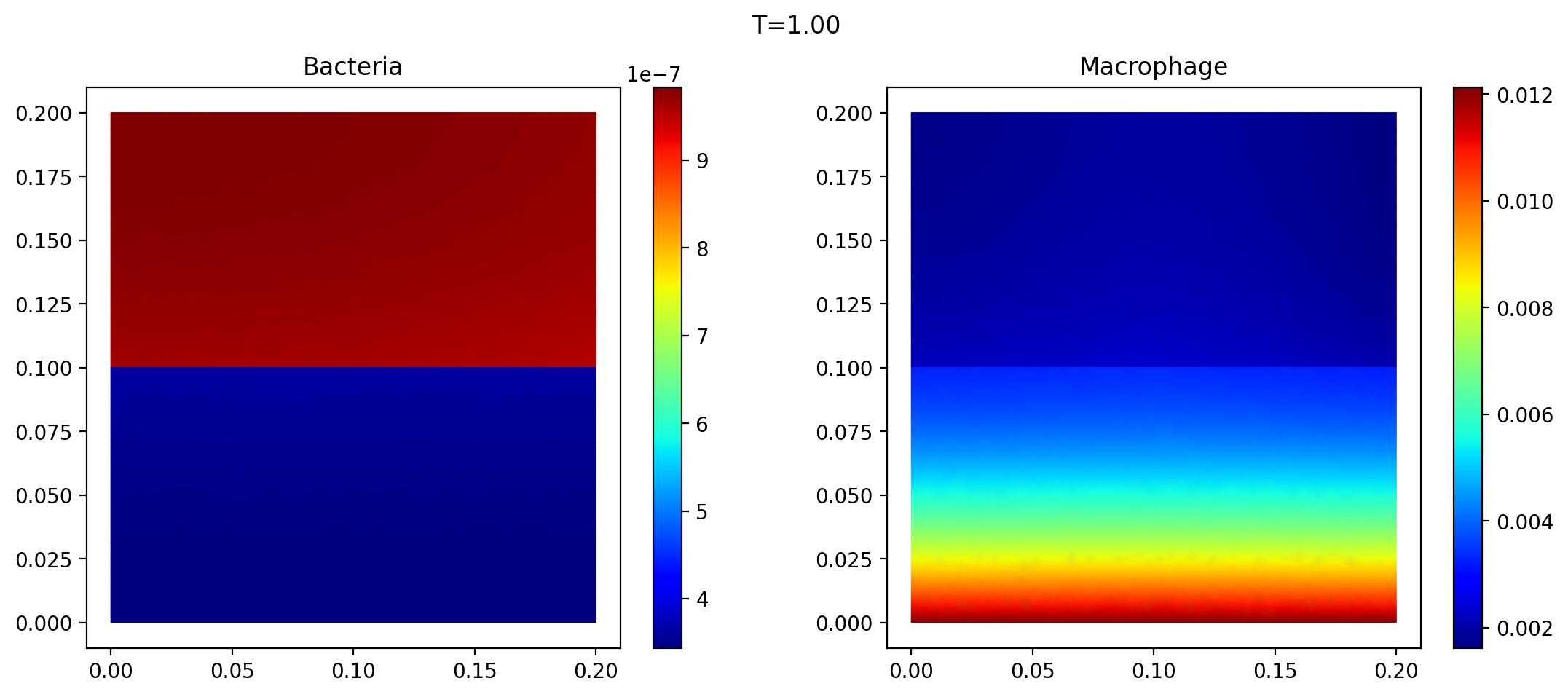}
        \caption{With velocity, $t=1$ day.}
        \label{fig:appendix-set646-velocity-t1}
    \end{subfigure}
    \hfill
    \begin{subfigure}{0.30\textwidth}
        \centering
        \includegraphics[width=\linewidth,height=0.76\textheight,keepaspectratio]{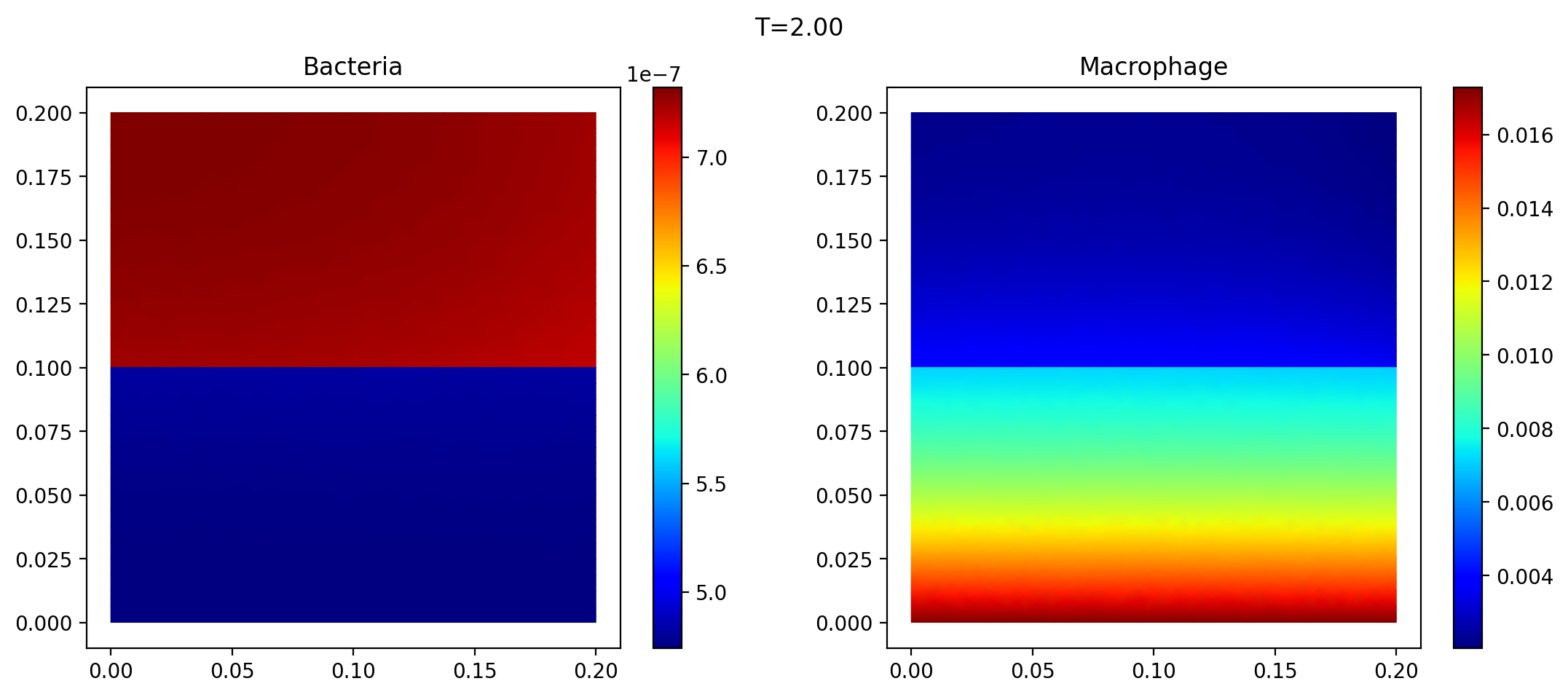}
        \caption{With velocity, $t= 2$ day.}
        \label{fig:appendix-set646-velocity-t2}
    \end{subfigure}

    \caption{
    Evolution of the system for parameter set No.~646, with and without the mucus velocity field.
    (a) Concentration fitting for bacteria and macrophages, with calibration reference data, mucus-region simulation, and tissue-region simulation shown in each plot.
    (b--d) Spatial distributions without mucus velocity field driven at $t=0$, $1$, and $2$ day.
    (e--g) Spatial distributions with the mucus velocity field driven at the same time points.
    Each spatial panel shows bacteria on the left and macrophages on the right.
    }
    \label{fig:appendix-set646-velocity-comparison}
\end{figure}
 \clearpage

\begin{figure}[H]
    \centering
    \begin{subfigure}{0.78\textwidth}
        \centering
        \includegraphics[width=\linewidth,height=0.76\textheight,keepaspectratio]{figs/spatial/set756_calibration_bm.png}
        \caption{Concentration fitting.}
        \label{fig:appendix-set756-calibration}
    \end{subfigure}

    \vspace{0.5em}
    \begin{subfigure}{0.30\textwidth}
        \centering
        \includegraphics[width=\linewidth,height=0.76\textheight,keepaspectratio]{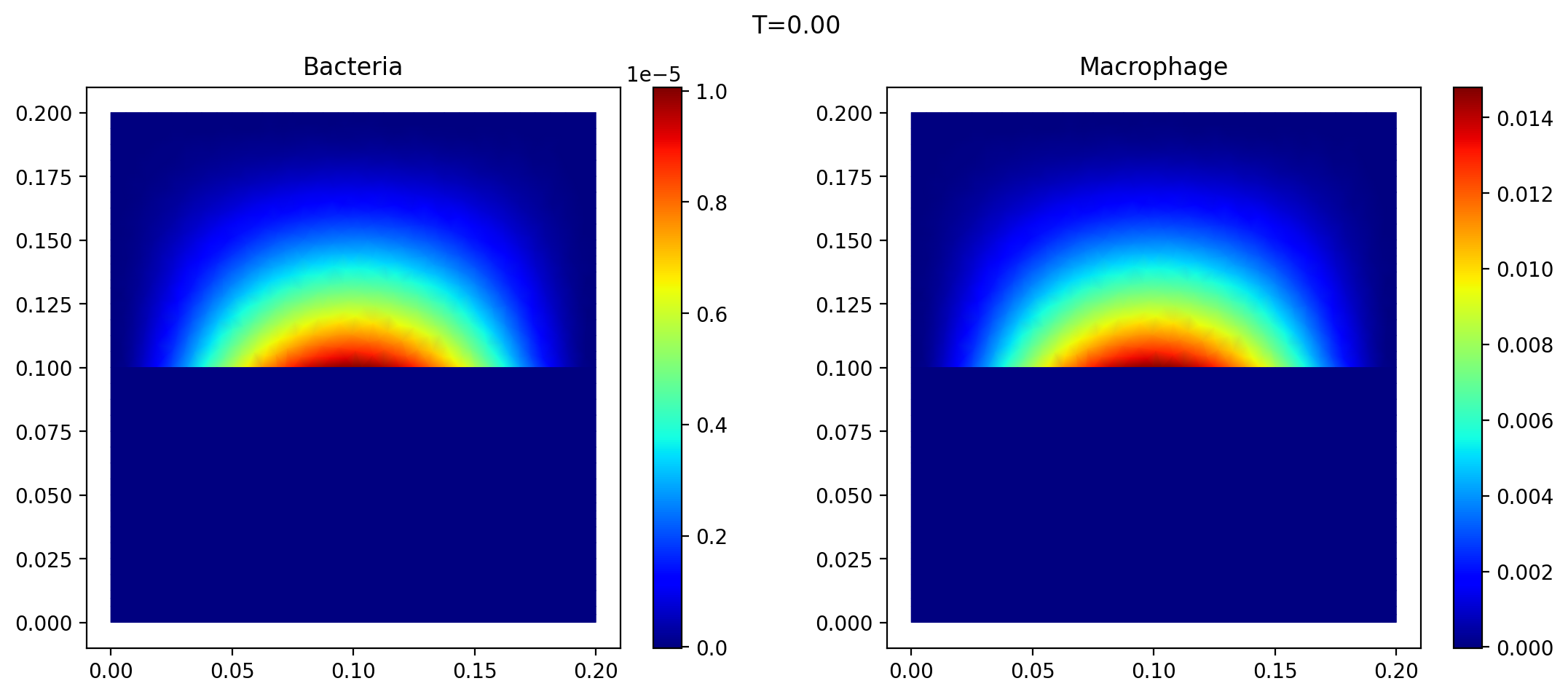}
        \caption{No velocity, $t=0$ day.}
        \label{fig:appendix-set756-no-velocity-t0}
    \end{subfigure}
    \hfill
    \begin{subfigure}{0.30\textwidth}
        \centering
        \includegraphics[width=\linewidth,height=0.76\textheight,keepaspectratio]{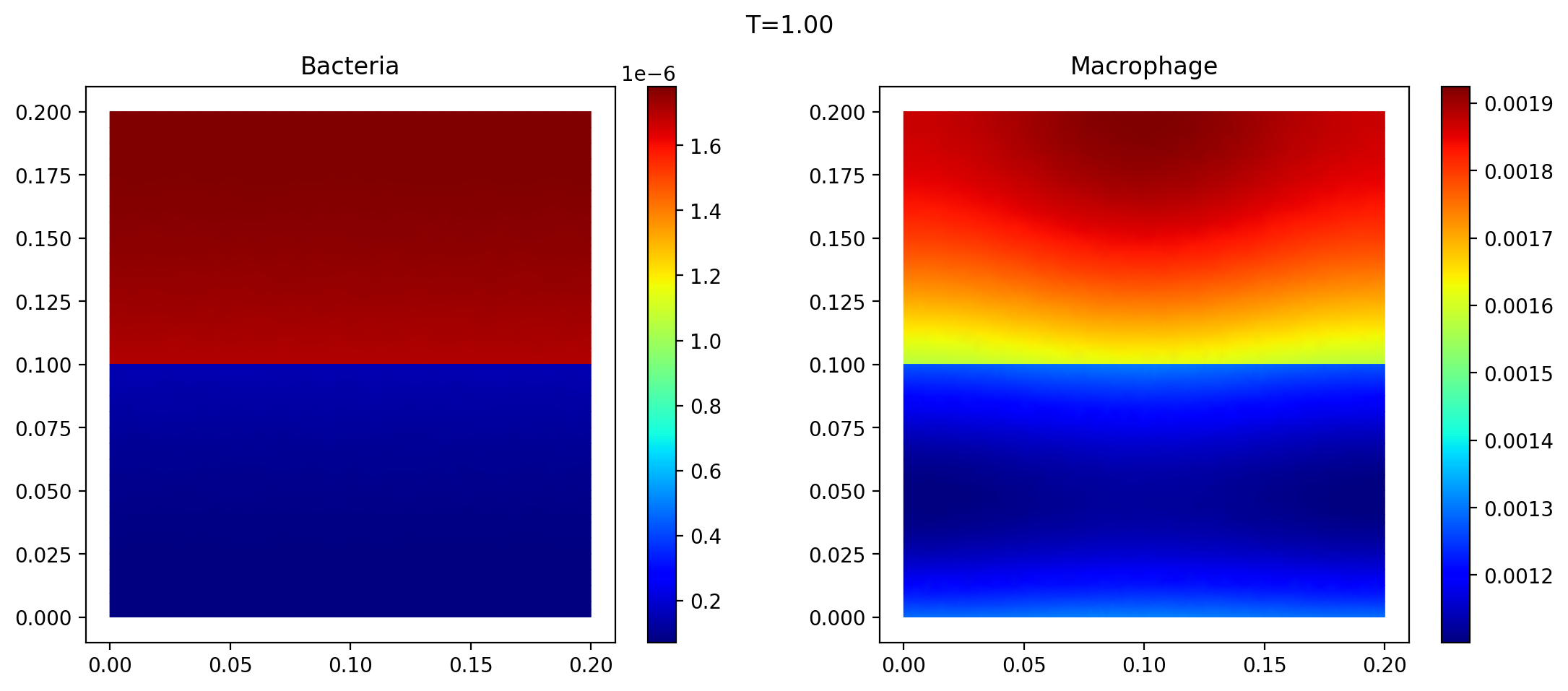}
        \caption{No velocity, $t=1$ day.}
        \label{fig:appendix-set756-no-velocity-t1}
    \end{subfigure}
    \hfill
    \begin{subfigure}{0.30\textwidth}
        \centering
        \includegraphics[width=\linewidth,height=0.76\textheight,keepaspectratio]{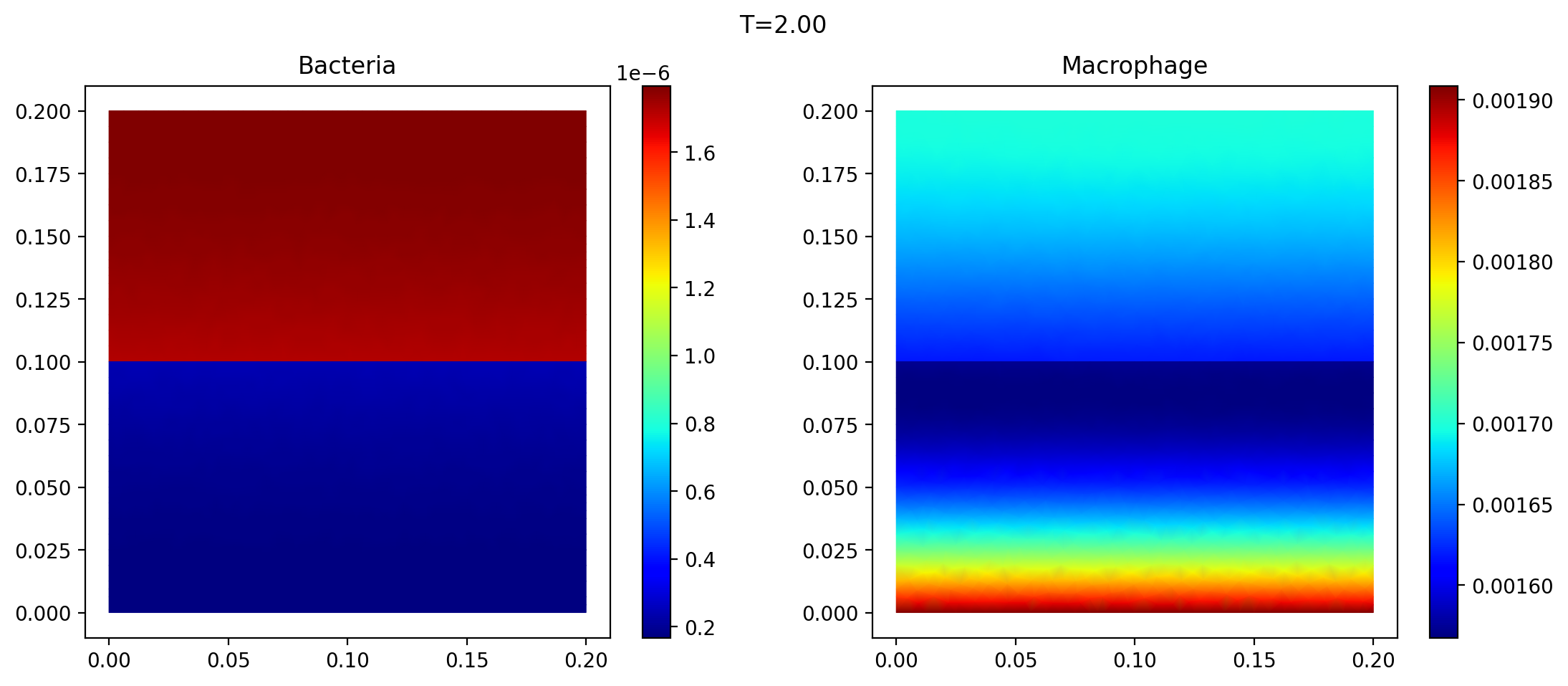}
        \caption{No velocity, $t=2$ day.}
        \label{fig:appendix-set756-no-velocity-t2}
    \end{subfigure}

    \vspace{0.5em}
    \begin{subfigure}{0.30\textwidth}
        \centering
        \includegraphics[width=\linewidth,height=0.76\textheight,keepaspectratio]{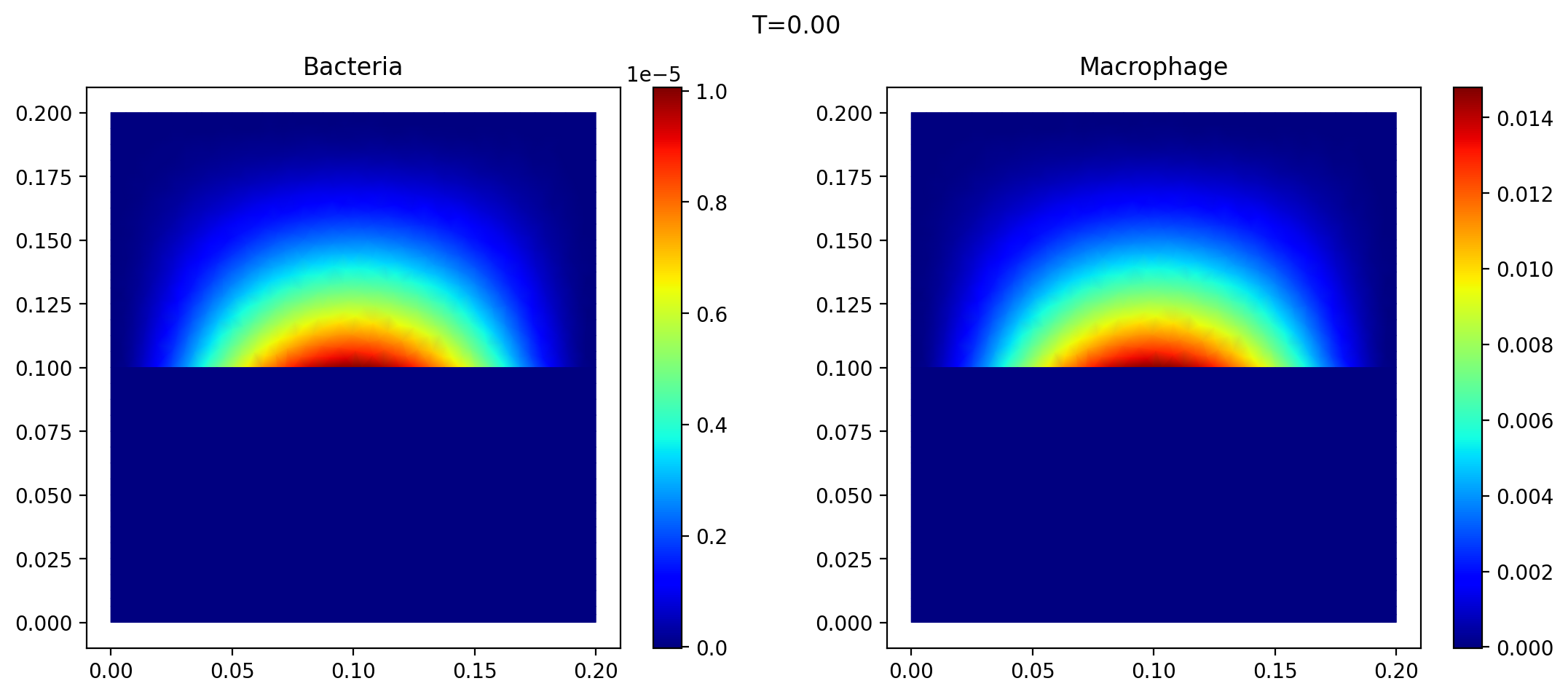}
        \caption{With velocity, $t=0$ day.}
        \label{fig:appendix-set756-velocity-t0}
    \end{subfigure}
    \hfill
    \begin{subfigure}{0.30\textwidth}
        \centering
        \includegraphics[width=\linewidth,height=0.76\textheight,keepaspectratio]{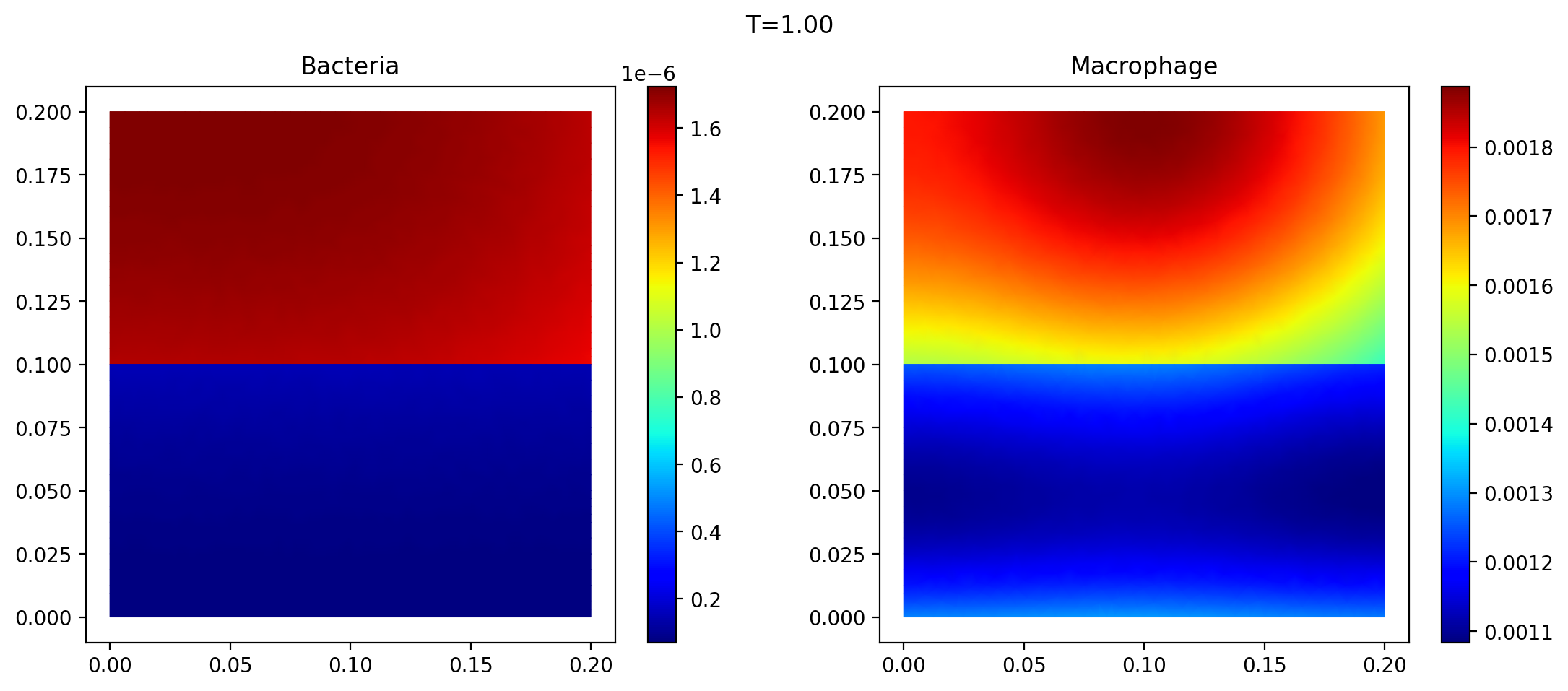}
        \caption{With velocity, $t=1$ day.}
        \label{fig:appendix-set756-velocity-t1}
    \end{subfigure}
    \hfill
    \begin{subfigure}{0.30\textwidth}
        \centering
        \includegraphics[width=\linewidth,height=0.76\textheight,keepaspectratio]{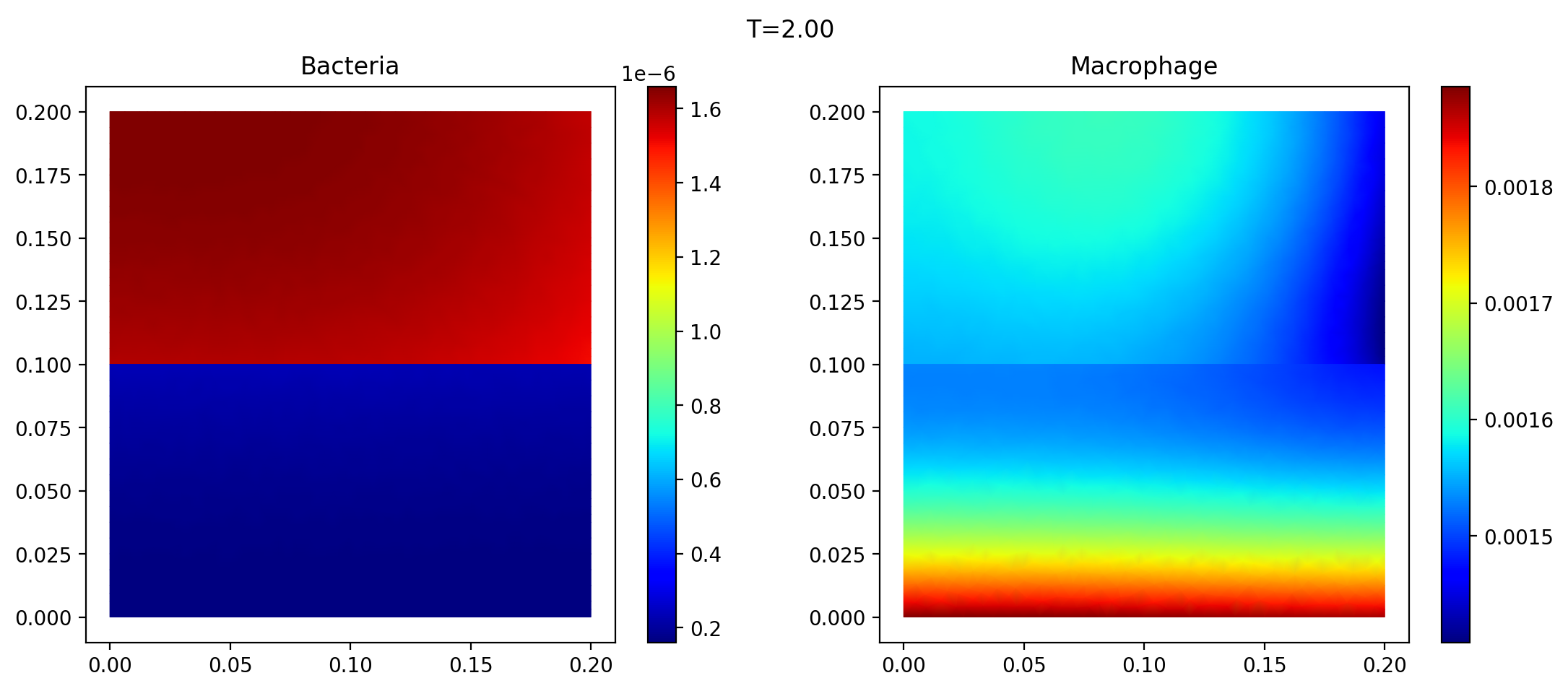}
        \caption{With velocity, $t= 2$ day.}
        \label{fig:appendix-set756-velocity-t2}
    \end{subfigure}

    \caption{
    Evolution of the system for parameter set No.~756, with and without the mucus velocity field.
    (a) Concentration fitting for bacteria and macrophages, with calibration reference data, mucus-region simulation, and tissue-region simulation shown in each plot.
    (b--d) Spatial distributions without mucus velocity field driven at $t=0$, $1$, and $2$ day.
    (e--g) Spatial distributions with the mucus velocity field driven at the same time points.
    Each spatial panel shows bacteria on the left and macrophages on the right.
    }
    \label{fig:appendix-set756-velocity-comparison}
\end{figure}
 \clearpage

\begin{figure}[H]
    \centering
    \begin{subfigure}{0.78\textwidth}
        \centering
        \includegraphics[width=\linewidth,height=0.76\textheight,keepaspectratio]{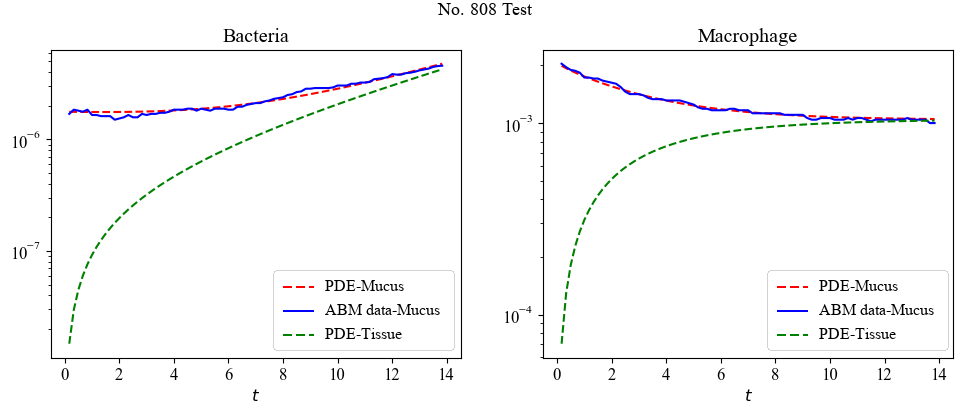}
        \caption{Concentration fitting.}
        \label{fig:appendix-set808-calibration}
    \end{subfigure}

    \vspace{0.5em}
    \begin{subfigure}{0.30\textwidth}
        \centering
        \includegraphics[width=\linewidth,height=0.76\textheight,keepaspectratio]{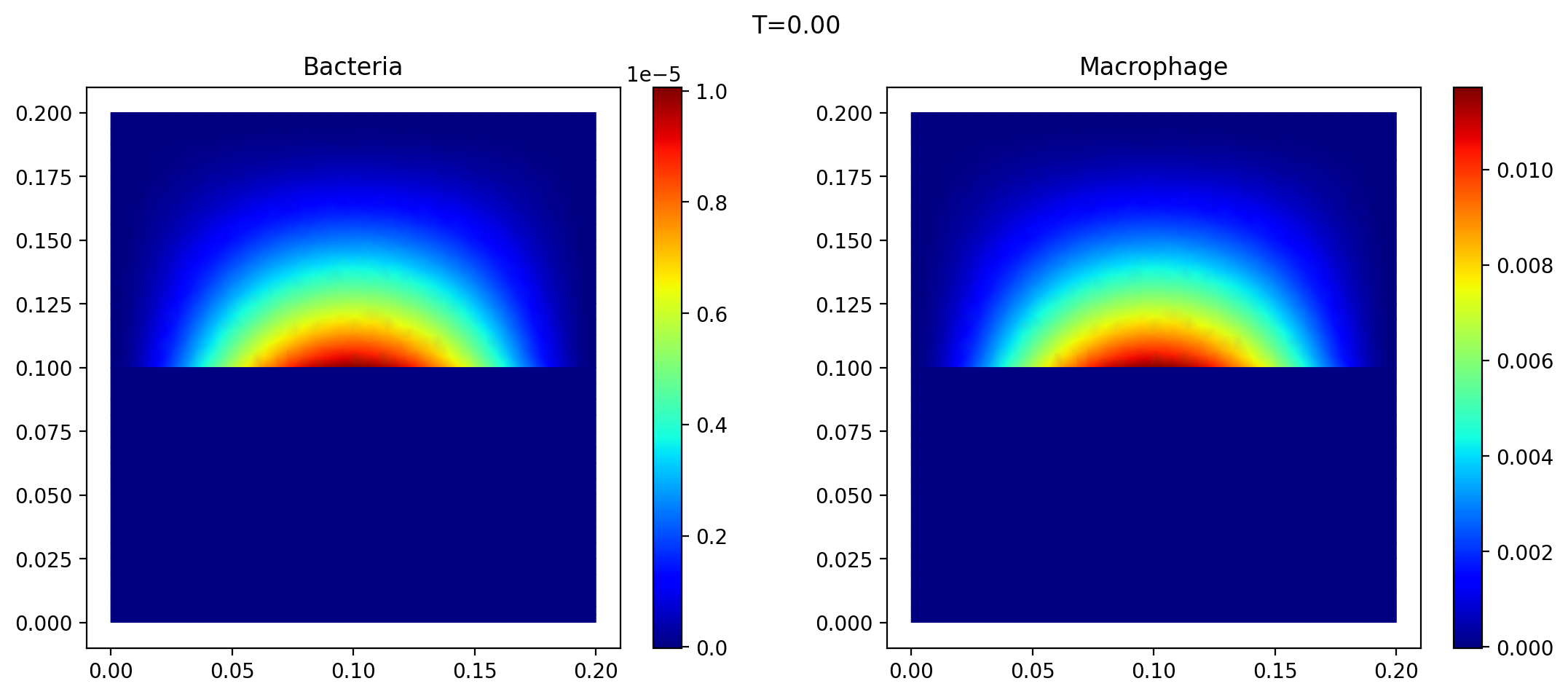}
        \caption{No velocity, $t=0$ day.}
        \label{fig:appendix-set808-no-velocity-t0}
    \end{subfigure}
    \hfill
    \begin{subfigure}{0.30\textwidth}
        \centering
        \includegraphics[width=\linewidth,height=0.76\textheight,keepaspectratio]{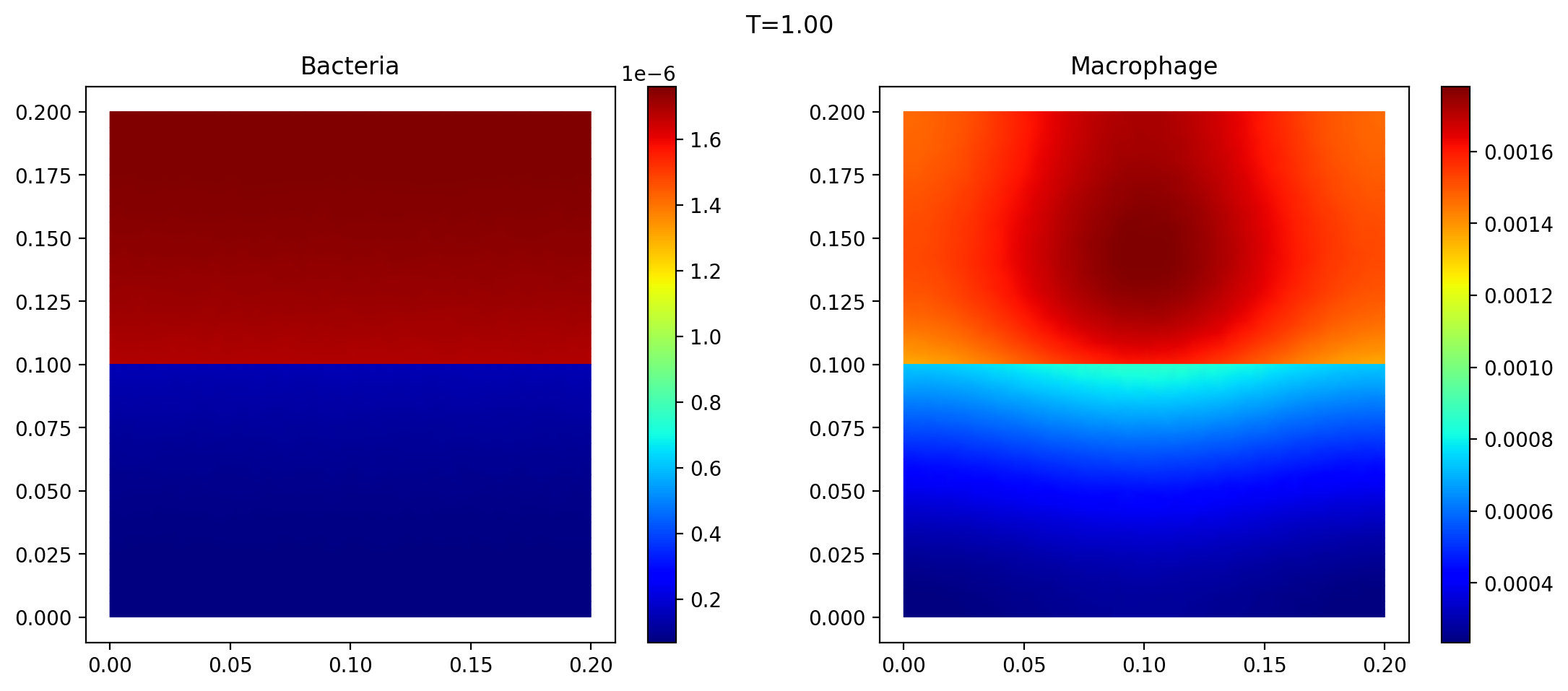}
        \caption{No velocity, $t=1$ day.}
        \label{fig:appendix-set808-no-velocity-t1}
    \end{subfigure}
    \hfill
    \begin{subfigure}{0.30\textwidth}
        \centering
        \includegraphics[width=\linewidth,height=0.76\textheight,keepaspectratio]{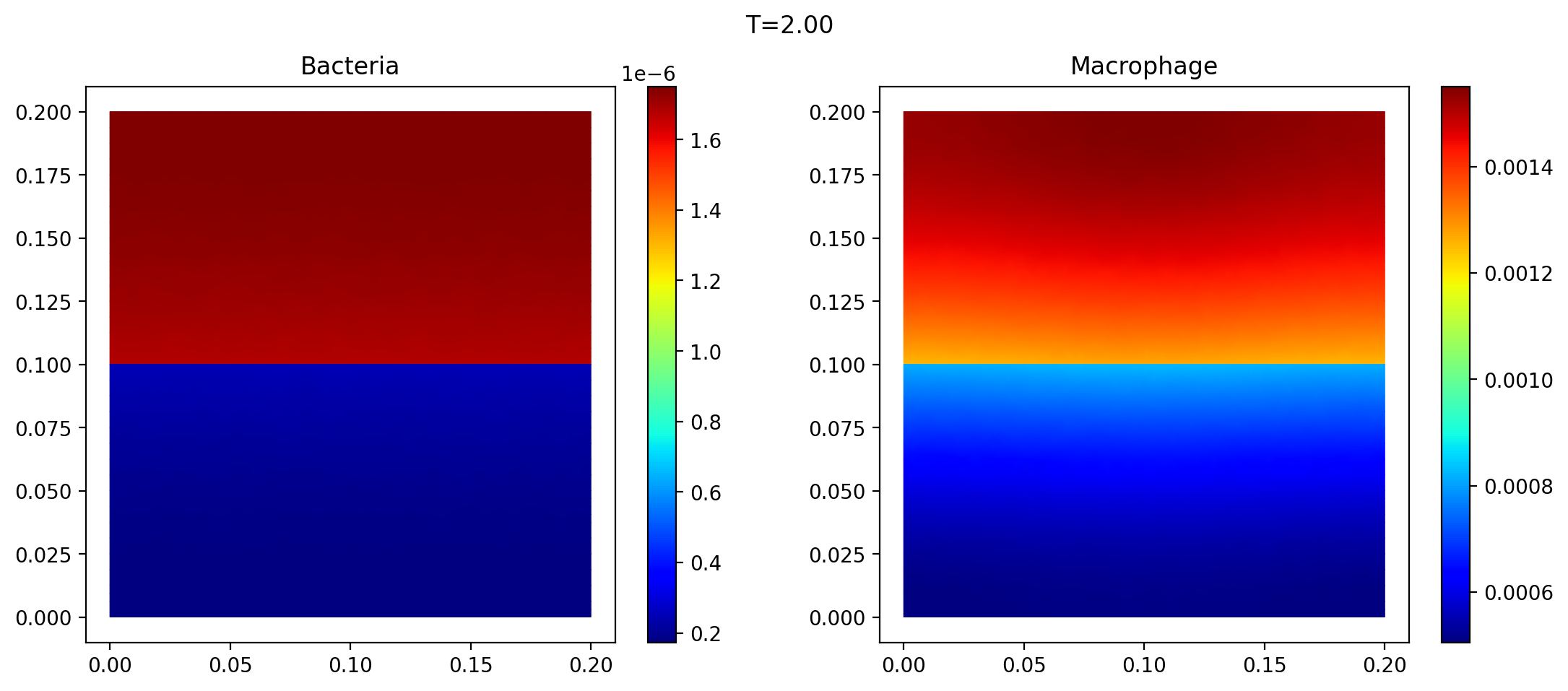}
        \caption{No velocity, $t=2$ day.}
        \label{fig:appendix-set808-no-velocity-t2}
    \end{subfigure}

    \vspace{0.5em}
    \begin{subfigure}{0.30\textwidth}
        \centering
        \includegraphics[width=\linewidth,height=0.76\textheight,keepaspectratio]{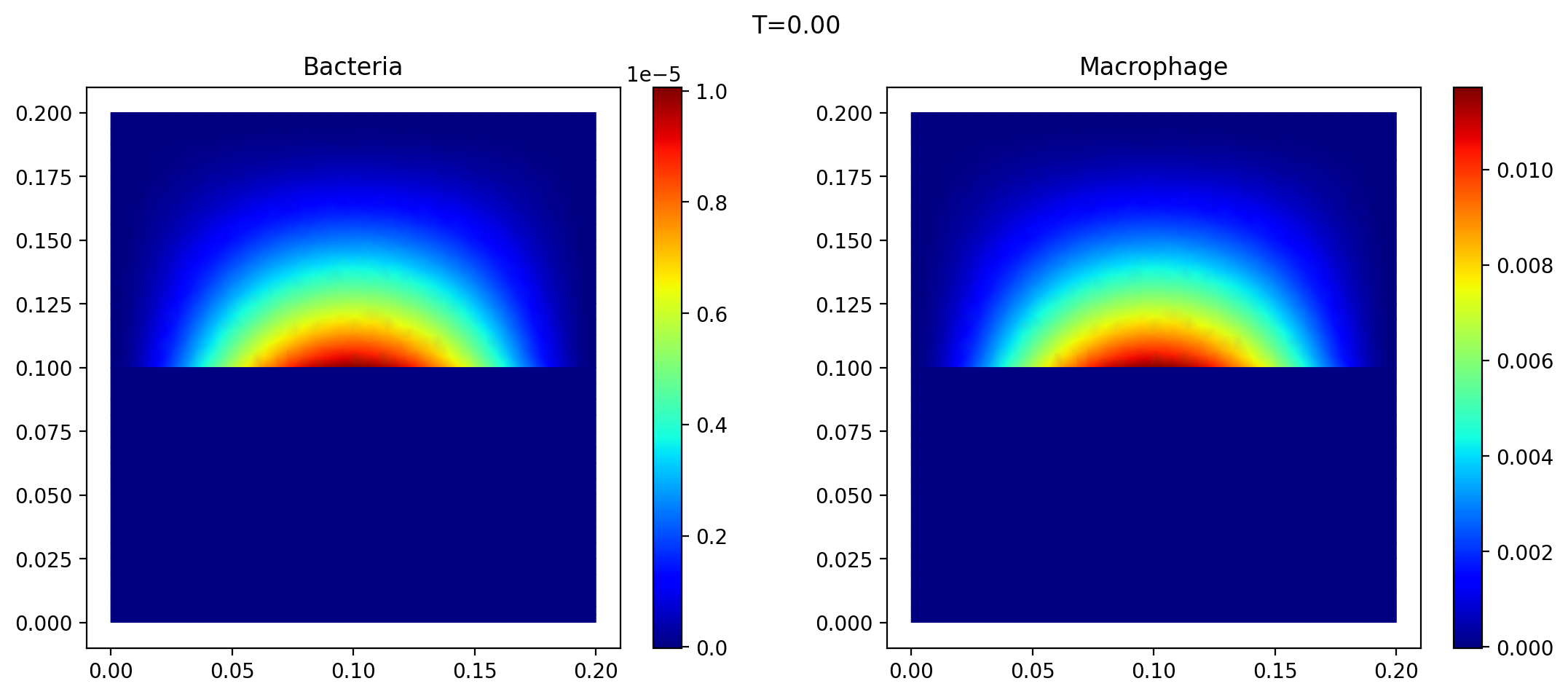}
        \caption{With velocity, $t=0$ day.}
        \label{fig:appendix-set808-velocity-t0}
    \end{subfigure}
    \hfill
    \begin{subfigure}{0.30\textwidth}
        \centering
        \includegraphics[width=\linewidth,height=0.76\textheight,keepaspectratio]{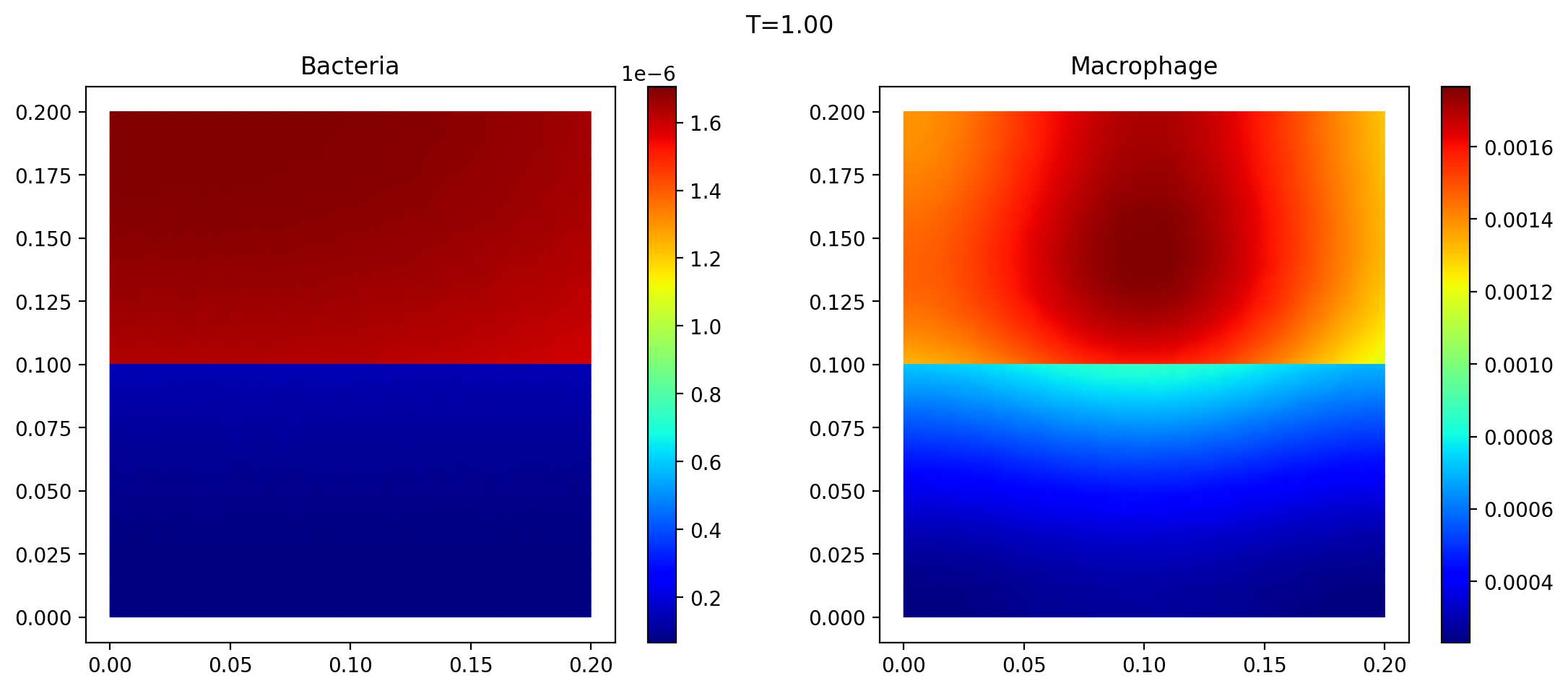}
        \caption{With velocity, $t=1$ day.}
        \label{fig:appendix-set808-velocity-t1}
    \end{subfigure}
    \hfill
    \begin{subfigure}{0.30\textwidth}
        \centering
        \includegraphics[width=\linewidth,height=0.76\textheight,keepaspectratio]{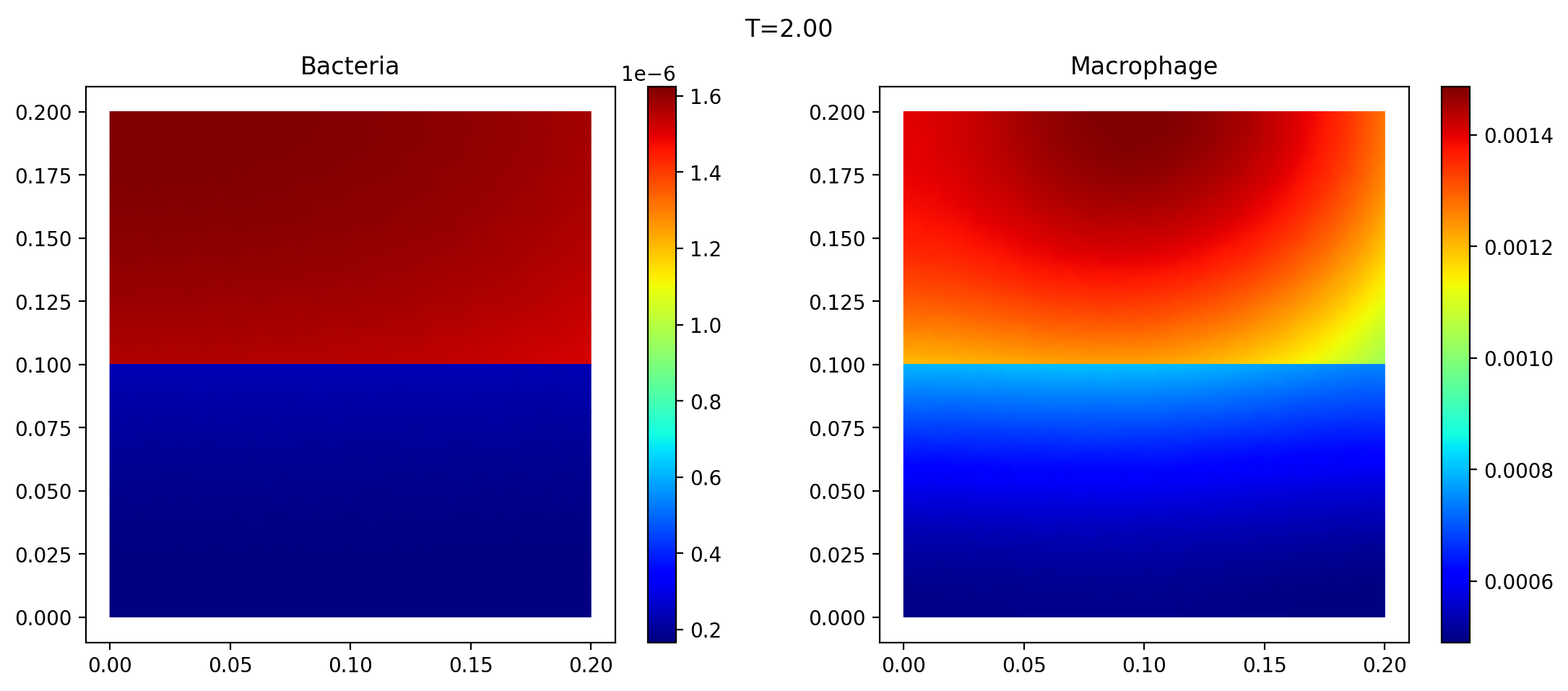}
        \caption{With velocity, $t= 2$ day.}
        \label{fig:appendix-set808-velocity-t2}
    \end{subfigure}

    \caption{
    Evolution of the system for parameter set No.~808, with and without the mucus velocity field.
    (a) Concentration fitting for bacteria and macrophages, with calibration reference data, mucus-region simulation, and tissue-region simulation shown in each plot.
    (b--d) Spatial distributions without mucus velocity field driven at $t=0$, $1$, and $2$ day.
    (e--g) Spatial distributions with the mucus velocity field driven at the same time points.
    Each spatial panel shows bacteria on the left and macrophages on the right.
    }
    \label{fig:appendix-set808-velocity-comparison}
\end{figure}
 \clearpage

\begin{figure}[H]
    \centering
    \begin{subfigure}{0.78\textwidth}
        \centering
        \includegraphics[width=\linewidth,height=0.76\textheight,keepaspectratio]{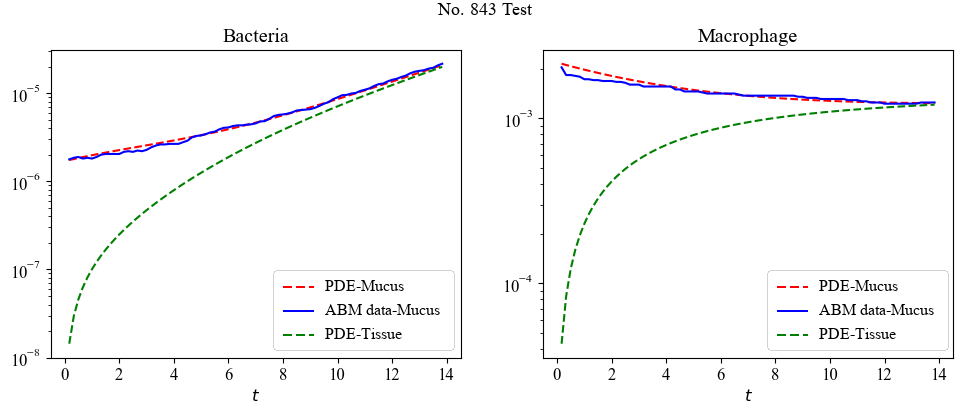}
        \caption{Concentration fitting.}
        \label{fig:appendix-set843-calibration}
    \end{subfigure}

    \vspace{0.5em}
    \begin{subfigure}{0.30\textwidth}
        \centering
        \includegraphics[width=\linewidth,height=0.76\textheight,keepaspectratio]{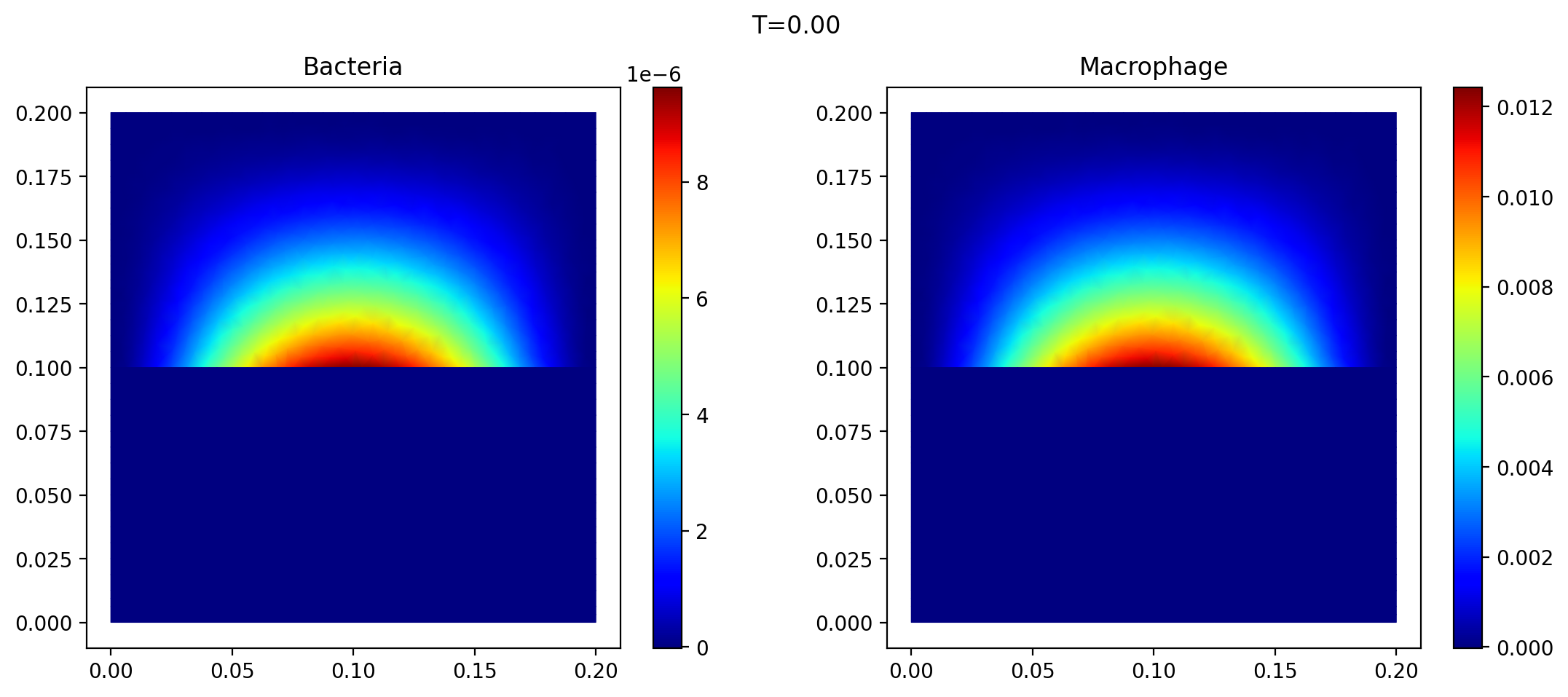}
        \caption{No velocity, $t=0$ day.}
        \label{fig:appendix-set843-no-velocity-t0}
    \end{subfigure}
    \hfill
    \begin{subfigure}{0.30\textwidth}
        \centering
        \includegraphics[width=\linewidth,height=0.76\textheight,keepaspectratio]{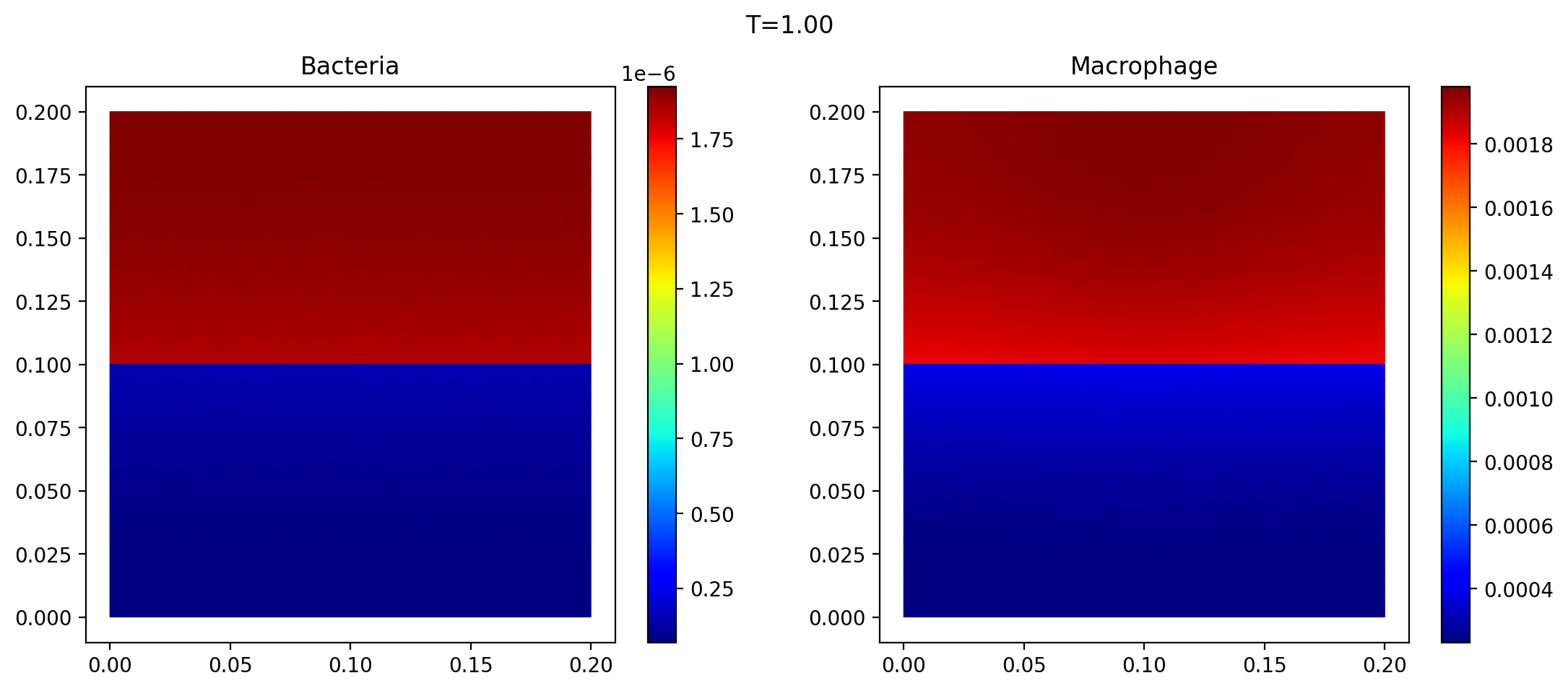}
        \caption{No velocity, $t=1$ day.}
        \label{fig:appendix-set843-no-velocity-t1}
    \end{subfigure}
    \hfill
    \begin{subfigure}{0.30\textwidth}
        \centering
        \includegraphics[width=\linewidth,height=0.76\textheight,keepaspectratio]{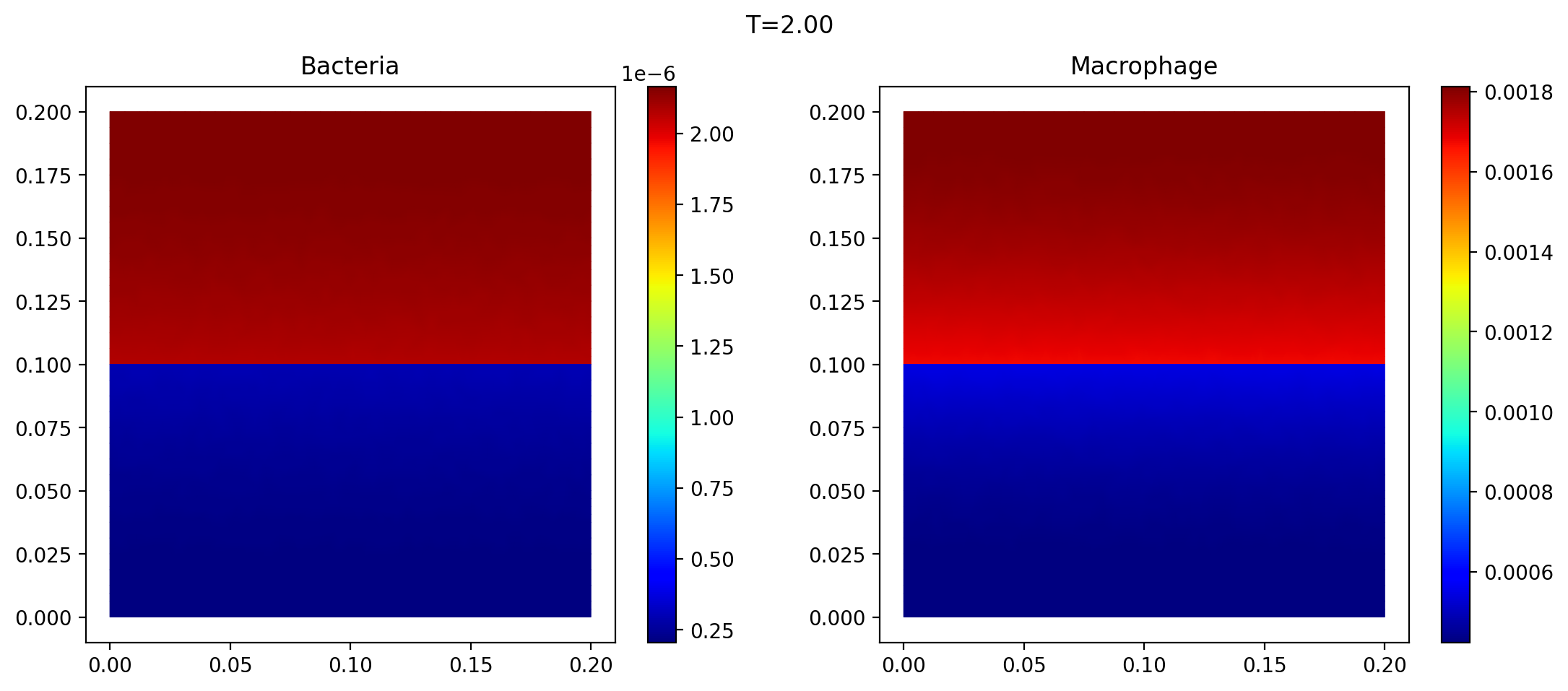}
        \caption{No velocity, $t=2$ day.}
        \label{fig:appendix-set843-no-velocity-t2}
    \end{subfigure}

    \vspace{0.5em}
    \begin{subfigure}{0.30\textwidth}
        \centering
        \includegraphics[width=\linewidth,height=0.76\textheight,keepaspectratio]{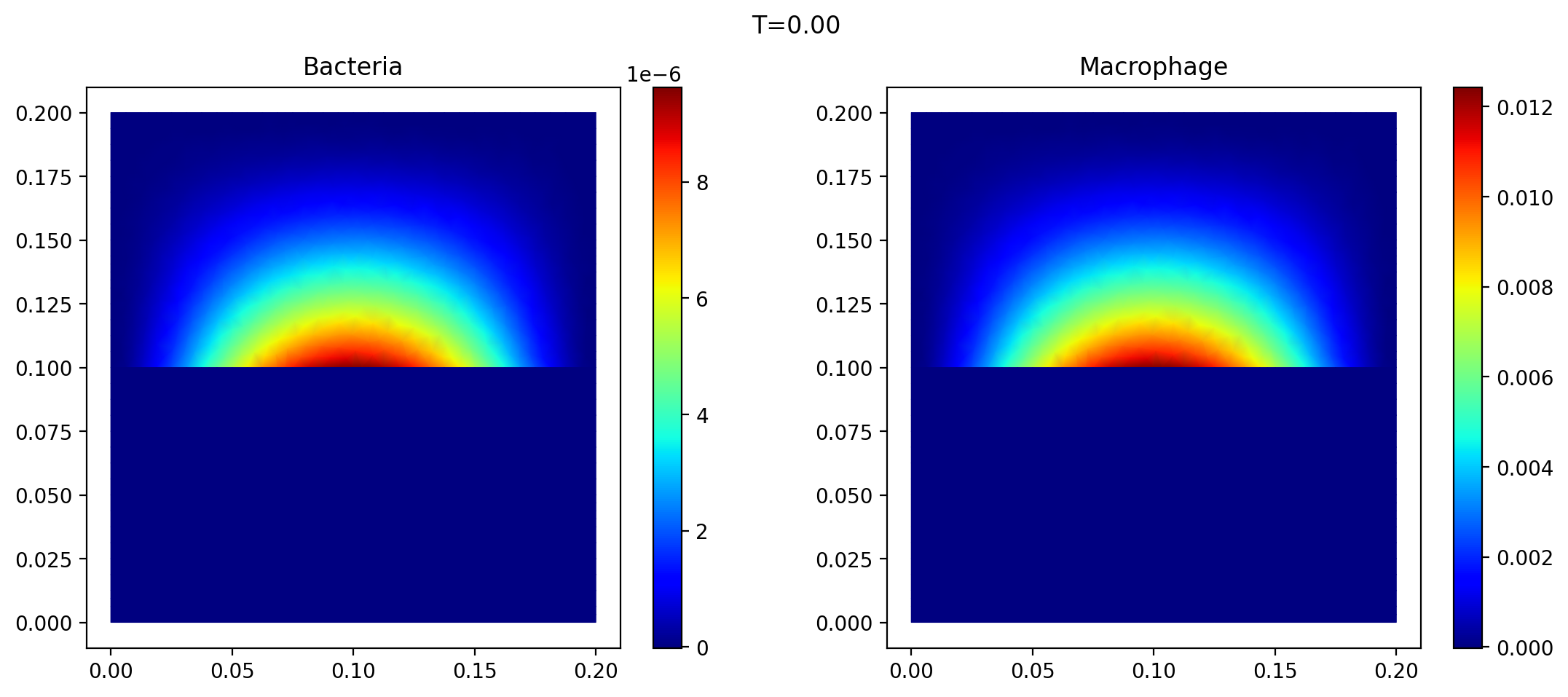}
        \caption{With velocity, $t=0$ day.}
        \label{fig:appendix-set843-velocity-t0}
    \end{subfigure}
    \hfill
    \begin{subfigure}{0.30\textwidth}
        \centering
        \includegraphics[width=\linewidth,height=0.76\textheight,keepaspectratio]{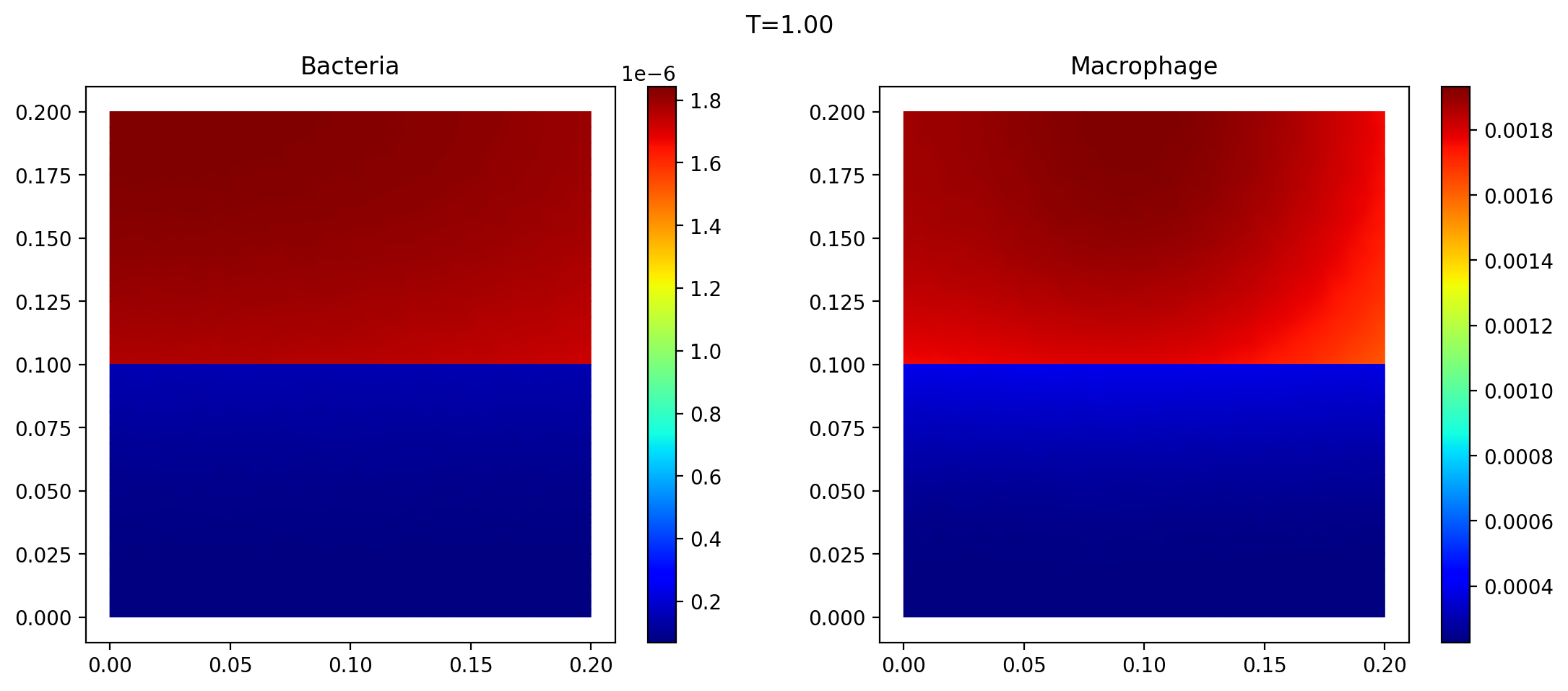}
        \caption{With velocity, $t=1$ day.}
        \label{fig:appendix-set843-velocity-t1}
    \end{subfigure}
    \hfill
    \begin{subfigure}{0.30\textwidth}
        \centering
        \includegraphics[width=\linewidth,height=0.76\textheight,keepaspectratio]{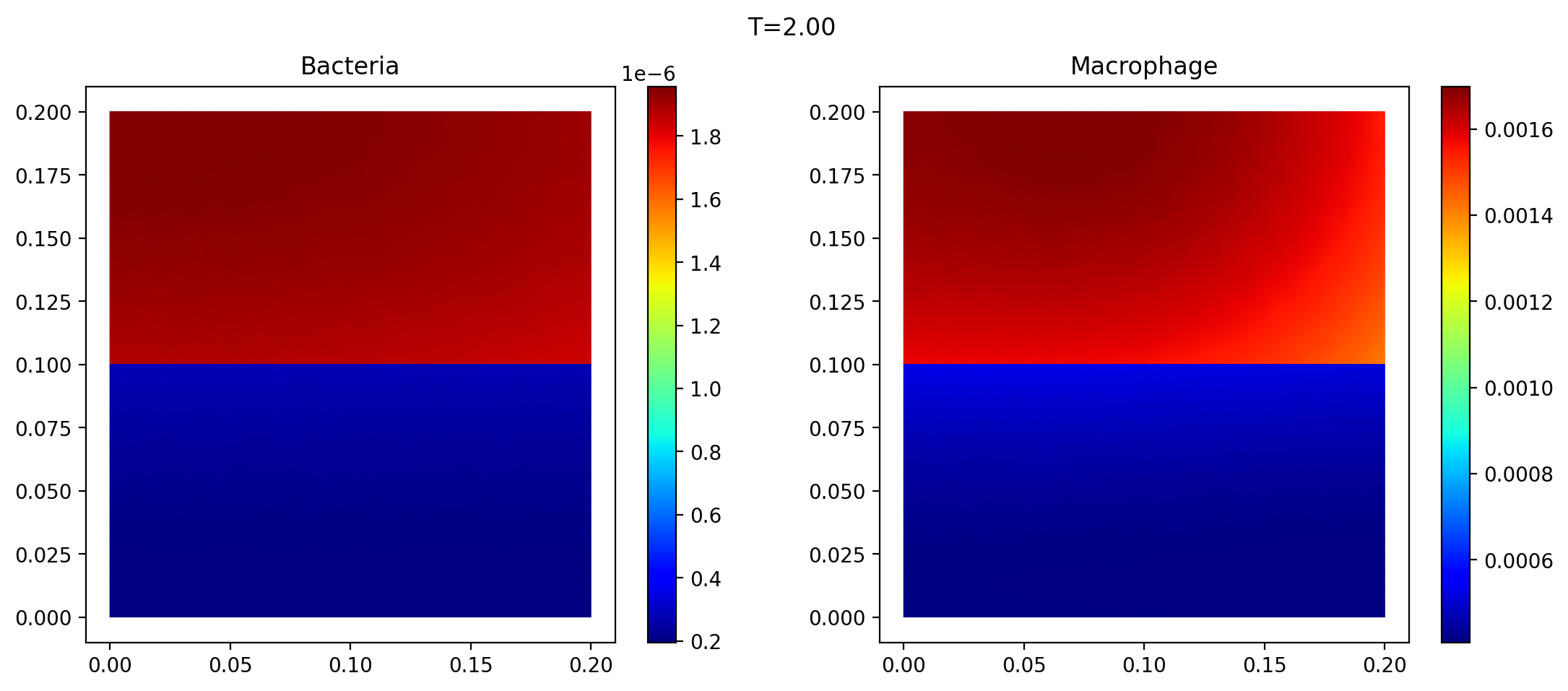}
        \caption{With velocity, $t= 2$ day.}
        \label{fig:appendix-set843-velocity-t2}
    \end{subfigure}

    \caption{
    Evolution of the system for parameter set No.~843, with and without the mucus velocity field.
    (a) Concentration fitting for bacteria and macrophages, with calibration reference data, mucus-region simulation, and tissue-region simulation shown in each plot.
    (b--d) Spatial distributions without mucus velocity field driven at $t=0$, $1$, and $2$ day.
    (e--g) Spatial distributions with the mucus velocity field driven at the same time points.
    Each spatial panel shows bacteria on the left and macrophages on the right.
    }
    \label{fig:appendix-set843-velocity-comparison}
\end{figure}
\clearpage

\section{Full Sensitivity Analysis Results}
\begin{figure}[H]
    \centering
    \includegraphics[width=0.7\linewidth,height=0.76\textheight,keepaspectratio]{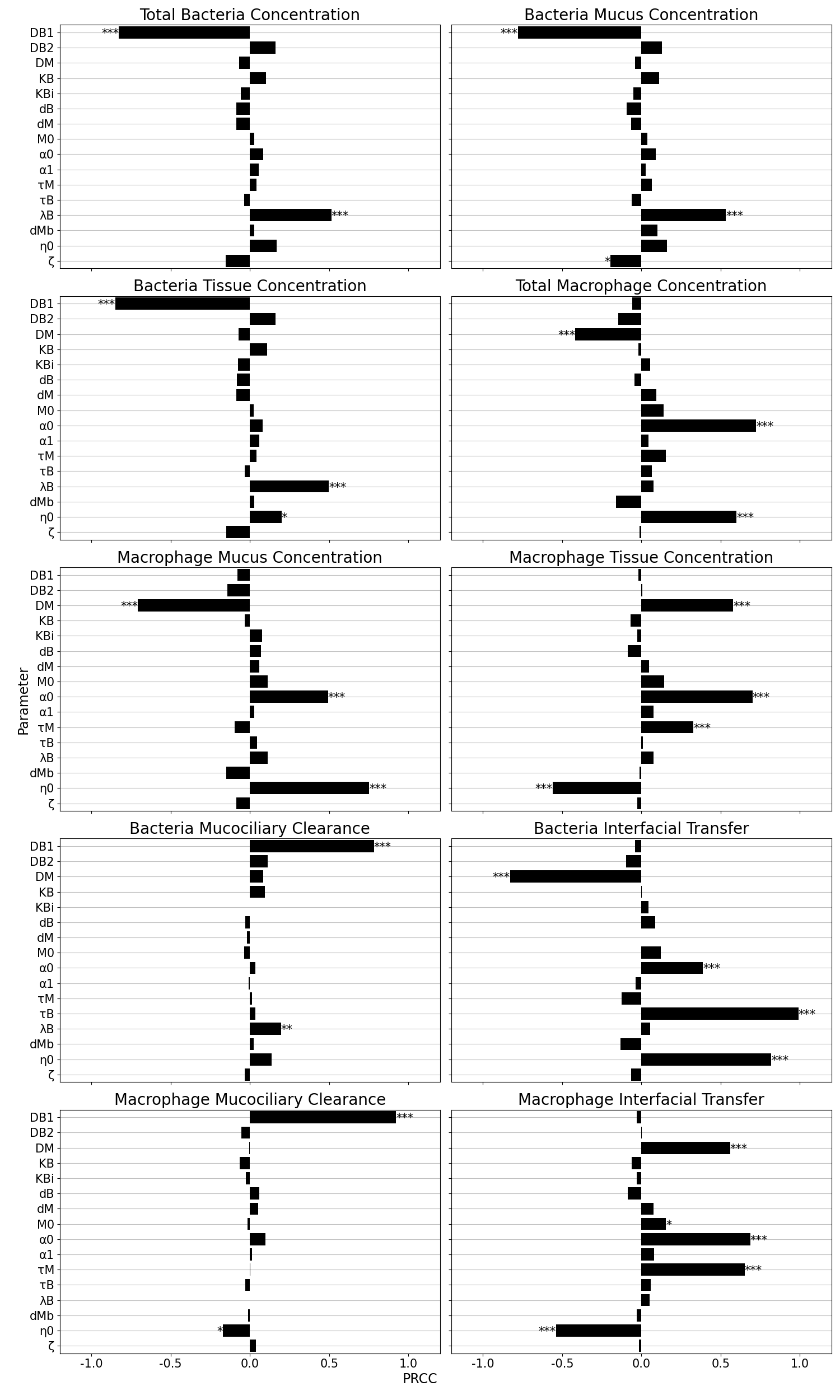}
    \caption{Partial rank correlation coefficients (PRCCs) between model parameters
    and all simulation outputs.}
    \label{fig:supp_full_PRCC}
\end{figure}
\clearpage

\section{Function heatmap of antibiotic bacteria killing rate and viscosity \texorpdfstring{$\eta$}{eta}}
\begin{figure}[H]
	\centering
	\includegraphics[width=0.48\textwidth,height=0.76\textheight,keepaspectratio]{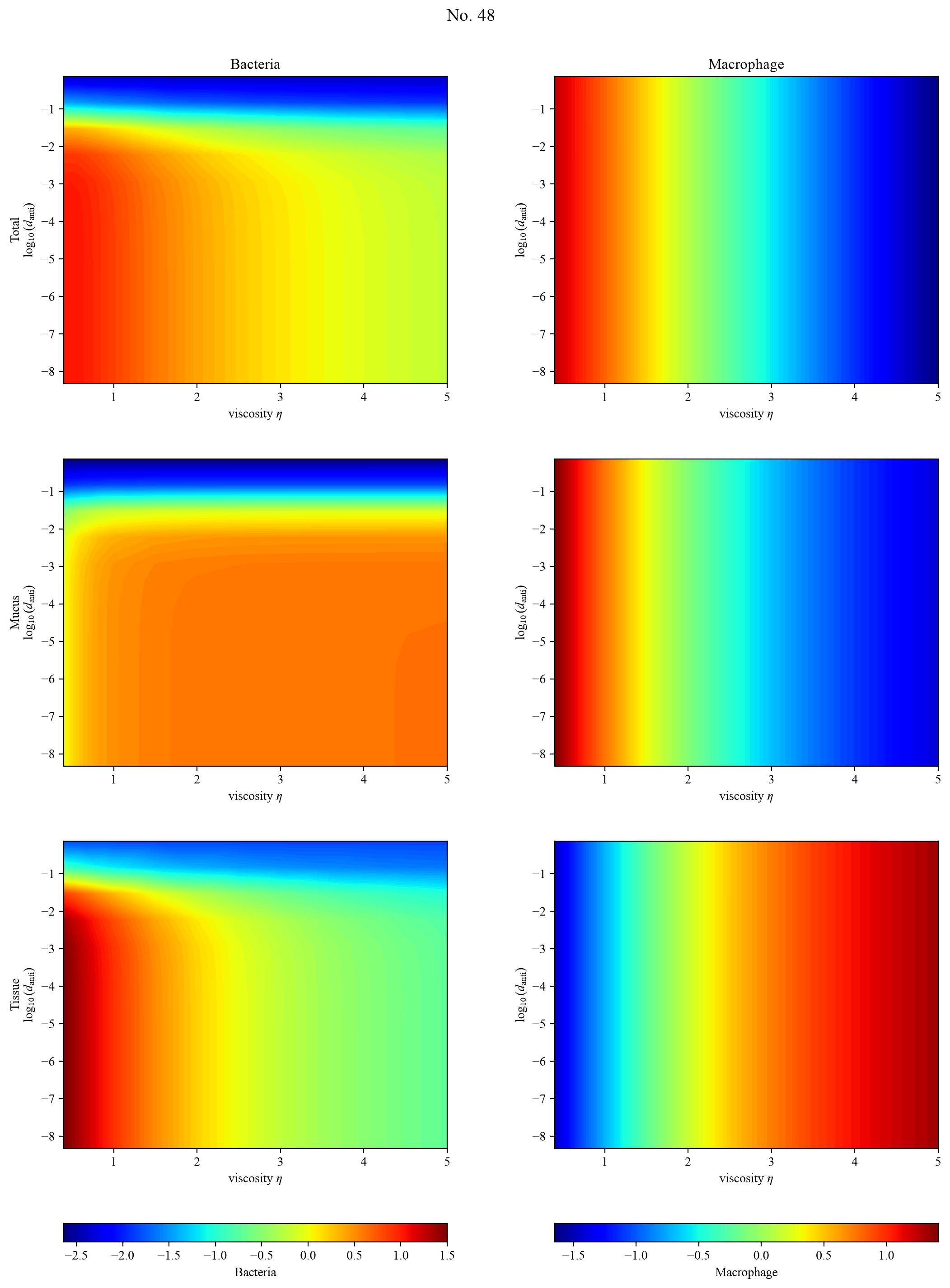}\hfill
	\includegraphics[width=0.48\textwidth,height=0.76\textheight,keepaspectratio]{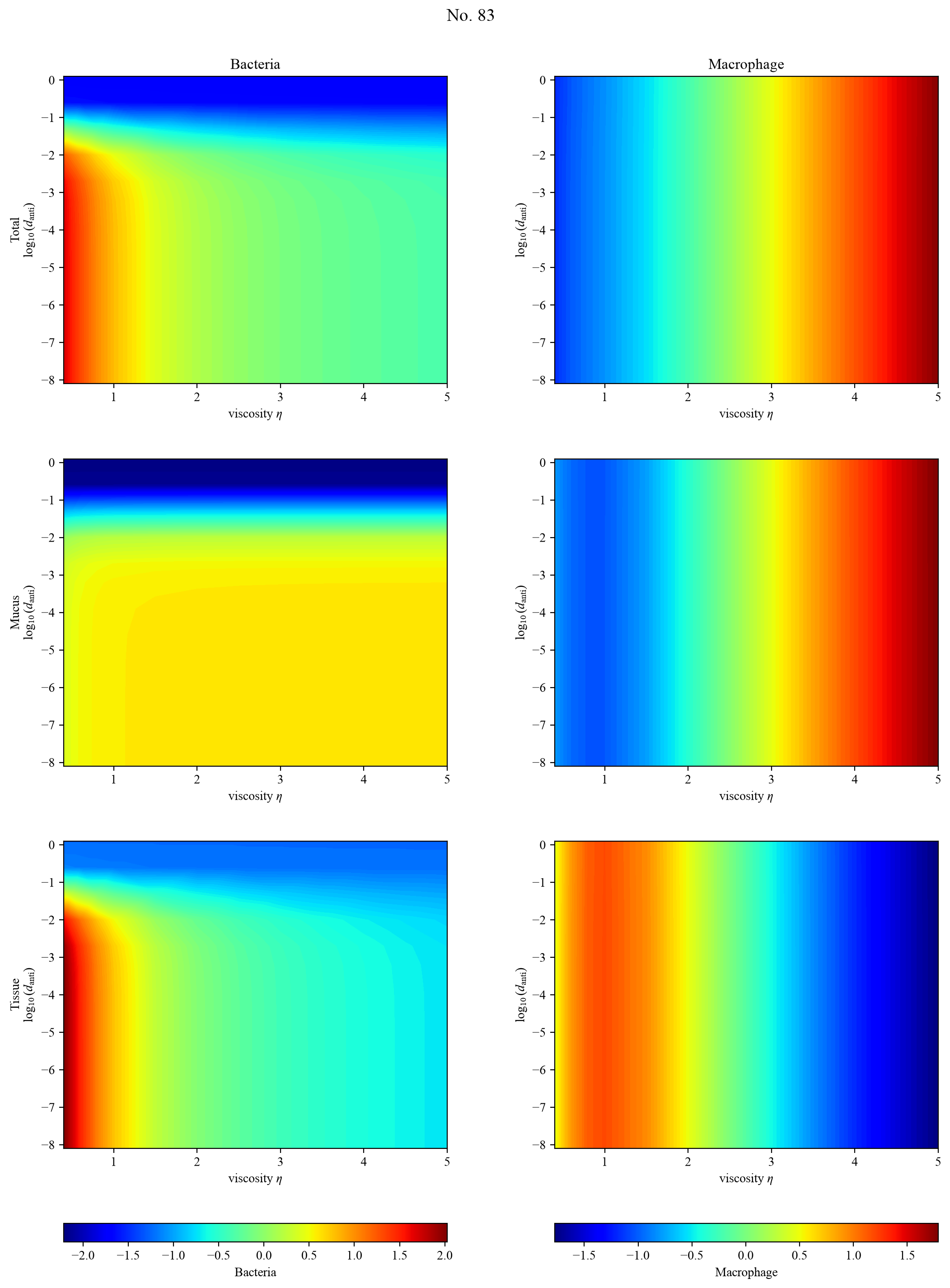}
	\caption{
		Standardized treatment-response heatmaps for parameter sets No.~48 (left) and No.~83 (right). Rows show total, mucus, and tissue responses, and columns show bacteria and macrophages. Mucus and tissue quantities are standardized as $(x-\mu)/\sigma$ using global statistics across all parameter sets; the total response is the sum of the standardized mucus and tissue responses.
	}
	\label{fig:ntm_heatmap_supp_48_83}
\end{figure}

\begin{figure}[H]
	\centering
	\includegraphics[width=0.48\textwidth,height=0.76\textheight,keepaspectratio]{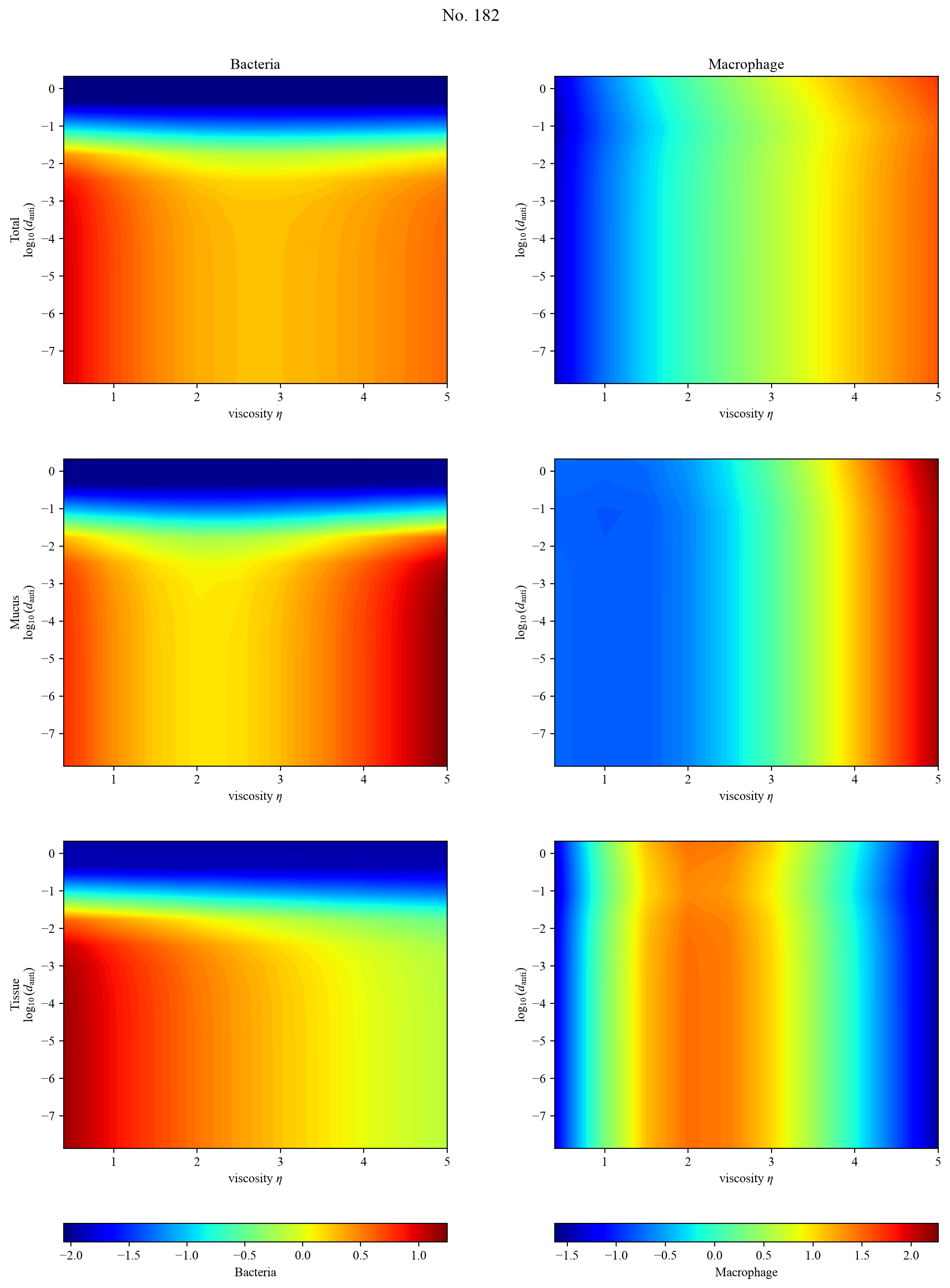}\hfill
	\includegraphics[width=0.48\textwidth,height=0.76\textheight,keepaspectratio]{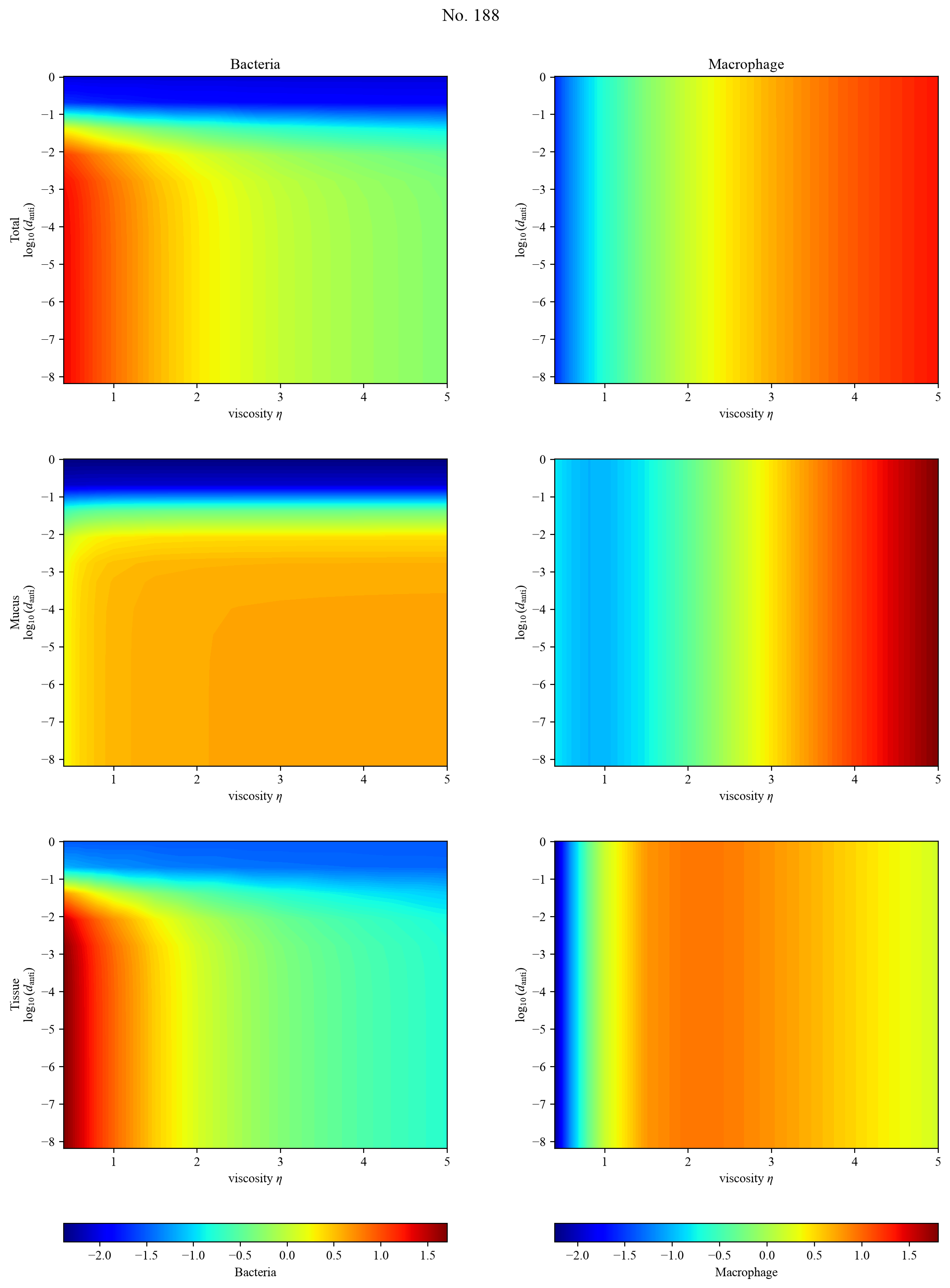}
	\caption{
		Standardized treatment-response heatmaps for parameter sets No.~182 (left) and No.~188 (right).
		Panel layout, normalization, total-response definition, and color-scale conventions are as described in Fig.~\ref{fig:ntm_heatmap_supp_48_83}.
	}
	\label{fig:ntm_heatmap_supp_182_188}
\end{figure}

\begin{figure}[H]
	\centering
	\includegraphics[width=0.48\textwidth,height=0.76\textheight,keepaspectratio]{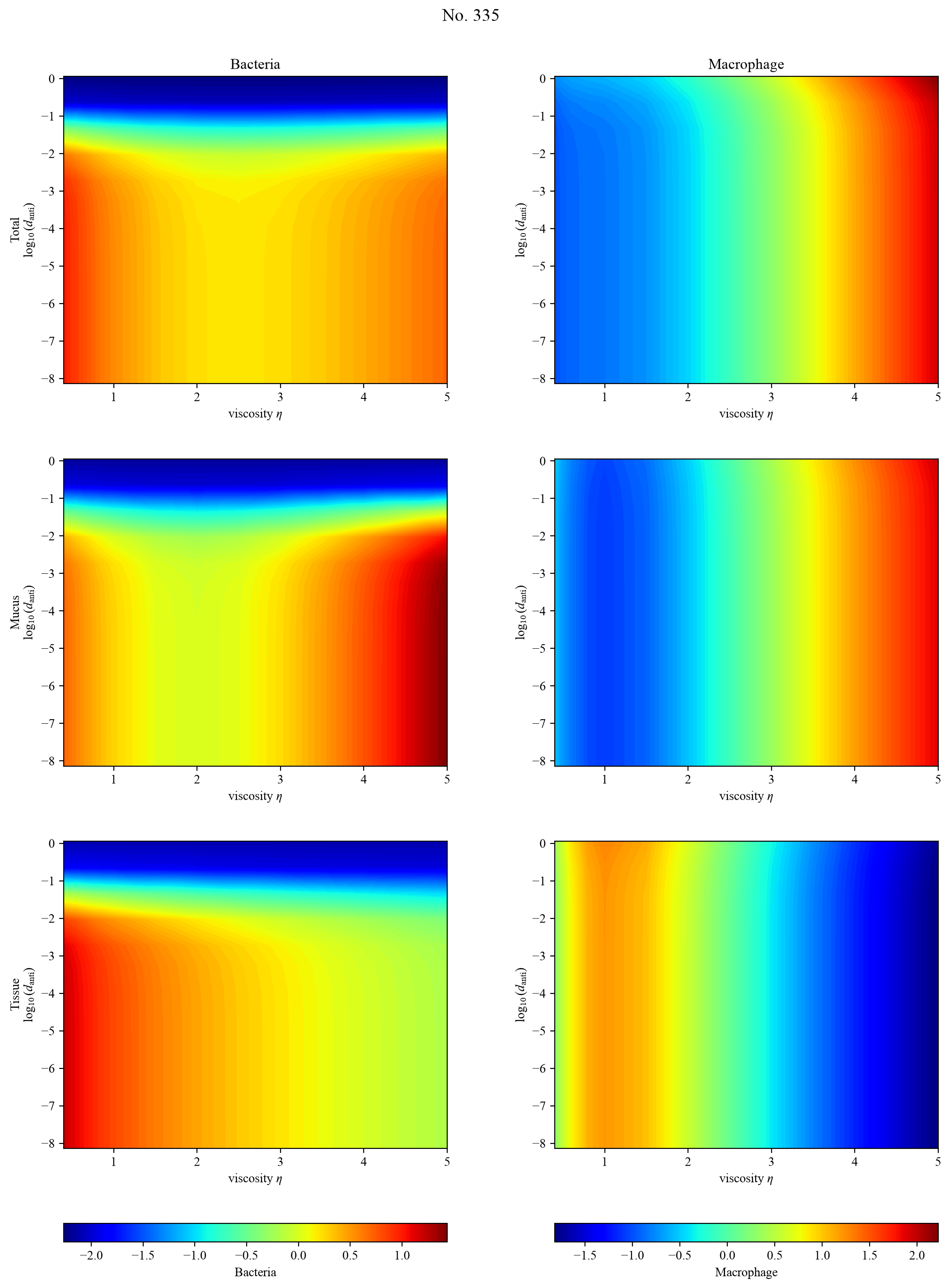}\hfill
	\includegraphics[width=0.48\textwidth,height=0.76\textheight,keepaspectratio]{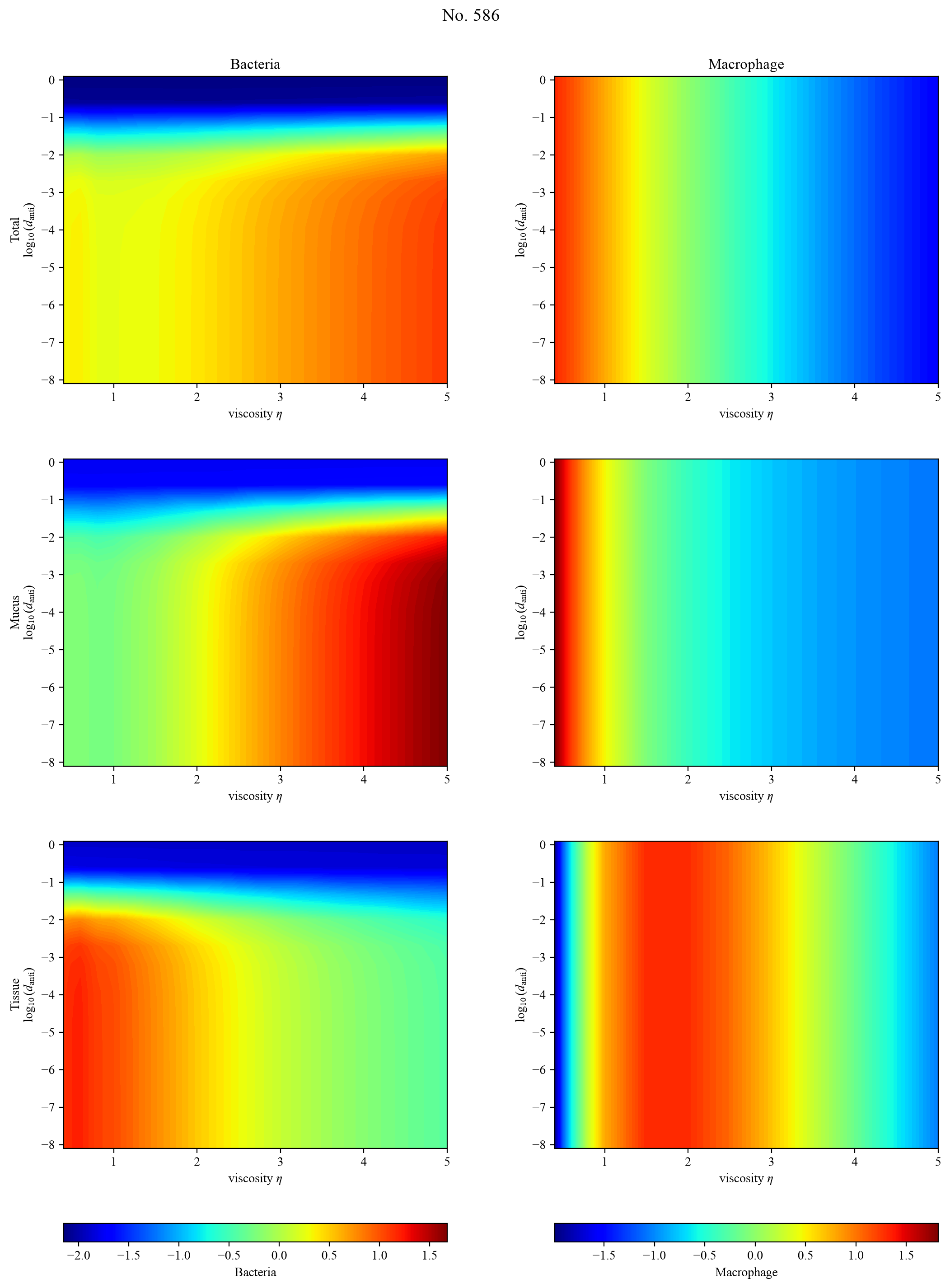}
	\caption{
		Standardized treatment-response heatmaps for parameter sets No.~335 (left) and No.~586 (right).
		Panel layout, normalization, total-response definition, and color-scale conventions are as described in Fig.~\ref{fig:ntm_heatmap_supp_48_83}.
	}
	\label{fig:ntm_heatmap_supp_335_586}
\end{figure}

\begin{figure}[H]
	\centering
	\includegraphics[width=0.48\textwidth,height=0.76\textheight,keepaspectratio]{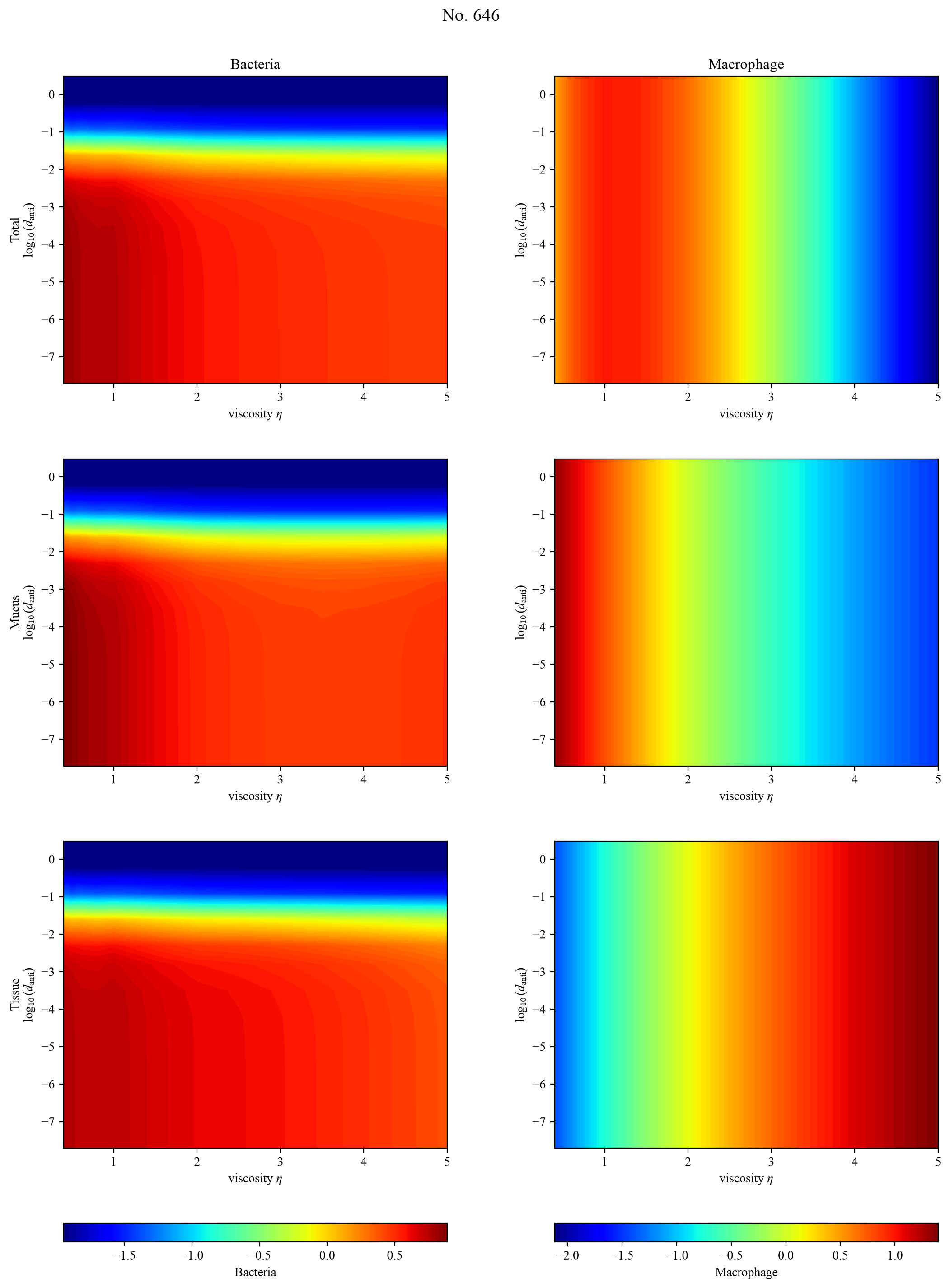}\hfill
	\includegraphics[width=0.48\textwidth,height=0.76\textheight,keepaspectratio]{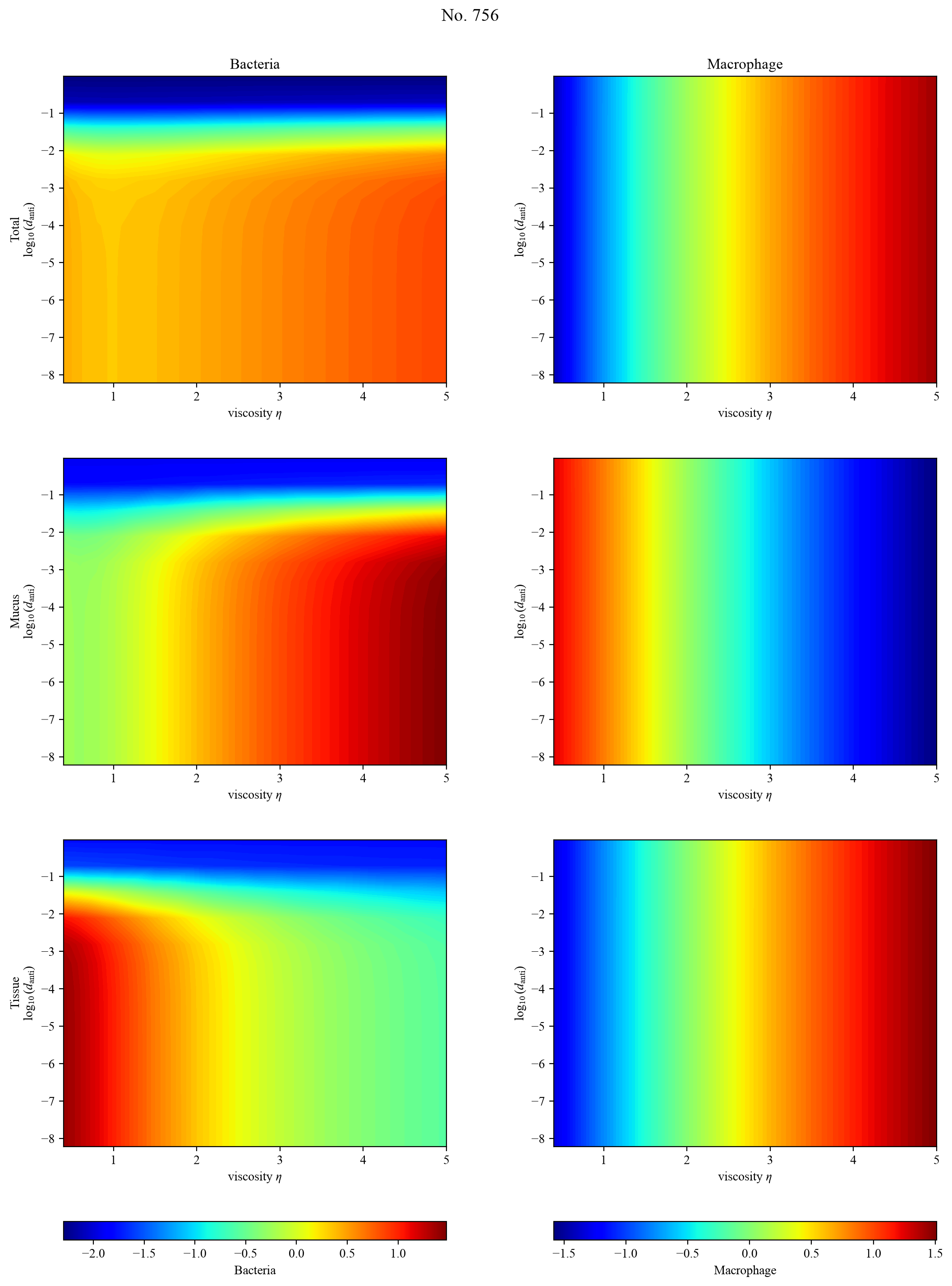}
	\caption{
		Standardized treatment-response heatmaps for parameter sets No.~646 (left) and No.~756 (right).
		Panel layout, normalization, total-response definition, and color-scale conventions are as described in Fig.~\ref{fig:ntm_heatmap_supp_48_83}.
	}
	\label{fig:ntm_heatmap_supp_646_756}
\end{figure}

\begin{figure}[H]
	\centering
	\includegraphics[width=0.48\textwidth,height=0.76\textheight,keepaspectratio]{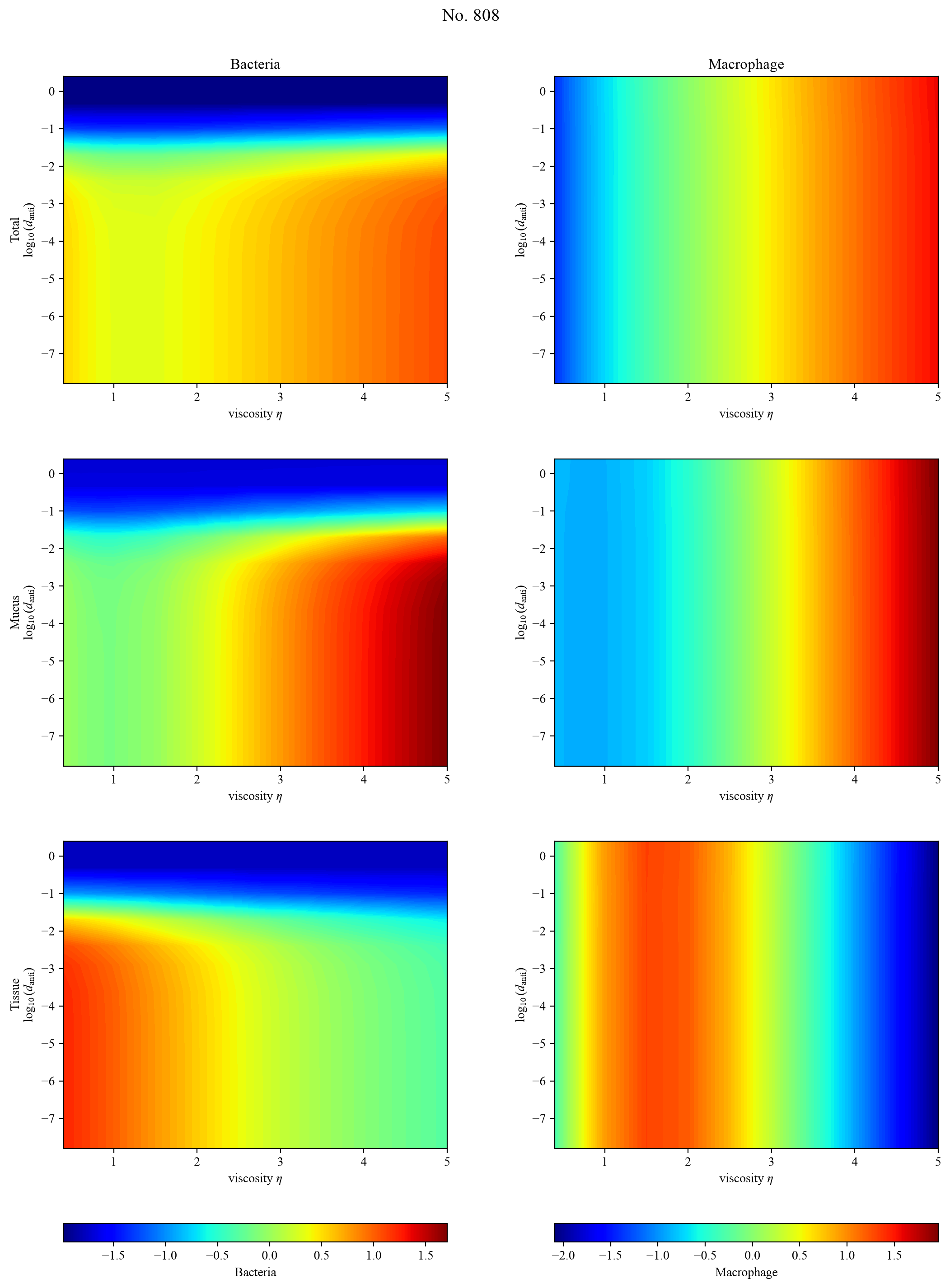}\hfill
	\includegraphics[width=0.48\textwidth,height=0.76\textheight,keepaspectratio]{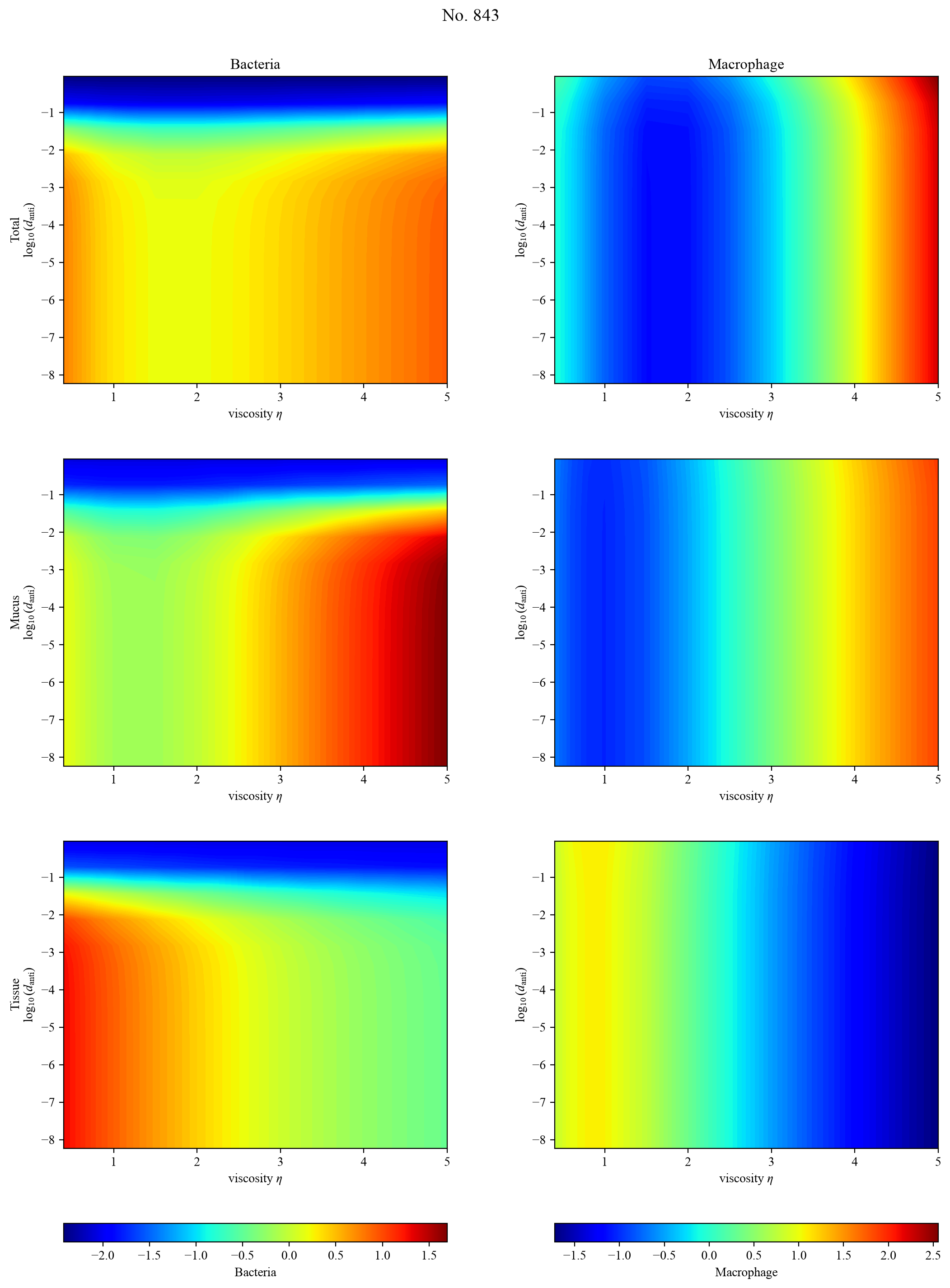}
	\caption{
		Standardized treatment-response heatmaps for parameter sets No.~808 (left) and No.~843 (right).
		Panel layout, normalization, total-response definition, and color-scale conventions are as described in Fig.~\ref{fig:ntm_heatmap_supp_48_83}.
	}
	\label{fig:ntm_heatmap_supp_808_843}
\end{figure}

\clearpage
\bibliographystyle{unsrtnat}
\bibliography{ref}
\end{document}